\documentclass{IEEEtran}

\usepackage[T1]{fontenc} 
\usepackage{amsmath,amssymb,amsthm}
\usepackage[cmintegrals]{newtxmath}
\usepackage{bm} 
\usepackage{float}
\usepackage{cite}
\usepackage{graphicx,subfigure,epstopdf}
\usepackage{multirow}
\usepackage{url}
\usepackage{algorithm}
\usepackage{algpseudocode,color,xcolor}

\newtheorem{thm}{Theorem}
\newtheorem{lem}{Lemma}

\newtheorem{prop}{Proposition}

\newcommand{\R}{\mathbb{R}}
\newcommand{\C}{\mathbb{C}}

\renewcommand{\Re}[1]{\operatorname{Re}\left\{#1\right\}}

\newcommand{\E}{\mathbb{E}}

\renewcommand{\d}[1]{d#1}

\newcommand{\vct}[1]{\boldsymbol{#1}}
\newcommand{\mtx}[1]{\boldsymbol{#1}}

\newcommand{\<}{\langle}
\renewcommand{\>}{\rangle}

\newcommand{\vctr}{\operatorname{vec}}

\newcommand{\set}[1]{\mathcal{#1}}

\DeclareMathOperator*{\minimize}{\text{minimize}}

\newcommand{\vc}{\vct{c}}

\newcommand{\ve}{\vct{e}}
\newcommand{\vf}{\vct{f}}
\newcommand{\vg}{\vct{g}}
\newcommand{\vh}{\vct{h}}

\newcommand{\vm}{\vct{m}}

\newcommand{\vp}{\vct{p}}
\newcommand{\vq}{\vct{q}}
\newcommand{\vr}{\vct{r}}
\newcommand{\vs}{\vct{s}}

\newcommand{\vu}{\vct{u}}
\newcommand{\vv}{\vct{v}}
\newcommand{\vw}{\vct{w}}
\newcommand{\vx}{\vct{x}}
\newcommand{\vy}{\vct{y}}
\newcommand{\vz}{\vct{z}}
\newcommand{\vzero}{\vct{0}}

\newcommand{\mA}{\mtx{A}}

\newcommand{\mC}{\mtx{C}}

\newcommand{\mF}{\mtx{F}}
\newcommand{\mG}{\mtx{G}}
\newcommand{\mH}{\mtx{H}}
\newcommand{\mI}{\mtx{I}}

\newcommand{\mR}{\mtx{R}}
\newcommand{\mS}{\mtx{S}}

\newcommand{\mX}{\mtx{X}}

\newcommand{\mZ}{\mtx{Z}}
\newcommand{\setA}{\set{A}}

\newcommand{\setC}{\set{C}}

\newcommand{\setH}{\set{H}}
\newcommand{\setI}{\set{I}}

\newcommand{\setN}{\set{N}}
\newcommand{\setO}{\set{O}}
\newcommand{\setP}{\set{P}}

\newcommand{\setS}{\set{S}}
\newcommand{\setT}{\set{T}}

\newcommand{\setX}{\set{X}}

\newcommand{\setZ}{\set{Z}}

\newcommand{\PP}{\mathbb{P}}

\def\gf{\nabla F}
\def\gg{\nabla G}
\def\gfh{\gf_{\vh}}
\def\ggh{\gg_{\vh}}
\def\gfx{\gf_{\vx}}
\def\ggx{\gg_{\vx}}
\def\tf{\tilde{F}}
\def\gtf{\nabla \tilde{F}}
\def\gtfh{\gtf_{\vh}}
\def\gtfx{\gtf_{\vx}}

\allowdisplaybreaks

\begin{document}
	

\title{Blind Deconvolution using Modulated Inputs}

\author{Ali~Ahmed
	\thanks{A. Ahmed is an Assistant Professor at the Department of Electrical Engineering, Information Technology University (ITU), Lahore 54000, Pakistan. Email: \texttt{alikhan@mit.edu.} This work was supported by the Higher Education Commission (HEC), Pakistan under the National Research Program for Universities (NRPU), Project no. 6856. Manuscript submitted on July 16, 2019.}}




\maketitle

\begin{abstract}
    This paper considers the blind deconvolution of multiple modulated signals/filters, and an arbitrary filter/signal. Multiple inputs $\boldsymbol{s}_1, \boldsymbol{s}_2, \ldots, \boldsymbol{s}_N =: [\boldsymbol{s}_n]$ are modulated (pointwise multiplied) with random sign sequences $\boldsymbol{r}_1, \boldsymbol{r}_2, \ldots, \boldsymbol{r}_N =: [\boldsymbol{r}_n]$,  respectively, and  the resultant inputs $(\boldsymbol{s}_n \odot \boldsymbol{r}_n) \in \mathbb{C}^Q, \ n \in [N]$ are convolved against an arbitrary input $\boldsymbol{h} \in \mathbb{C}^M$ to yield the measurements $\boldsymbol{y}_n = (\boldsymbol{s}_n\odot \boldsymbol{r}_n)\circledast \boldsymbol{h}, \ n \in [N] := 1,2,\ldots,N,$ where $\odot$ and $\circledast$ denote pointwise multiplication, and circular convolution. Given $[\boldsymbol{y}_n]$, we want to recover the unknowns $[\boldsymbol{s}_n]$ and $\boldsymbol{h}$.  We make a structural assumption that unknowns $[\boldsymbol{s}_n]$ are members of a known $K$-dimensional (not necessarily random) subspace, and prove that the unknowns can be recovered from sufficiently many observations using a regularized gradient descent algorithm  whenever the modulated inputs $\boldsymbol{s}_n \odot \boldsymbol{r}_n$ are long enough, i.e, $Q \gtrsim KN+M$ (to within logarithmic factors,  and signal dispersion/coherence parameters).  
	
	 Under the bilinear model, this is the first result on multichannel  ($N\geq 1$) blind deconvolution 
	 with provable recovery guarantees under near optimal (in the $N=1$ case) sample complexity estimates, and comparatively lenient structural assumptions on the convolved inputs. A neat conclusion of this result is that modulation of a bandlimited signal protects it against an unknown convolutive distortion. We discuss the applications of this result in passive imaging, wireless communication in unknown environment, and image deblurring.  A thorough numerical investigation of the theoretical results is also presented using phase transitions, image deblurring experiments, and noise stability plots.

\end{abstract}



\section{Introduction}

This paper considers blind deconvolution of a single input convolved against multiple modulated inputs. We observe the convolutions of  $\vh\in \C^M$, and\footnote{The notation $[\vs_n]$ along with other notations in the paper is introduced in Notations section below.}  $[\vs_n] \in \C^Q$  modulated in advance with known $[\vr_n] \in \R^Q$, respectively. Modulation $(\vs_n \odot \vr_n)$ is simply the pointwise multiplication (Hadamard product) of the entries of $\vs_n$, and $\vr_n$. Mathematically, the observed convolutions are
\begin{align}\label{eq:model}
\vy_n = \vh \circledast (\vr_n\odot \vs_n), \ n \in [N], 
\end{align}
where $\circledast$ denotes an $L$-point circular convolution\footnote{Two vectors in $\C^M$, and $\C^Q$ are zero-padded to length $L$ and then circularly convolved to return an $L$-point circular convolution.} operator giving $[\vy_n] \in \C^L$. We always set $L \geq \max(Q,M)$. We want to recover $[\vs_n]$, and $\vh$ from the circular convolutions $[\vy_n]$. We make a structural assumption that the entries of the modulating sequences $[\vr_n]$ are random, in particular,  binary $\pm1$. This implicitly means that the signs of the inputs $\vs_n\odot\vr_n, \ n \in [N]$ are random/generic. In applications, this may either be justified by analog signal modulations with binary waveforms (easily implementable) prior to convolutions or might implicitly hold when the inputs are naturally sufficiently diverse (dissimilar signs).

 This problem arises in passive imaging. An ambient uncontrollable source generates an unstructured signal that drives multiple convolutive channels, and one aims to recover the source signal, and channel impulse responses (CIRs) from the recorded convolutions. Recovering CIRs reveals important information about the structure of the environment such as in seismic interferometry \cite{curtis2006seismic,bharadwaj2018focused}, and passive synthetic aperture radar imaging \cite{marron1995passive}. Recovering the source signal is required in, for example, underwater acoustics to classify the identity of a submerged source \cite{sabra2004blind,sabra2010ray,tian2017multichannel}, and in speech processing to clean the reverberated records \cite{yoshioka2012making}. In general, the problem setup \eqref{eq:model} is of interest  in signal processing, wireless communication, and system theory.  Apart from other applications, one specific and interesting  result of general interest is that randomly modulating a bandlimited analog signal protects it against an unknown linear time invariant (or convolutive) system; this will also be discussed in detail below.

We assume that $[\vs_n]$ live in known subspaces, that is, each $\vs_n$ can be written as the multiplication of a known $Q \times K$ tall orthonormal matrix\footnote{The model and all the results in the paper can be easily be generalized to different matrix $\mC_n$ for each $n$.} $\mC$, and a short vector $\vx_n$ of expansion coefficients. Mathematically, 
\begin{align}\label{eq:subspaces}
\vs_n = \mC\vx_n, \ n \in [N]. 
\end{align}
 
  It is important to point out critical differences in the structural assumptions compared to the contemporary recent literature \cite{ahmed2015convex}, and \cite{ahmed2012blind,li2016rapid,lee2018spectral,ahmed2016leveraging,ahmed2018convex} on blind deconvolution, where at least one of the convolved signals is assumed to live in a known subspace spanned by the columns of a random Gaussian matrix.  Signal subspace in actual applications is often poorly described by the Gaussian model. We relinquish the restrictive Gaussian subspace assumption, and give a provable blind deconvolution result by only assuming random/generic sign (modulated) signals that reside in realistic subspaces spanned by DCT, wavelets, etc. 
 
 Additionally, we take $\vh$ to be completely arbitrary, and do not assume any structure such as a known subspace or sparsity in a known basis; this is unlike some of the other recent works  \cite{ahmed2016leveraging,ahmed2012blind,li2016rapid,lee2018spectral}.  In general, we take $M \leq L$ as already assumed in the observation model \eqref{eq:model}. Equivalently,  we assume length $M$ vector $\vh$ to be zero-padded to length $L$ before the $L$-point circular convolution. 
 In the particular case of multiple ($N>1$) convolutions, as will be shown later, no zero-padding is assumed, i.e., one can set $M$ as large as $L$. This is important in applications such as passive imaging, where often the source signal is uncontrolled, and unstructured, and hence cannot be realistically assumed to be zero-padded; see, the discussion in Section \ref{sec:applications}.
 
 Let $\mF$ be an $L \times L$ DFT matrix with entries $F[\omega,n] = \tfrac{1}{\sqrt{L}}e^{-\iota 2\pi \omega n/L}, \ (\omega,n) \in [L]\times[L].$  Formally, we take the measurements in the Fourier domain 
\begin{align}\label{eq:measurements}
\hat{\vy}_n &= \sqrt{L}\big(\mF_M\vh_0 \odot \mF_Q(\vr_n\odot \vs_{0,n})\big) + \hat{\ve}_n\notag\\
&= \sqrt{L}(\mF_M\vh_0 \odot \mF_Q\mR_n\mC\vx_{0,n}) + \hat{\ve}_n,
\end{align}
where $\mR_n := \text{diag}(\vr_n)$ is a $Q\times Q$ diagonal matrix, $\vh_0$, $\vx_0 : = \vctr([\vx_{0,n}])$ are the ground truths, $\hat{\vy}_n = \mF \vy_n$, and $[\hat{\ve}_n] \in \C^L$ denote the additive noise in the Fourier domain. To deconvolve, we minimize the measurement loss by taking a gradient step in each of the unknowns $\vh_0$, and $[\vx_{0,n}]$ while keeping the others fixed. This paper details a particular set of conditions on the sample complexity, subspace dimensions, and the signals/filters under which this computationally feasible gradient descent scheme provably succeeds.

\subsection{Notations}\label{sec:Notation}
Standard notation for matrices (capital, bold: $\mC$, $\mF$, etc.),  column vectors (small, bold: $\vx$, $\vy$, etc.), and scalars ($\alpha$, $c$, $C$, etc.) holds. Matrix and vector conjugate transpose is denoted by $*$ (e.g. $\mA^*$, $\vx^*$). A bar over a column vector $\bar{\vx}$ returns the same vector with each entry complex conjugated. 

In general, the notation $[\vx_n]$ refers to the set of vectors $\{\vx_1, \vx_2, \ldots, \vx_N\}$. Moreover, $[\vx_n] \in \C^Q$ means that $\vx_n \in \C^Q$ for every $n$. For a scalar $N$, we define $[N] := \{1,2,3,\ldots,N\}$. The notation $\vctr([\vx_n])$ refers to a  concatenation of $N$ vectors $\vx_1, \vx_2, \ldots, \vx_N$, i.e., $\vctr([\vx_n]) := [\vx_1^*, \vx_2^*, \ldots,\vx_N^*]^*$. 

For a matrix $\mC$, we denote by $\mC^{\otimes N}$, a block diagonal matrix $\sum_{n=1}^N \ve_n \ve_n^* \otimes \mC$, where $\otimes$ is the standard Kronecker product, and $[\ve_n]$ are the standard $N$-dimensional basis vectors. Building on this notation, we define, for example,  $(\mR_n\mC)^{\otimes N} := \sum_{n=1}^N\ve_n \ve_n^* \otimes \mR_n\mC$ for a sequence of matrices $\mR_1, \mR_2, \ldots, \mR_N$, and $\mC$.  We denote by $\mF_J$, a submatrix formed by first $L \times J$ columns of an $L \times L$ matrix $\mF$; e.g. $(\mF_J)^*$ denotes a $J \times L$ matrix. We denote by $\mI_K$, a $K \times K$ identity matrix. Absolute constants will be denoted by $C$, $c_1, c_2, \ldots $, their value might change from line to line. We will write $ A \gtrsim B$ if there is an absolute constant $c_1$ for which $A \geq c_1 B$. We use $\|\cdot\|_2$, $\|\cdot\|_\infty$ to denote standard $\ell_2$, and $\ell_\infty$ norms, respectively. Moreover,  $\|\cdot\|_*$, $\|\cdot\|_{2 \rightarrow 2}$, and $\|\cdot\|_F$  signify matrix nuclear, operator, and Frobenius norms, respectively. 
\subsection{Coherence Parameters}
Our main theoretical results depend on some signal dispersion measures that characterize how diffuse signals are in the Fourier domain. Intuitively, concentrated (not diffuse) signals in the Fourier domain \textit{annihilate}  the measurements in \eqref{eq:measurements} making it relatively difficult (more samples required) to recover such signals. We refer to the signal diffusion measures as coherence parameters, defined and discussed below. 

For arbitrary vectors $\vh \in \C^M$, $\vx := \vctr([\vx_n])\in \C^{KN}$, where $[\vx_n] \in \C^{K}$ for every $n$, we define coherences 
\begin{align}\label{eq:muh-nux}
\mu_h^2 : = L \frac{\|\mF_M\vh\|_\infty^2}{\|\vh\|_2^2}, \  \nu_x^2 := QN \frac{\|\mC^{\otimes N}\vx\|_\infty^2 }{\|\vx\|_2^2}, \ \text{and}  \ \nu_{\max}^2 := Q\|\mC\|_\infty^2,
\end{align}
where as noted in the notations section above, $\mC^{\otimes N}$ is a block diagonal matrix formed by stacking $N$ matrices $\mC$. 

Similar coherence parameters appear in the related recent literature on blind deconvolution \cite{ahmed2012blind,ahmed2016leveraging}, and elsewhere in compressed sensing \cite{candes09ex,candes2007sparsity}, in general. Without loss of generality, we only assume that $\|\vh_0\|_2 = \sqrt{d_0}$, and $\|\vx_0\|_2 = \sqrt{d_0}$. For brevity, we will denote the coherence parameters $\mu^2_{h_0}$, and $\nu^2_{x_0}$ of the fixed ground truth vectors $(\vh_0,\vx_0)$ by
\begin{align}\label{eq:mu-nu}
\mu^2:=\mu^2_{h_0}, \ \text{and} \ \nu^2 := \nu^2_{x_0}. 
\end{align}

In words, coherence parameter $\mu_h^2$ is the peak value of the frequency spectrum of a fixed norm vector $\vh$. A higher value roughly indicates a concentrated spectrum and vice versa. It is easy to check that $1 \leq \mu_h^2 \leq L$. 

On the other hand, $\nu_x^2$ quantifies the dispersion (not in the Fourier domain) of the signals $\vs_n = \mC\vx_n$. A signal concentrated in time (mostly zero) remains somewhat oblivious to the random sign flips $\vr_n \odot \vs_n$, and as a result is not as well-dispersed in the frequency domain. Let $\vc_{q,n}^*$ be the rows of $\mC^{\otimes N}$. By definition, $\nu_x^2\|\vx\|_2^2 \geq QN |\vc_{q,n}^*\vx|^2$ for any $(q,n) \in [Q] \times [N]$. Summing over $(q,n) \in [Q]\times[N]$ on both sides, and using the isometry of $\mC$ gives us the inequality $\nu_x^2 \geq 1$. The upper bound $\nu_x^2 \leq QN$ is easy to see using Cauchy Schwartz inequality, hence, $ 1 \leq \nu_x^2 \leq QN$. 

The third coherence parameter $\nu_{\max}^2$ in \eqref{eq:muh-nux} ensures that the subspace of vectors $\vs_n$ is generic---it is spanned by well-dispersed vectors. One can easily check that $1 \leq \nu_{\max}^2 \leq Q$, and the upper bound is achieved for $\mC = \begin{bmatrix} \mI_K \\ \mathbf{0} \end{bmatrix}$.

Our results indicate that for successful recovery, the sample complexity or the number $LN$ of measurements, and the length of the modulated signals $QN$ scale with $\mu^2$, $\nu^2$, and $\nu_{\max}^2$. 

\subsection{Signal Recovery via Regularized Gradient Descent}\label{sec:gradient-descent}
Notice that the measurements in \eqref{eq:measurements} are non-linear in the unknowns $(\vh_0,\vx_0)$, however, are linear in the rank-1 outer-product $\vh_0\bar{\vx}_0^*$.  To see this, let $\vf^*_\ell \in \C^M$ be the $\ell$th row of $\mF_M$, an $L \times M$ submatrix of the $L \times L$ normalized DFT matrix $\mF$, and $\hat{\vc}_{\ell,n} \in \C^{KN}$ be the $\ell$th row in the $n$th block-row of the  $LN \times KN$ block-diagonal matrix $\sqrt{L}(\mF_Q\mR_n\mC)^{\otimes N}$. The $\ell$th entry $\hat{y}_n[\ell]$ of measurements $\hat{\vy}_n$ in \eqref{eq:measurements} is then simply 
\begin{align}\label{eq:entrywise-measurements}
\hat{y}_n[\ell] = \vf_\ell^*\vh_0\bar{\vx}_0^*\hat{\vc}_{\ell,n}+\hat{e}_n[\ell] = \< \vf_\ell\hat{\vc}_{\ell,n}^*,\vh_0\bar{\vx}_0^*\>+\hat{e}_n[\ell];
\end{align}
the linearity of the measurements in $\vh_0\bar{\vx}_0^*$ is clear from the last inequality above. We also define a linear map $\setA: \C^{M \times KN} \rightarrow \C^{LN}$ that maps $\vh_0\vx_0^*$ to the vector $\hat{\vy} := \vctr([\hat{\vy}_n])$. The action of $\setA$ on a rank-1 matrix $\vh\vx^*$ returns 
\begin{align}\label{eq:linear-map}
&\setA(\vh\vx^*) := \{\vf_\ell^*\vh \bar{\vx}^*\hat{\vc}_{\ell,n}\}_{\ell,n}, \ (\ell, n)  \in [L] \times [N], \notag\\
&\text{and therefore,} \ \hat{\vy} = \setA(\vh_0\vx_0^*) + \ve,
\end{align}
where in the last display above $\ve := \vctr([\hat{\ve}_n])$, and we used the definition of $\setA$ to compactly express \eqref{eq:entrywise-measurements}.  It also shows that multichannel blind deconvolution with a shared input $\vh_0$ can be treated jointly as a rank-1 matrix $\vh_0\vx_0^* \in \C^{M \times KN}$ recovery problem, and the observations $[\hat{\vy}_n]$ in all the channels are the linear measurements of this common rank-1 object.
 
Given measurements $\hat{\vy}$ of the ground truth $(\vh_0,\vx_0)$, we employ a regularized gradient descent algorithm that aims to minimize a loss function:
\begin{align}\label{eq:tF-def}
\tf(\vh,\vx) := F(\vh,\vx)+G(\vh,\vx).
\end{align}
w.r.t. $\vh$, and $\vx$, where the functions $F(\vh,\vx)$, and $G(\vh,\vx)$  account for the measurement loss, and regularization, respectively; and are defined below 
\begin{align}\label{eq:F-def}
& F(\vh,\vx) := \|\setA(\vh\vx^*) - \hat{\vy}\|_2^2 =  \|\setA(\vh\vx^*-\vh_0\vx_0^*)-\ve\|_2^2= \notag \\
&\|\setA(\vh\vx^*-\vh_0\vx_0^*)\|_2^2 + \|\ve\|_2^2 - 2 \text{Re}(\<\setA^*(\ve),\vh\vx^*-\vh_0\vx_0^*\>),
\end{align}
and
\begin{align}\label{eq:G-def}
& G(\vh,\vx) := \rho\Bigg[ G_0\left( \frac{\|\vh\|_2^2}{2d} \right)+ G_0\left(\frac{\|\vx\|_2^2}{2d} \right)+\sum_{\ell=1}^L G_0\left( \frac{L |\vf_\ell^* \vh|^2}{8d \mu^2}\right) \notag\\
&\qquad\qquad + \sum_{q=1}^Q\sum_{n=1}^N G_0\left( \frac{QN |\vc_{q,n}^* \vx|^2}{8d \nu^2}\right)\Bigg], 
\end{align}
where $G_0 (z ) = \max\{z-1,0\}^2$. Conforming to the choice in the proofs below, we set $\rho \geq d^2+\|\ve\|_2^2$, and $ 0.9 d_0 \leq d \leq 1.1 d_0$ (proof of Theorem \ref{thm:initialization} below). Together, the first and third term in the regularizer $G(\vh,\vx)$ penalize any $\vh$ for which $\|\vh\|_2^2 > 2d$, and $\|\vh\|_2^2\mu_h^2 = L \|\mF_M\vh\|_\infty^2 > 8d \mu^2$. Similarly, the second and fourth term penalize the coherence $\nu_x^2$ and norm of $\vx$. In words, the regularizer keeps the coherences $\mu_h^2$, $\nu_x^2$; and norms of $(\vh,\vx)$ from deviating too much from those of the ground truth $(\vh_0,\vx_0)$.

The proposed regularized gradient descent algorithm takes alternate Wirtinger gradient (of the loss function $\tf{(\vh,\vx)}$) steps in each of  the unknowns $\vh$, and $\vx$ while fixing the other; see Algorithm \ref{algo:gradient-descent} below for the pseudo code. The Wirtinger gradients are defined as\footnote{For a complex function $f(\vz)$, where $\vz = \vu + \iota \vv \in \C^L$, and $\vu, \vv \in \R^L$, the Wirtinger gradient is defined as $
\frac{\partial f}{\partial \bar{\vz}} = \frac{1}{2}\left( \frac{\partial f}{\partial \vu} + \iota \frac{\partial f}{\partial \vv}  \right).$} 
\begin{align}\label{eq:gFh-gFx-def}
\nabla \tilde{F}_{\vh} := \frac{\partial \tilde{F}}{\partial \bar{\vh}} = \frac{\overline{\partial \tilde{F}}}{\partial {\vh}}, \ \text{and} \ \nabla \tilde{F}_{\vx} := \frac{\partial \tilde{F}}{\partial \bar{\vx}} = \frac{\overline{\partial \tilde{F}}}{\partial {\vx}}.
\end{align}

We explicitly write here the gradients of $F(\vh,\vx)$ in \eqref{eq:F-def} w.r.t. $\vh$, and $\vx_n$ to shed a bit more light on how the multichannel problem is jointly solved across all the channels for the same $\vh$.
Recall from \eqref{eq:linear-map} that by definition 
\begin{align*}
\setA(\vh\vx^*) = \begin{bmatrix} 
\mF_M\vh\odot \mF_Q \mR_1 \mC\vx_1\\
\mF_M\vh\odot \mF_Q \mR_2 \mC\vx_2\\
\vdots\\
\mF_M\vh\odot \mF_Q \mR_N \mC\vx_N
\end{bmatrix}.
\end{align*}
It is then easy to see that the gradients of $\mF(\vh,\vx)$ w.r.t. $\vh$, and $[\vx_n]$ are
\begin{align*}
&\gfh =\sum_{n=1}^N\mF_M^*\big[(\overline{\mF_Q\mR_n\mC\vx_n}) \odot (\mF_M\vh\odot \mF_Q \mR_n \mC\vx_n - \hat{\vy}_n)\big]\\
& \nabla F_{\vx_n}  = (\mF_Q \mR_n \mC)^*\big[(\overline{\mF_M\vh}) \odot (\mF_M\vh\odot \mF_Q \mR_n \mC\vx_n - \hat{\vy}_n)\big],
\end{align*}
where recall that bar notation $\bar{\vz}$ represents the entry wise conjugate of a vector $\vz$.
The gradient $\gfx$ is obtained by stacking $[\nabla F_{\vx_n}]$ in a column vector. However, as $\vh$ is fixed across channels, its gradient update is jointly computed as an average of the contributions from all the $N$ channels. Jointly solving for $\vh$  enables recovery of arbitrary $\vh$ with no additional assumption such as $\vh$ lying in a  known subspace as is the case for single-channel blind deconvolution \cite{ahmed2012blind,li2016rapid}.

Similar algorithms with provable recovery results appeared earlier beginning with \cite{jain2013low,sun2016guaranteed} for matrix completion, and \cite{candes2015phase} for phase retrieval, and in \cite{li2016rapid} for blind deconvolution, however, with observation model different from \eqref{eq:model} considered here.

\begin{algorithm}[H]
	\caption{Wirtinger gradient descent with a step size $\eta$}
	\label{algo:gradient-descent}
	\begin{algorithmic}\small
	    \State \textbf{Input:} Obtain $(\vu_0,\vv_0)$ via Algorithm \ref{algo:initialization} below.
		\For{$t = 1, \ldots$}
	    \State $ \vu_t \gets  \vu_{t-1} - \eta \nabla \tf_{\vh} (\vu_{t-1},\vv_{t-1} )$
		\State $ \vv_t \gets  \vv_{t-1} - \eta \nabla \tf_{\vm} (\vu_{t-1},\vv_{t-1} )$
		\EndFor
	\end{algorithmic}
\end{algorithm}

Finally, a suitable initialization $(\vu_0,\vv_0)$ for Algorithm \ref{algo:gradient-descent} is computed using Algorithm \ref{algo:initialization} below. In short, the left and right singular vectors of $\setA^*(\hat{\vy})$ when projected in the set of sufficiently incoherent (measured in terms of the coherence $\mu$, $\nu$ of the original vectors $\vh_0$, and $\vx_0$) vectors supply us with the initializers $(\vu_0, \vv_0)$. 
\begin{algorithm}[H]
	\caption{Initialization}\small
	\label{algo:initialization}
	\begin{algorithmic}
		\State \textbf{Input:} Compute $\setA^*(\hat{\vy})$, and find the leading singular value $d$, and the corresponding left and right singular vectors $\hat{\vh}_0$, and $\hat{\vx}_0$, respectively. 
		\State Solve the following optimization programs 
		\State $\vu_0 \gets \underset{\vh}{\text{argmin}} \ \|\vh - \sqrt{d} \hat{\vh}_0\|_2, \ \text{subject to} \  \sqrt{L}\|\mF_M\vh\|_\infty \leq 2 \sqrt{d} \mu,$\ \text{and} \ 
		\State $\vv_0 \gets \underset{\vx}{\text{argmin}} \ \|\vx - \sqrt{d} \hat{\vx}_0\|_2, \ \text{subject to} \ \sqrt{QN}\|\mC^{\otimes N}\vx\|_\infty \leq 2 \sqrt{d} \nu.$
		\State \textbf{Output:} $(\vu_0,\vv_0)$. 
	\end{algorithmic}
\end{algorithm}

\subsection{Main Results} 
Our main result shows that given the convolution measurements \eqref{eq:measurements}, a \textit{suitably initialized} Wirtinger gradient-descent Algorithm \ref{algo:gradient-descent} converges to the true solution, i.e., $(\vu_t, \vv_t) \approx (\vh_0, \vx_0)$ under an appropriate choice of $Q$, $N$, and $L$. To state the main theorem, we need to introduce some neighborhood sets. For vectors $\vh \in \C^M$, and $\vx \in \C^{KN}$, we define the following sets of neighboring points of $(\vh,\vx)$ based on either, magnitude, coherence, or the distance from the ground truth.
\begin{align}
&\setN_{d_0} := \{(\vh,\vx) | \|\vh\|_2 \leq 2\sqrt{d_0}, \ \|\vx\|_2\leq 2 \sqrt{d_0} \},\label{eq:setNd-def}\\
&\setN_\mu := \{ (\vh,\vx) | \sqrt{L}\|\mF_M\vh\|_\infty \leq 4\mu\sqrt{d_0} \},\label{eq:setNmu-def}\\
&\setN_\nu : = \{(\vh,\vx)|  \sqrt{QN}\|\mC^{\otimes N}\vx\|_\infty \leq 4\nu\sqrt{d_0} \},\label{eq:setNnu-def}\\
&\setN_{\varepsilon} := \{(\vh,\vx) | \|\vh\vx^*-\vh_0\vx_0^*\|_{F} \leq \varepsilon d_0 \}\label{eq:setNe-def}.
\end{align}

Our main result on blind deconvolution from modulated inputs \eqref{eq:measurements} is stated below.  
\begin{thm}\label{thm:convergence}
	 Fix $0 < \varepsilon \leq 1/15$. Let $\mC \in \R^{Q \times K}$ be a tall basis matrix, and set $\vs_{0,n} = \mC\vx_{0,n}$; and $\vx_{0,n} \in \C^K$ for every $n = 1,2,3, \ldots, N$, and $\vh_0 \in \C^M$ be arbitrary vectors. Let the coherence parameters of $\mC$ and  $(\vh_0,\vx_0)$ be as defined in \eqref{eq:mu-nu}.  Let $[\vr_n]$ be independently generated $Q$-vectors with standard iid Rademacher entries. We observe the $L$-point circular convolutions of the random sign vectors $\vr_n\odot\vs_{0,n}$ with $\vh_0$, where $ L \geq \max(Q,M)$, leading to observations \eqref{eq:measurements} contaminated with additive noise $\ve$. Assume that the initial guess $(\vu_0,\vv_0)$ of $(\vh_0, \vx_0)$ belongs to $\tfrac{1}{\sqrt{3}}\setN_{d_0} \cap \tfrac{1}{\sqrt{3}}\setN_{\mu}\cap \tfrac{1}{\sqrt{3}}\setN_{\nu} \cap \setN_{\frac{2}{5}\varepsilon},$ then Algorithm \ref{algo:gradient-descent} will create a sequence $(\vu_t,\vv_t) \in \setN_{d_0}\cap \setN_{\mu} \cap \setN_{\nu} \cap \setN_{\varepsilon}$,  which converges geometrically (in the noiseless case, $\ve = \vzero$) to $(\vh_0,\vx_0)$, and there holds 
	 \begin{align}\label{eq:stable-recovery-bound}
	 \|\vu_{t+1}\vv^*_{t+1}-\vh_0\vx_0^*\|_F &  \leq  \tfrac{2}{3}(1-\eta\omega)^{(t+1)/2}\varepsilon d_0
	 + 50 \|\setA^*(\ve)\|_{2\rightarrow 2}
	 \end{align}
	 with probability at least 
	 \begin{align}\label{eq:probability-main-thm}
	 1-2\exp\left(-c\delta_t^2QN/\mu^2\nu^2\right),
	 \end{align}
	  whenever 
	 \begin{align}\label{eq:sample-complexity-main-thm}
	 QN \geq \frac{c}{\delta_t^2} (\mu^2\nu_{\max}^2 KN^2 + \nu^2 M)\log^4(LN),
	 \end{align}
     where $\delta_t := \|\vu_t\vv^*_t-\vh_0\vx_0^*\|_F/d_0$, $\omega > 0$, and $\eta$ is the fixed step size. Fix $\alpha \geq 1$. For noise $\ve \sim \text{Normal}(\mathbf{0},\frac{\sigma^2 d_0^2}{2LN}\mI_{LN}) + \iota \text{Normal}(\mathbf{0},\frac{\sigma^2 d_0^2}{2LN}\mI_{LN})$, $\|\setA^*(\ve)\|_{2\rightarrow 2} \leq \frac{2\varepsilon}{50} d_0$ with probability at least $1-\setO((LN)^{-\alpha})$ whenever 
	\begin{align}\label{eq:sample-complexity-LN}
	LN \geq  c_\alpha\frac{\sigma^2}{\varepsilon^2}  \max(M,KN\log(LN))\log(LN).
	\end{align}
\end{thm}
The above theorem claims that starting from a good enough initial guess the gradient descent algorithm converges super linearly to the ground truth in the noiseless case. The theorem below guarantees that the required good enough initialization: $(\vu_0,\vv_0) \in \frac{1}{\sqrt{3}}\setN_{d_0} \cap \frac{1}{\sqrt{3}}\setN_{\mu}  \cap \frac{1}{\sqrt{3}}\setN_{\nu}\cap \setN_{\frac{2}{5}\varepsilon}$ is supplied by Algorithm \ref{algo:initialization}. 
\begin{thm}\label{thm:initialization}
	The initialization obtained via Algorithm \ref{algo:initialization} satisfies $
	(\vu_0,\vv_0) \in\frac{1}{\sqrt{3}}\setN_{d_0} \cap \frac{1}{\sqrt{3}}\setN_{\mu}  \cap \frac{1}{\sqrt{3}}\setN_{\nu}\cap \setN_{\frac{2}{5}\varepsilon},$ and $0.9 d_0 \leq d \leq 1.1 d_0$ 	holds with probability at least $1-2\exp\left(-c\varepsilon^2 QN/\mu^2\nu^2\right)$ whenever 
	\[
	QN \geq \frac{c}{\varepsilon^2} \left(\mu^2\nu_{\max}^2 KN^2 + \nu^2 M\right) \log^4(LN).
	\]
\end{thm}
\noindent Proofs of Theorem \ref{thm:convergence}, and \ref{thm:initialization} are given in Section \ref{sec:convergence}, and Appendix \ref{sec:Theorem2-Proof}, respectively. 
\subsection{Discussion}\label{sec:discussion}

Theorem \ref{thm:convergence}, and \ref{thm:initialization} together prove that randomly modulated unknown  $Q$-vectors $[\vs_{0,n}]$, and unknown $M$-vector $\vh_0$ can be recovered with desired accuracy from their $N$ circular convolutions $\vh_0 \circledast (\vr_n\odot\vs_{0,n}), \ n \in [N]$ under suitably large $Q$, $N$, and $L$. We will refer to the bounds in \eqref{eq:sample-complexity-main-thm}, \eqref{eq:sample-complexity-LN}, and $L \geq \max(Q,M)$ as \textit{sample complexity bounds}. Together these give
\begin{align}\label{eq:final-sample-complexity}
LN \geq QN \gtrsim (\mu^2\nu_{\max}^2 KN^2 + \nu^2 M)\log^4 (LN)
\end{align}
for a fixed desired accuracy $\delta_{t+1}$ of the recovery. We want to remark here that the result only guarantees an approximate recovery as is the case in some earlier works \cite{keshavan2012efficient,jain2013low,hardt2014understanding} on matrix completion using non-convex methods. For example, [16] uses fresh independent samples to compute a stochastic gradient update for technical reasons of avoiding dependencies among the iterates; this leads to an approximate recovery. On the other hand, our measurement model gives rise the dependent scalar measurements \eqref{eq:entrywise-measurements} across index $\ell$, and does not give way to a natural splitting of measurements in batches of independent samples. Fortunately, we do not need batch splitting in this proof method, however, we are still only able to guarantee approximate recovery; the main technical reason being a limited structured randomness in the linear map $\setA$, which does not lead to strong concentration bounds. 
For "exact" recovery, i.e., with error $\delta_{t+1} = 0$ the method requires infinitely many samples. We leave this mainly a technical challenge of improving the results to finite sample complexity to the future work. All the remaining discussion in this section will be assuming a fixed accuracy $\delta_t$, and hence  a constant factor $1/\delta^2_t$ in the sample complexity bound. 

 We now provide a discussion on the interpretation of these sample complexity bounds in  several interesting scenarios such as single ($N=1$) and multiple ($N > 1$) convolutions to facilitate the understanding of the reader. 

\textbf{Sample Complexity:} Observe that the number of unknowns in the system of equations \eqref{eq:measurements} is $KN+M$. Combining $L \geq \max(Q,M)$,  \eqref{eq:sample-complexity-main-thm}, and \eqref{eq:sample-complexity-LN}, it becomes clear that number $LN$ of measurements required for successful recovery scale with $KN^2+M$ (within coherences, and log factors). This shows that the sample complexity results are off by a factor of $N$ compared to optimal scalings. We believe this is mainly a limitation of the proof technique; the phase transitions in Section \ref{sec:numerics} show that \eqref{eq:sample-complexity-main-thm} is a conservative bound, and successful deconvolution generally occurs when $LN\geq QN \gtrsim KN+M$; see phase transitions in numerical simulations Section \ref{sec:phase-transitions}.

In general, for multiple convolution, we require $LN \geq QN \gtrsim (\mu^2\nu_{\max}^2 KN^2 + \nu^2 M)\log^4 (LN).$ In passive imaging problem, where an ambient source drives multiple CIRs, the above bound places a minimum required length $QN$ of CIRs for a successful blind deconvolution from the recorded data.

 In the case of single ($N=1$) convolution, the above bound reduces to $L \geq Q \gtrsim (\mu^2\nu_{\max}^2 K + \nu^2 M)\log^4 L$. The number of unknowns in this case are only $K +M$. Unlike the multiple convolutions case above,  for a desired accuracy of the recovered estimate, the bound on $L$ above is information theoretically optimal (within log factors and coherence terms). This sample complexity result almost matches the results in \cite{ahmed2012blind,li2016rapid} except for an extra log factor. However, the important difference is, as mentioned in the introduction, that unlike  \cite{ahmed2012blind,ahmed2016leveraging,li2016rapid,ahmed2018convex,lee2018spectral}, the inputs are not assumed to reside in Gaussian subspaces rather only have random signs, which if not given can also be enforced through random modulation in some applications; see, Section \ref{sec:applications}. 

\textbf{No zero-padding of $\vh_0$:} 
Assume that the unknown filter $\vh_0$  is completely arbitrary, that is, its length $M$ can be as large as $L$. Equivalently, no zero-padding or, in general, no known subspace is assumed.  Even in the case of linear systems of equations, the recovery is only possible in this case whenever $LN \geq KN+L$, i.e., the number $LN$ of measurements exceeds the number $KN+L$ of unknowns. Evidently, $LN \geq KN+L$ can never be achieved in the single channel scenario ($N=1$). However, in the multichannel scenario ($N$ strictly bigger than $1$)  $LN \geq KN+L$  is possible by setting $N$, and $L$ to be sufficiently large, and hence successful recovery may also be possible. In light of \eqref{eq:final-sample-complexity}, we have that inputs $[\vx_{0,n}]$, and a filter  $\vh_0$ with length $M$, possibly as large as $L$, can be recovered whenever $L$, and $N$ are chosen to be sufficiently large such that $LN \geq QN \gtrsim (\mu^2 \nu^2_{\max} KN^2 + \nu^2 L)\log^4 (QN)$ holds. The numerics consistently show that the successful recovery occurs at a near optimal sample complexity, i.e., $LN \geq QN \gtrsim (KN+L)$ indicating a probable room of improvement in the derived performance bounds.

No zero padding has practical importance in passive imaging, where an unstructured and uncontrollable source signal with no discernible on, or off time is driving the CIRs \cite{sabra2010ray}, and hence, one cannot realistically assume any zero-padding.

Finally, the bound in \eqref{eq:sample-complexity-main-thm} might appear contradictory to a reader as it guarantees recovery for a longer filters/signals (large enough $Q$) whereas one should expect that deconvolution must be easier if the convolved signals are shorter (fewer overlapping copies); for example, deconvolution is immediate in the trivial case of one tap $Q=1$ filters. However, it is important to note that the bound \eqref{eq:sample-complexity-main-thm} only gives a range of $Q$, and $N$ under which recovery is certified, and in no way eliminates the possibility of a successful blind deconvolution when it is violated. Roughly speaking in our case, longer length $Q$ only introduces more sign randomness and makes the inverse problem (blind deconvolution) well-conditioned. 

 \section{Applications}\label{sec:applications}
The measurement model in \eqref{eq:measurements} finds many applications owing to minimal structural assumptions on the signals. We present three applications scenarios including an implementable modulation system to protect an analog signal against convolutive interference using a real time preprocessing, channel protection in wireless communications, random mask imaging, and passive imaging. 

\subsection{Channel Protection using Random Modulators}\label{sec:channel-protection-application}

One of the important results of this paper is that binary modulations of an analog bandlimited signal protect it against unknown linearly convolutive channels. For illustration, consider a simple scenario of wireless communication of a periodic\footnote{We restrict the discussion to a periodic signal mainly to reduce the mathematical clutter. A non-periodic signal can be handled within a time limited window and smoothing around the edges.} signal $s(t)$ in $t \in [0,1)$, bandlimited to $B$ Hertz. Expansion of $s(t)$ using Fourier basis functions is 
\begin{align*}
s(t) = \sum_{k = - B}^{B} x[k] \mathrm{e}^{\iota 2 \pi k t}.
\end{align*}
The signal $s(t)$ can be captured by taking $Q \geq K := 2B+1 $ equally spaced samples at time instants $t \in \setT_Q:= \{0,\tfrac{1}{Q}, \ldots, 1-\tfrac{1}{Q}\}$. Let $\mF_K$ be a matrix formed by the $K$ columns (corresponding to the signal frequencies) of a normalized $Q \times Q$ DFT matrix $\mF$. Then the samples of $s(t)$ can be expressed as $\vs = \mF_K\vx$, where the Fourier coefficients $x[k]$ are the entries of the $K$-vector $\vx$. The signal $s(t)$ is modulated by a binary waveform $r(t)$ alternating at a rate $Q$.  Let $\vr$ be a $Q$-vector of samples of binary waveform $r(t)$ in $t \in \setT_Q$. The modulated signal $s(t)\odot r(t)$ undergoes an unknown linear transformation $(s(t)\odot r(t)) \circledast h(t) = y(t)$ through an LTI system, where $h(t)$ is the impulse response of an LTI system, and is given as 
\begin{align*}
h(t) = \sum_{m=1}^M h[m] \delta(t-t_m),  \ \text{where} \ t_m \in \setT_Q.
\end{align*}
Assuming $h[m]$ to be the $m$th entry of an $M$-vector $\vh$. The samples $\vy$ of the transformed signal $y(t)$ exactly take the form
\begin{align}\label{eq:observation-model}
\vy = \mR\mF_{K}\vx \circledast \vh, 
\end{align}
where as before $\mR = \text{diag}(\vr)$, and $\vy \in \C^L$.

Since the observation model \eqref{eq:observation-model} aligns with the model considered in this paper, a direct application of our main result shows that $\vs$, and hence $s(t)$ can be recovered from the received signal $y(t) = (s(t)\odot r(t)) \circledast h(t)$ without knowing CIR by operating the random binary waveform $r(t)$ at rate $Q \gtrsim (\mu^2K+\nu^2M)\log^4 L$, and sampling the received signal $y(t)$ at a rate $L =Q$, where we used the fact that $\nu_{\max}^2 = 1$ for the DFT matrix $\mF_K$ above.  The coherences $\nu^2$, and $\mu^2$ are simply the peak values in time $ \|\vs\|_\infty^2$, and frequency domain $\|\mF_M\vh\|_\infty^2$, respectively, where $\mF_M$ are the first $M$ columns of a normalized $L\times L$ DFT matrix. 

Binary modulation of an analog signal can be easily implemented using switches that flip the signs of the signal in real time; the setup is shown in Figure \ref{fig:Modulators}. Fast rate binary switches can be easily implemented; see, for example, \cite{laska2007theory}, and  \cite{tropp2010beyond,ahmed2015compressive,ahmed2018compressive} for the use of binary switches in other applications in signal processing. The implementation potential combined with the ubiquity of blind deconvolution make this result interesting in system theory, applied communications, and signal processing, among other.
\begin{figure}
	\centering
	\begin{tabular}{cc}
		\includegraphics[scale = 0.55, trim = 8.5cm 7cm 8cm 0cm,clip]{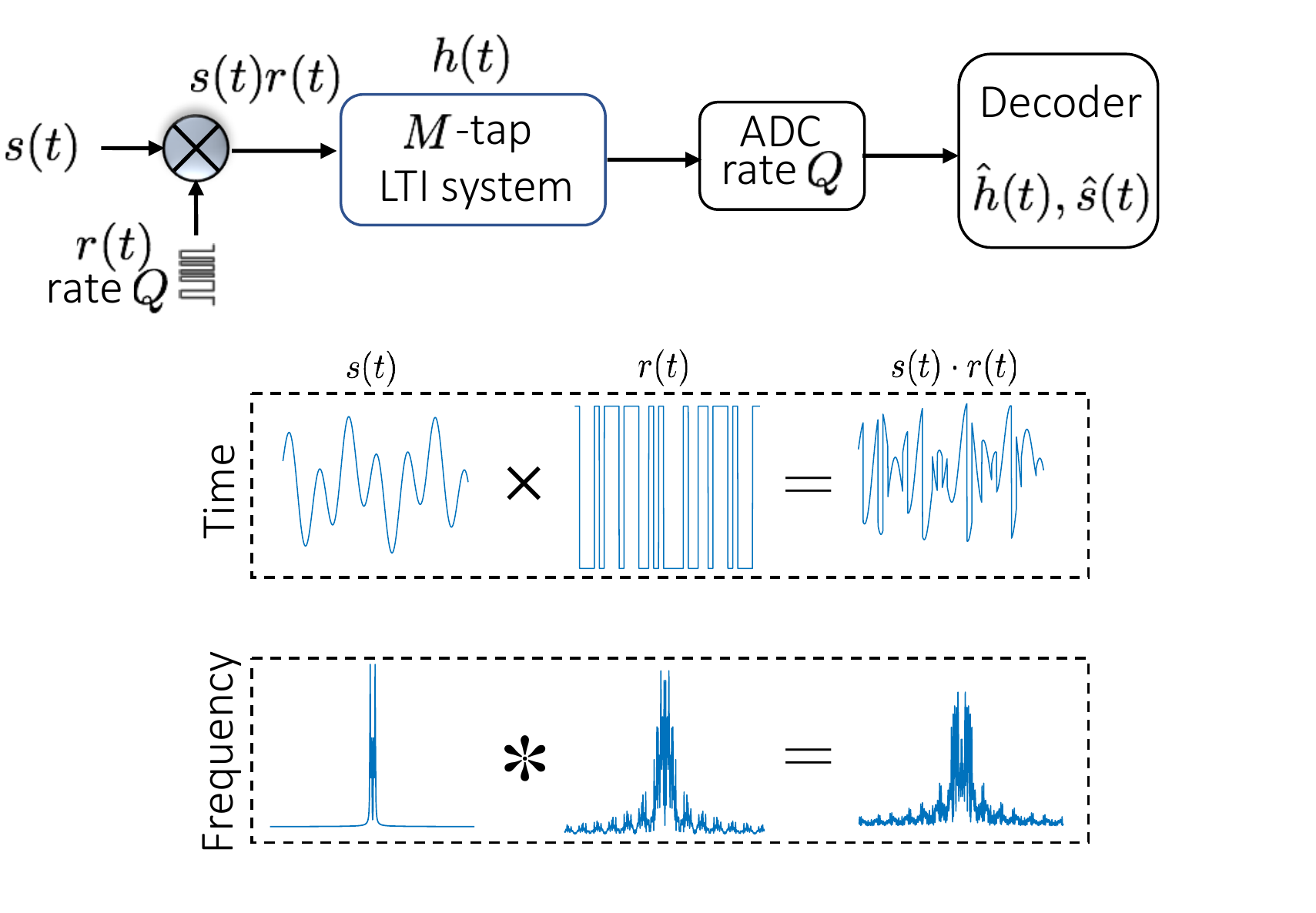}
	\end{tabular}
	\caption{\small\sl Analog implementation for real time protection against channel intereference.  A continuosus time signal $s(t)$, bandlimited to $B$ Hz, is modulated with a random binary waveform $r(t)$ alternating at a rate $Q$. The modulated signal drives an unknown LTI system characterized by an $M$-tap impulse response $h(t)$. The resulting signal is sampled at a rate $Q$. Operate the modulator and ADC at a rate $Q \gtrsim \max(B,M)$ (to within a constant, log factors and coherences), and recover $s(t)$, and $h(t)$ using algorithm \ref{algo:gradient-descent}. Underneath, the preprocessing is shown in time, and frequency domain. Modulation in time domain spreads the spectrum, and the resulting higher frequency signal remains protected against the distortions caused by an unknown LTI system.}
	\label{fig:Modulators}
\end{figure} 

The signal subspace can be other than Fourier vectors in practical application in wireless communications, for example, channel coding protects a message against unknown errors by introducing redundancy in the messages. This operation can be viewed as the multiplication of the message vector with a tall matrix $\mC$. The coded message $\vs = \mC\vx$ is transmitted over an unknown channel characterized by an impulse response $\vh \in \R^M$.  A simple, and easy to implement additional step of randomly flipping the signs $\vr\odot\vs$ of the coded message  enables the decoder to recover $\vx$ from several delayed, and attenuated overlapping copies $(\vr\odot\vs)\circledast \vh$ of the transmitted codeword; see Figure \ref{fig:Channel-Protection} for a pictorial illustration.  \begin{figure*}
	\centering
	\begin{tabular}{cc}
		\includegraphics[scale = 0.5, trim = 0cm 11.7cm 0 0.7cm,clip]{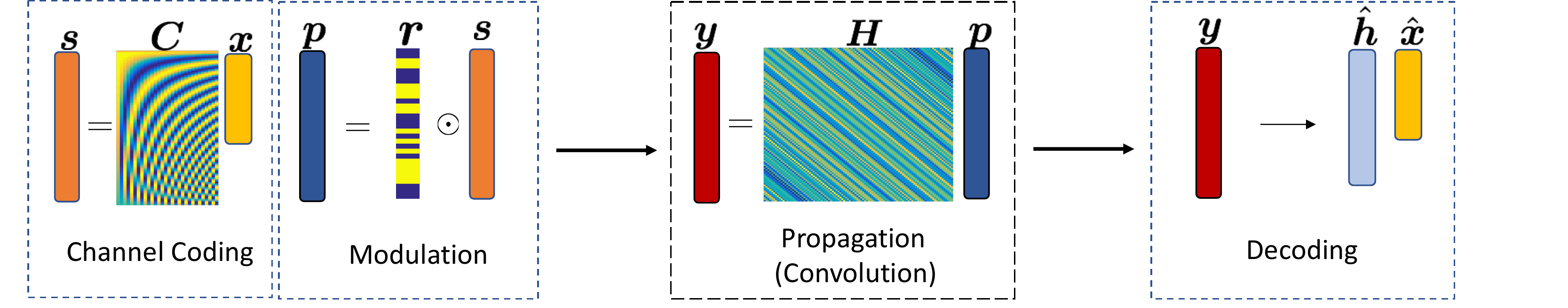}
	\end{tabular}
	\caption{\small\sl  Channel protection in wireless communications: A user message $\vx$ is encoded by multiplying with a tall coding matrix $\mC$ (known subspace) in the conventional channel coding block. The signs of the resultant symbols $\vs$ are randomly flipped (modulation). This signal is then transmitted and undergoes a series of reverberations and distortions (modeled as a convolution with CIR) while propagating to the receiver. The decoder estimates the symbols $\vx$, and  CIR $\vh$. }
	\label{fig:Channel-Protection}
\end{figure*} 
\subsection{Random Mask Imaging}\label{sec:RMI}
A stylized application of the blind deconvolution problem in \eqref{eq:measurements} is in image deblurring using random masks. Images are observed through a lens, which introduces an unknown blur in the images. To deblur the images, this paper suggests an image acquisition system, shown in Figure \ref{fig:RMI}, in which a programmable spatial light modulating (SLM) array is placed between the image and the lens. SLM modulates ($\pm 1$) the light reflected off of the object before it passes through the lens. While ideal binary masks are $0/1$, we consider $(\pm 1)$ for technical reasons; the $(\pm 1)$ masks can be implemented in practices using a $(0/1)$ mask together with all $1$'s mask. The light impinging on the detector array is convolution of point spread function of lens with randomly modulated images. Assuming  an apriori knowledge of the subspace of each image, which might be a subset of a carefully selected wavelet or DCT bases functions, we can deblur the images using gradient descent as discussed in Section \ref{sec:gradient-descent}. The relative dimension of the image subspaces w.r.t. image, and blur size must obey the sample complexity bounds presented in Theorem \ref{thm:convergence}; see Section \ref{sec:numerics} for details. 

It is instructive to compare our results with a recent and closely related random mask imaging (RMI) setup given in \cite{bahmani2015lifting} for image deblurring. A similar physical setup is studied, and recovery of a blurred image is achieved by placing a random mask between the lens and image, however, two important differences exist compared to our approach. Firstly, in \cite{bahmani2015lifting}, and other works in this direction \cite{harikumar1999perfect,sroubek2012robust}, one image is fed multiple times through different random masks to improve the conditioning of the inverse problem, whereas in our setup we use a different unknown image every time. This is very important in applications, where it is not possible to obtain multiple snapshots of the same scene as it is dynamic. For example, imagine imaging a culture of micro-organisms; the moving organisms and the fluid around continuously changes the formation of the micro-organisms. Secondly, the image deblurring in \cite{bahmani2015lifting} is achieved via a computationally expensive semidefinite program operating in the \textit{lifted} space of dimension $KNM$. On the other hand, the gradient descent scheme in Algorithm \ref{algo:gradient-descent} is  computationally efficient as it operates in natural parameter space of dimension only $KN+M$. 
\begin{figure}
	\centering
	\begin{tabular}{cc}
		\includegraphics[scale = 0.28, trim = 1cm 1cm 0 3cm,clip]{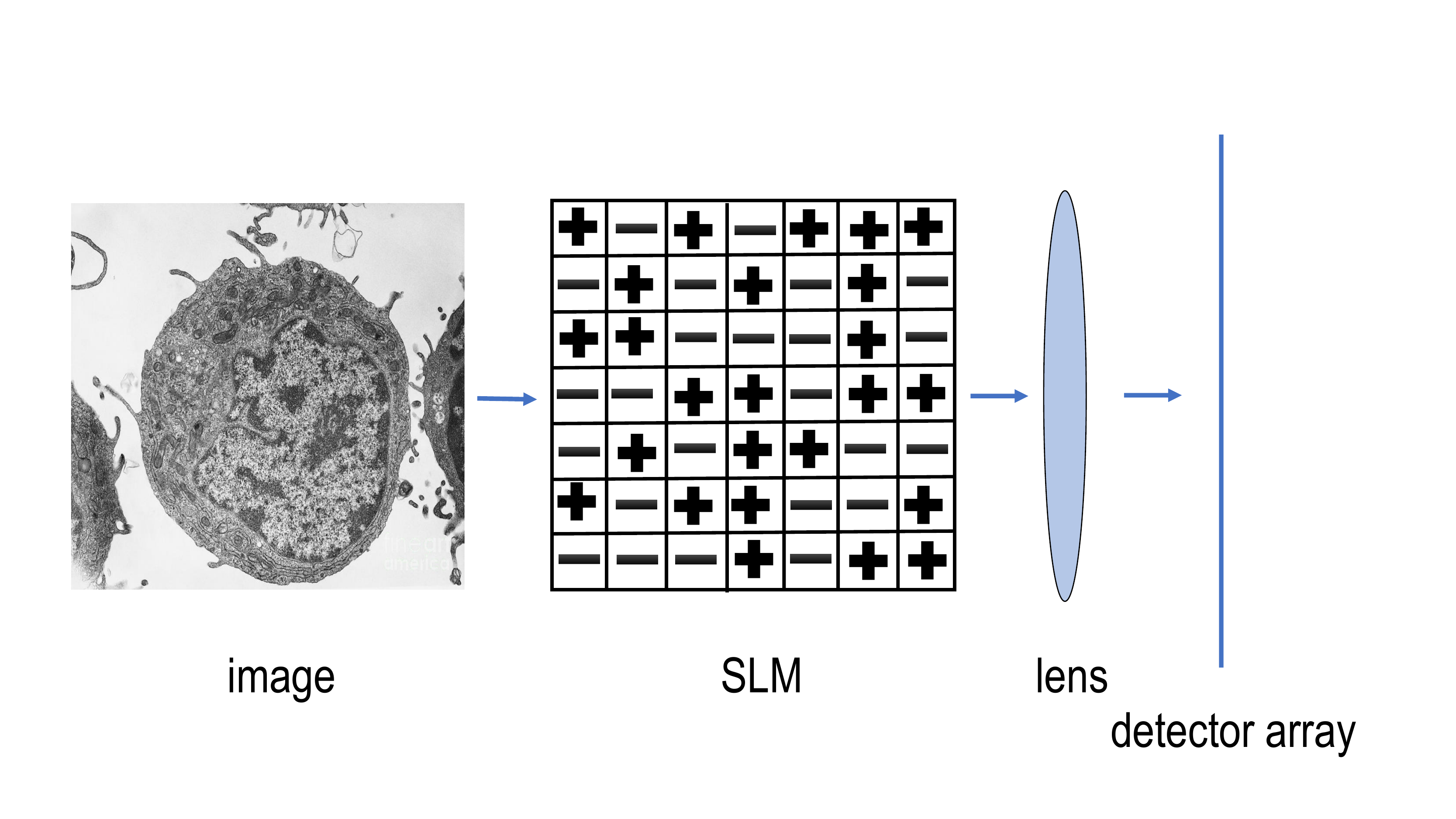}
	\end{tabular}
	\caption{\small\sl Schematic of random mask imaging setup. The reflection of a target image from a spatial light modulator (SLM) is blurred by a lens and the resultant intensities are measured on a detector array. Everytime a new image is observed through this system (with a new mask pattern) and eventaully the lens blur kernel and all the images are discovered using a gradient descent algorithm.}
	\label{fig:RMI}
\end{figure} 

\subsection{Passive Imaging: Multichannel Blind Deconvolution}\label{sec:MBD-applications}

In passive imaging, a source signal $s(t)$ feeds multiple convolutive channels. The signal $s(t)$ is not observed/controlled, and is unstructured. For example, in seismic experiments, a drill generates noise like signature that propagates through earth subsurfaces. The reflected copies from earth layers overlap and are recorded at multiple receivers. To characterize the subsurfaces, a multichannel blind deconvolution (MBD) on the received data discovers the Green's function; for details, see an interesting recent work \cite{bharadwaj2018focused}, and references therein. In underwater acoustics, a submerged source signal is distorted, reverberated while propagating through the water media. Multiple passive sensors on water surface record the distorted signals. The source recognition is better if the recorded data is cleaned using blind deconvolution \cite{byun2017blind}. 

The recorded data at each of the receivers in the passive imaging applications above takes the form\footnote{Compared to the model in \eqref{eq:model}, the role of $s$, and $h$ is swapped in this section as there is one source signal $s(t)$ and multiple CIRs $h_n(t).$}.  
\begin{align}\label{eq:passive-imaging}
y_n(t) = s(t) \circledast h_n(t), \ n = 1,2,3,\ldots, N,
\end{align}
where $h_n(t)$'s are short CIRs.
Importantly , Theorem \ref{thm:convergence} clearly determines the combined length $QN$ of CIRs must exceed the length $M$ of the source signal, as is evident from the bound in \eqref{eq:sample-complexity-main-thm}. This means that for longer (meeting the generic sign assumption) CIRs, one can guarantee to resolve a longer length of source signal from the recorded data.

MBD was studied with keen interest in 90's; see, \cite{douglas1997multichannel,amari1997multichannel} for some of the least squares based approaches.   Using commutativity of convolutions, an effective strategy  \cite{xu1995least,moulines1995subspace} relies on the null space of the cross correlation matrix of the recorded outputs.  Recovery using these spectral methods depends on the condition that CIRs do not share common roots in the $z$-domain --- some of the MBD schemes developed based on this observation can be found in \cite{subramaniam1996cepstrum,lin2006two,huang2002adaptive}.

MBD has also been reexamined more recently using semidefinite programming (SDP) \cite{ahmed2015convex,ahmed2016leveraging,ahmed2018convex}, and spectral methods \cite{lee2018spectral} that enjoy theoretical performance guarantees under restrictive Gaussian known subspace assumptions on CIRs. In comparison to computationally expensive SDP operating in lifted domain, and spectral methods, we present a gradient descent scheme for MBD with provable guarantees under a weaker random signs assumption on the CIRs. The generic/random sign assumptions on the CIRs might implicitly hold naturally, or could be made more likely to hold using indirect means such as  arranging the locations of the receivers at \textit{dissimilar} points might lead to diverse CIRs. Moreover, as already discussed in Section \ref{sec:discussion}, we donot assume any unrealistic structure such as known subspace, or zero-padding on the source signal $s(t)$, and it can be completely arbitrary. This perfectly models the unstructured source signal in passive imaging.
\subsection{Other Related Work}\label{sec:related-work}
 A  regularized gradient descent scheme to minimize the non-convex measurement loss was rigorously analyzed recently in \cite{li2016rapid} for the single channel ($N=1$) blind deconvolution, and was shown to be provably effective under the known Gaussian subspace assumption. In comparison, we study the multichannel blind deconvolution, and the problem set up \eqref{eq:measurements} also has much limited and structured randomness in a diagonal $\mR_n$ compared to a dense Gaussian matrix used in \cite{li2016rapid}. This requires a considerably more intricate proof argument based on generic chaining \cite{talagrand2005generic} to show approximate stable recovery using a regularized gradient descent algorithm. Recently, \cite{ma2017implicit} showed that (vanilla) gradient descent without the additional regularization term such as \eqref{eq:G-def} shows provably similar recovery guarantees for blind deconvolution under Gaussian subspace model as given in \cite{li2016rapid}. Extending a result  similar to \cite{ma2017implicit} for our case \eqref{eq:measurements} remains challenging as unlike the case of Gaussian subspace considered in \cite{ma2017implicit}, the samples in $\hat{\vy}_n$ in \eqref{eq:measurements} are statistically dependent. Numerically, we observe  similar performance to Algorithm \ref{algo:gradient-descent} even if the regularization term \eqref{eq:G-def} is not included.
 
Observations in \eqref{eq:measurements} are bilinear in the unknowns $(\vh_0,\vx_0)$.
Denoting $\mF_M\vh_0 = \hat{\vh}_0$, the measurements \eqref{eq:measurements} can be rescaled on both sides by the (element-wise) inverse $\vh_0^{-1}$ to give
\begin{align*}
\hat{\vh}_0^{-1}\odot \hat{\vy}_n = \sqrt{L}\mF_Q\mR_n\mC \vx_{0,n} + \hat{\vh}_0^{-1}\odot \hat{\ve}_n.
\end{align*}
Clearly, the problem is now linearized  \cite{balzano2007blind} in the unknowns $(\hat{\vh}_0^{-1},\vx_{0,n})$, and one can proceed with the recovery using the least squares objective below 
\[
\underset{\vg,\{\vx_{0,n}\}_n}{\minimize}\ \sum_{n=1}^N\|\vg\odot\hat{\vy}_n - \sqrt{L} \mF_Q\mR_n\mC\vx_{0,n}\|_2^2.
\]
 The drawbacks of this approach are its sensitivity to the noise components that are amplified due to the weighting  $\hat{\vh}_0^{-1}\odot \hat{\ve}_n$, and will affect the overall least squares recovered solution. Moreover, the problem can be framed as finding a smallest eigenvector of a matrix, and an inherent ambiguity exists if it has more than one-dimensional null space. \cite{ling2018self} gives provable recovery results using a least squares approach under various random subspace models. \cite{li2018blind} relinquishes the subspace model and instead assumes $\vs_n$ admit sparse representations in Gaussian random matrices, and proves the signal recovery using the same linearized eigenvector approach using a power iteration algorithm under strict spectrally flatness conditions  on the signals. The performance under noise in these linearized schemes \cite{ling2018self,li2018blind} is only guaranteed under additional assumptions on filter invertibility $\vh_0^{-1}$, and on the magnitudes of the entries of $\vh_0^{-1}$. In comparison, we directly work with the bilinear model, and give the first provable approximate stable (under noise) recovery results for blind deconvolution using random modulations.

Multichannel blind deconvolution problem can be framed as a rank-1 matrix recovery problem \cite{romberg2013multichannel,ahmed2016leveraging}. Exact and stable recovery results from optimally many measurements are derived in \cite{ahmed2016leveraging} when the signals lie in random Gaussian subspaces. 

The question of the uniqueness of the solution $(\vh_0,\vx_0)$ (up-to global scaling) of the multichannel bilinear problems of the form 
\begin{align*}
\hat{\vy}_n = \hat{\vh}_0 \odot \hat{\mC}\vx_{0,n}
\end{align*}
has been studied in \cite{li2016optimal}. In particular, necessary and sufficient conditions for the identifiability of $(\hat{\vh}_0,\vx_{0})$ were given for almost all $(\hat{\vh}_0,\hat{\mC},\vx_0)$. In the particular case of $\vh_0 \in \C^L$, and choosing $\hat{\vh}_0 = \mF\vh_0$, and $\hat{\mC} = \mF_Q \mR_n \mC$ makes the last display above equivalent to the measurement model \eqref{eq:measurements} in the noiseless case. Thus applying Theorem 2.1 in \cite{li2016optimal} would imply that if 
\[
L > K, ~ \text{and} ~ \frac{L-1}{L-K} \leq N \leq K
\]
then for almost all $\hat{\vh}_0$, $\mF_Q \mR_n \mC$, and $\vx_0$, the pair $(\hat{\vh}_0, \vx_0)$ is identifiable up to global scaling. The results show that identifiability is possible under optimal sample complexity $LN \geq KN+L-1$ for almost all $(\hat{\vh}_0,\mF_Q \mR_n \mC,\vx_0)$. Compared to this results, our derived sample complexity bound \eqref{eq:final-sample-complexity} is off by a factor $N$ (to within log factors, and coherences). The numerics also show that this additional factor of $N$ on the right hand side of \eqref{eq:final-sample-complexity} is not required to obtain successful recovery in practice. Necessary and sufficient conditions on the modulation rate $Q$ for the identifiability of the unknowns, however, do not directly follow from the work in [39], and are an open question. The numerics suggest that successful recovery occurs whenever $QN$ is roughly of the order of the number of unknowns.

Multichannel blind deconvolution from observations $\vy_n = \vh_0 \circledast \vs_n$ under the assumption that $\vs_n$ are sparse vectors has also been studied \cite{wang2016blind,li2018global}. The blind inverse problem is solved by looking for a filter $\vg$ such that $[\vg \circledast \vy_n]$ are sparse. Sparsity is promoted using a convex penalty such as $\ell_1$ norm \cite{wang2016blind}, or more recently using a different convex relaxation involving $\ell_4$ norm \cite{wang2016blind}. However, the provable sample complexity results are far from optimal; for details, see \cite{li2018global}. In comparison, we assume that $\vs_n$ reside in a known subspace, and have generic sign patterns that either exist naturally or can be explicitly enforced using random modulation. This model nicely fits some practical applications as already laid out in Section \ref{sec:applications}. 

We would also like to discuss a related paper \cite{cosse2017note} that considers recovering $\vs_n \in \C^L$, $n \in [N]$, and $\vh \in \C^L$ from the convolutions
\[
\vy_n = \vh \circledast (\vr_n \odot \vs_n), n \in [N],
\]
where $\|\mF\vs_n\|_0 \leq K$. Theorem 1.1 in \cite{cosse2017note} claims that $\vh$, and $\vs_n$ can be recovered with probability at least $1-\setO(L^{-\beta})$ by solving a convex program whenever $L \gtrsim \beta N$, and $N \gtrsim \beta K^2$, and that 
\begin{align}\label{eq:coherence}
\frac{ \sup_{n,k} |s_n[k]|^2}{\inf_k \sum_{n=1}^N |s_n[k]|^2} = \set{O}\left( \frac{1}{N}\right).
\end{align}
However, it seems that the at least the statement of Theorem 1.1 in \cite{cosse2017note} is not correct as there are several assumptions made in the proof argument such as $|\hat{s}_n[k]|^2 = \setO(1/L)$, and $h^2_{\min} := \min_{\ell} |\hat{h}[\ell]|^2 = \setO(1/L)$, which do not appear in the statement of Theorem 1.1 in \cite{cosse2017note}. In addition, Theorem 1.1 \cite{cosse2017note} claims recovery under comparatively strict 'coherence requirements such as 
\begin{align*}
\frac{ \sup_{n,k} |s_n[k]|^2}{\inf_k \sum_{n=1}^N |s_n[k]|^2} = \set{O}\left( \frac{1}{N}\right).
\end{align*}
For example, to satisfy this coherence condition, it must always be true that
\[
\max_n \frac{\|\vs_n\|_2^2}{\sum_{n=1}^N \|\vs_n\|_2^2} = \setO\left( \frac{1}{N}\right),
\]
which says that energy must be roughly equally shared among all the inputs $\vs_n$. Not only that the share of energy should be roughly equal across the corresponding entries of $s_n[k]$ as well as is clear from \eqref{eq:coherence}. Together with this, the proof also uses other strict flatness conditions:
\begin{align*}
|\hat{s}_n[k]|^2 = \setO\left(\frac{1}{L}\right), \quad \text{and} \quad h^2_{\min} = \setO\left(\frac{1}{L}\right), 
\end{align*}
where $\|\vs_n\|_2^2 = 1$ for every $n \in [N]$, and $\|\vh\|_2^2 = 1$. These conditions basically enforce that $\vs_n$, $\vh$  have to be flat in the frequency domain for successful recovery. In comparison, the required coherence parameters \eqref{eq:muh-nux} in our paper are much milder and successful recovery is still possible under Theorem 1 of our paper for any value of these coherences (smaller or larger).

Blind deconvolution has also been studied under various assumptions on input statistics, some important references are \cite{tong1995blind,tong1994blindtime}. We complete the brief tour of the related works in the above sections by pointing readers to survey articles \cite{campisi2016blind,tong1998multichannel} to account for other interesting works that we might have missed in the expansive literature on this subject.

\section{Numerical Simulations}\label{sec:numerics}
In this section, we numerically investigate the sample complexity bounds using phase transitions. We showcase random mask image deblurring results, and also report stable recovery in the presence of additive measurement noise. 
\subsection{Phase transitions}\label{sec:phase-transitions}
We present phase transitions to numerically investigate the constraints in \eqref{eq:sample-complexity-main-thm}, and \eqref{eq:sample-complexity-LN} on the dimensions $Q$, $N$, $M$, $K$, and $L$ for the gradient descent algorithm to succeed with high probability. The shade represents the probability of failure, which is computed over hundred independent experiments. For each experiment, we generate Gaussian random vectors $\vh_0$, and $[\vx_{0,n}]$, and choose $\mC$ to be the subset of the columns of a DCT matrix for every $[\vs_{0,n}]$. The synthetic measurements are then generated following the model \eqref{eq:model}. We run Algorithm \ref{algo:gradient-descent} initialized via Algorithm \ref{algo:initialization}, and classify the experiment as successful if the relative error 
\begin{align}\label{eq:relative-error}
\text{Relative Error}: = \frac{\|\hat{\vh}\hat{\vx}^*-\vh_0\vx_0^*\|_F}{\|\vh_0\vx_0^*\|_F}
\end{align}
is below $10^{-2}.$ The probability of success at each point is computed over hundred such independent experiments. 

The first four (left  to right)  phase diagrams in Figure \ref{fig:phase-transitions} investigate successful recovery using four different (one for each phase diagram) lengths $Q$ of the modulated inputs, and varying values of $K$, and $M$ while keeping $L$, and $N$ fixed. We set $L = 3200$, and $N = 1$ in all four phase transitions while $Q$ is fixed at $800, 1600, 2400$, and $3200$, respectively. Clearly, the white region (probability of success almost 1) expands with increasing $Q$. For example, in first (top left) phase transition, successful recovery occurs almost always when the measurements are a factor of $9$ above the number of unknowns, that is, $L \approx 9(K+M)$, and this factor improves to $5, 3$, and $2.8$ from second to fourth phase transition, respectively. These phase transitions show that successful recovery is obtained for a wide range of shorter to longer unknown random sign filters/signals, however, successful recovery happens more often for longer (larger $Q$) modulated inputs.

\begin{figure*}
	\centering
	\begin{tabular}{cccc}
		& \includegraphics[scale = 0.27, trim = 0cm 5cm 1.3cm 5cm,clip]{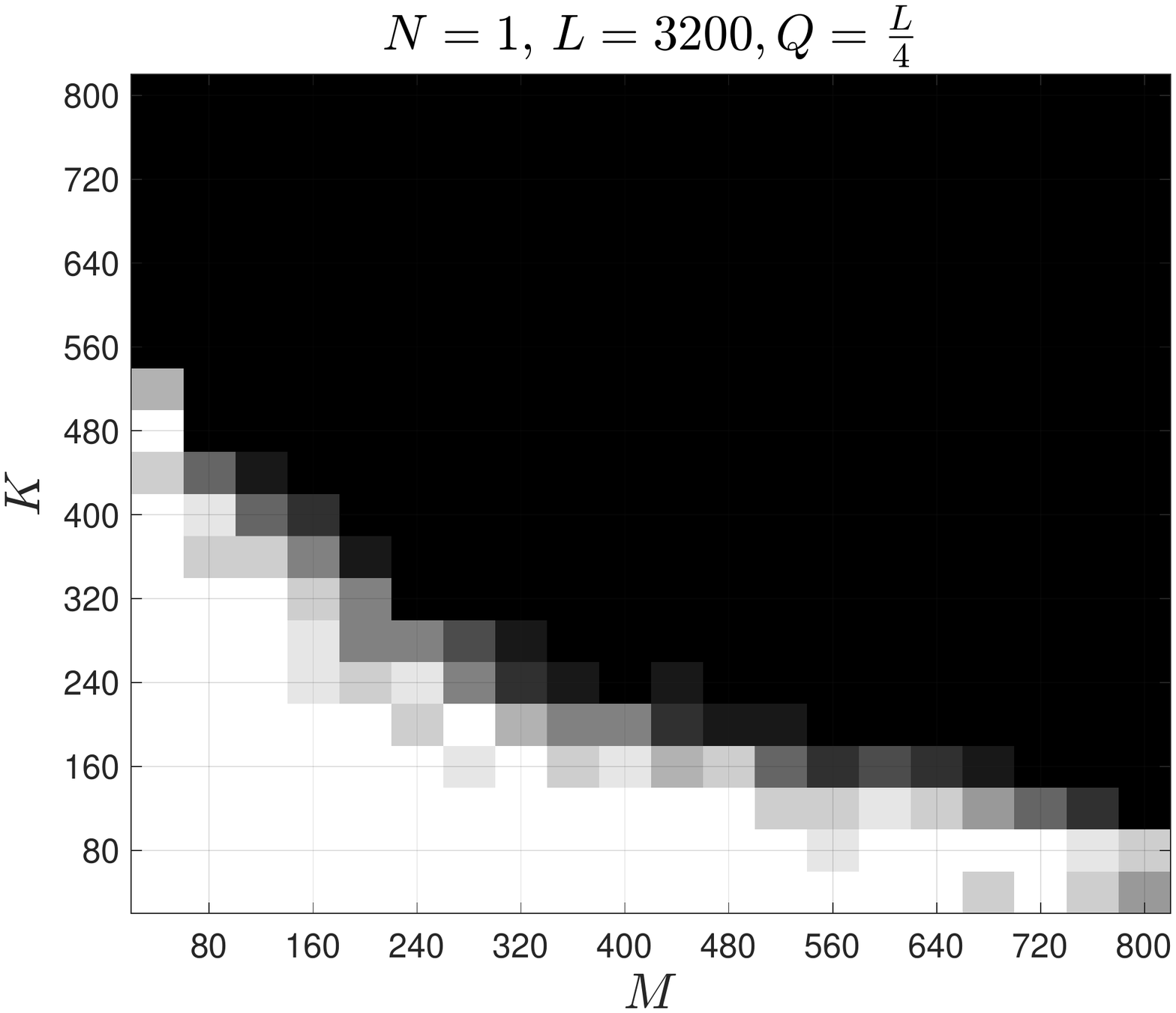}
		&\includegraphics[scale = 0.27, trim = 0.2cm 5cm 1cm 5cm,clip]{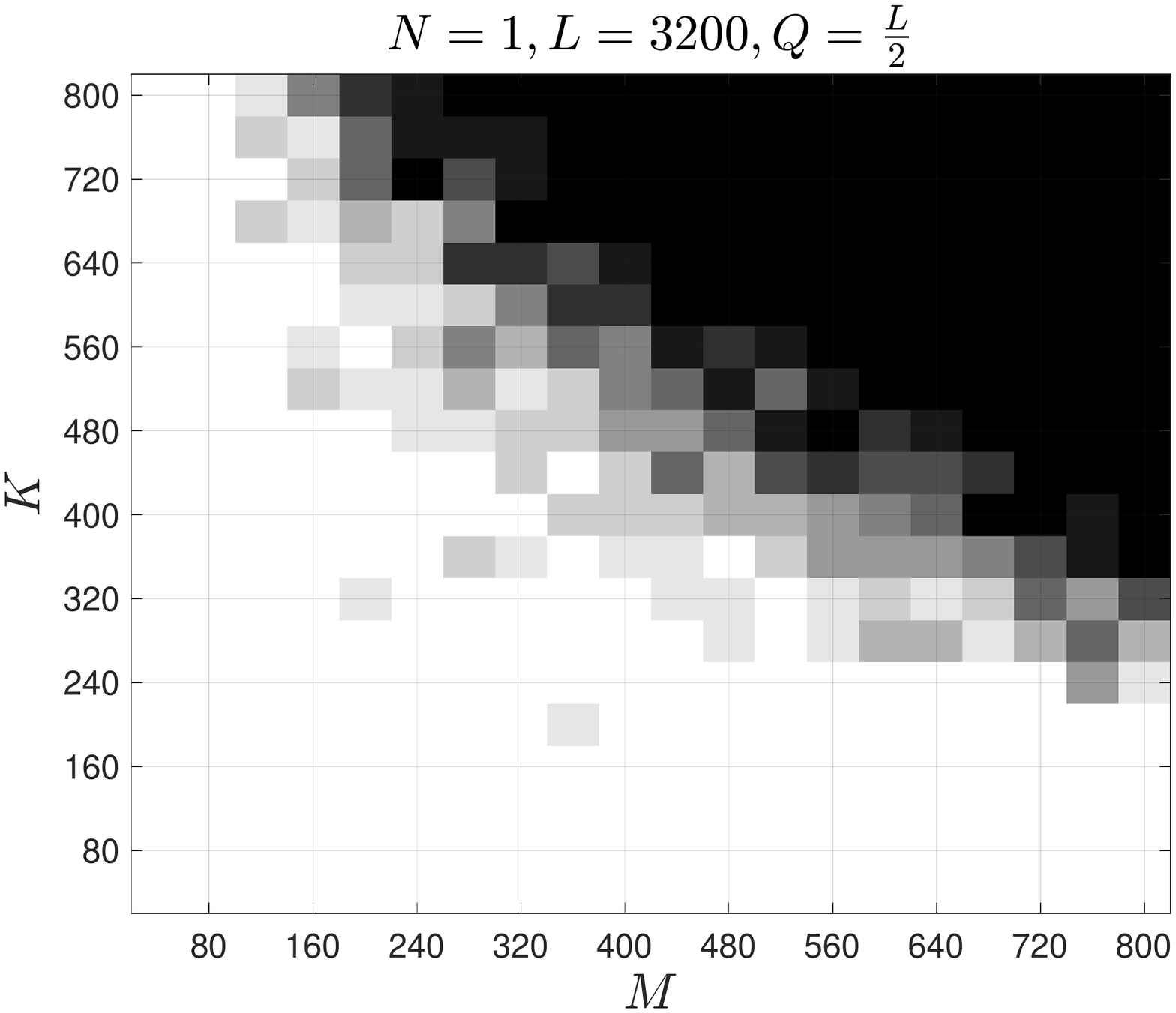}
		&\includegraphics[scale = 0.27, trim = 0cm 5cm -2cm 5cm,clip]{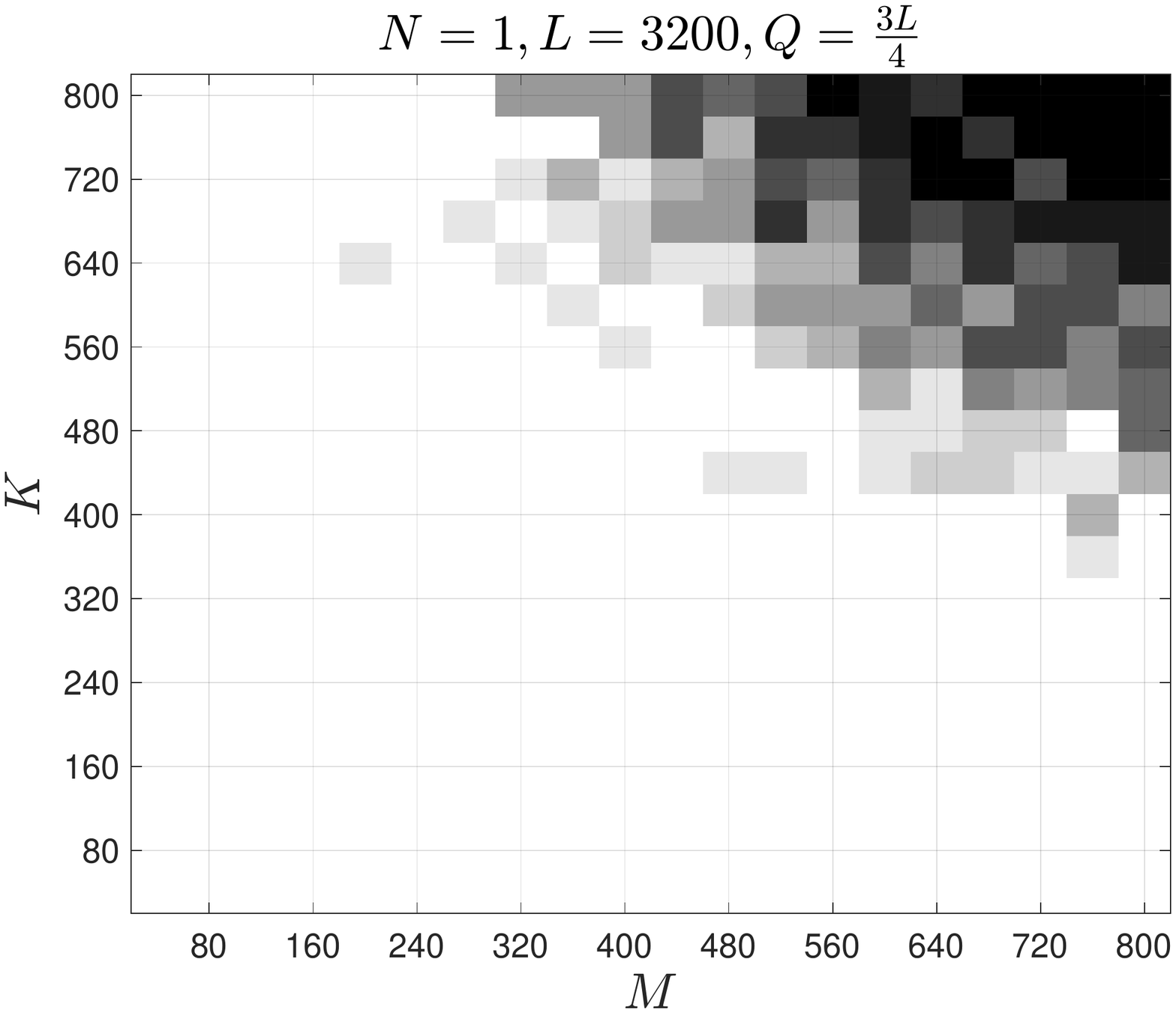}\\
		&\includegraphics[scale = 0.27, trim = 0cm 5cm 1cm 5cm,clip]{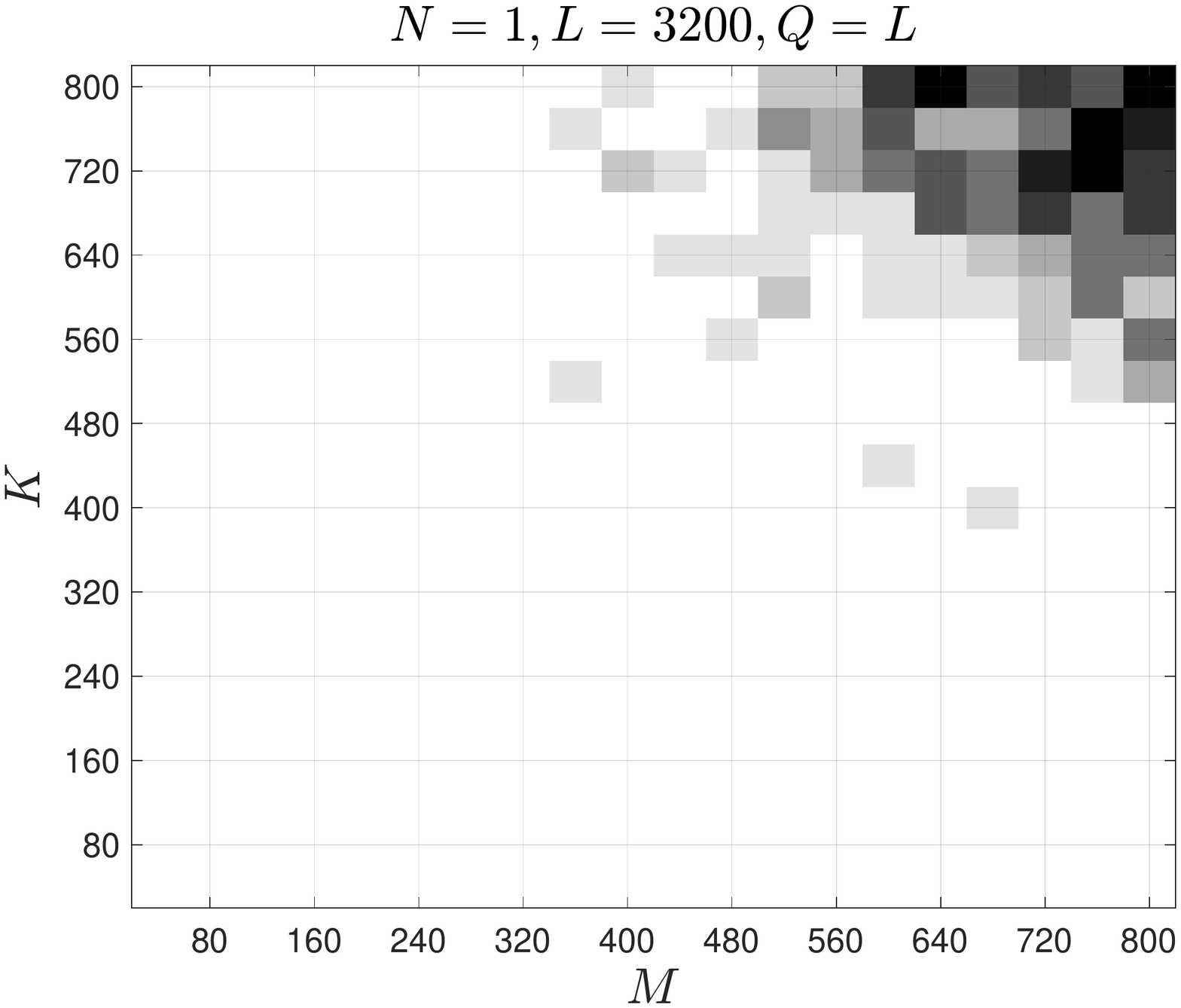}
		&\includegraphics[scale = 0.27, trim = 0cm 5cm 1 5cm,clip]{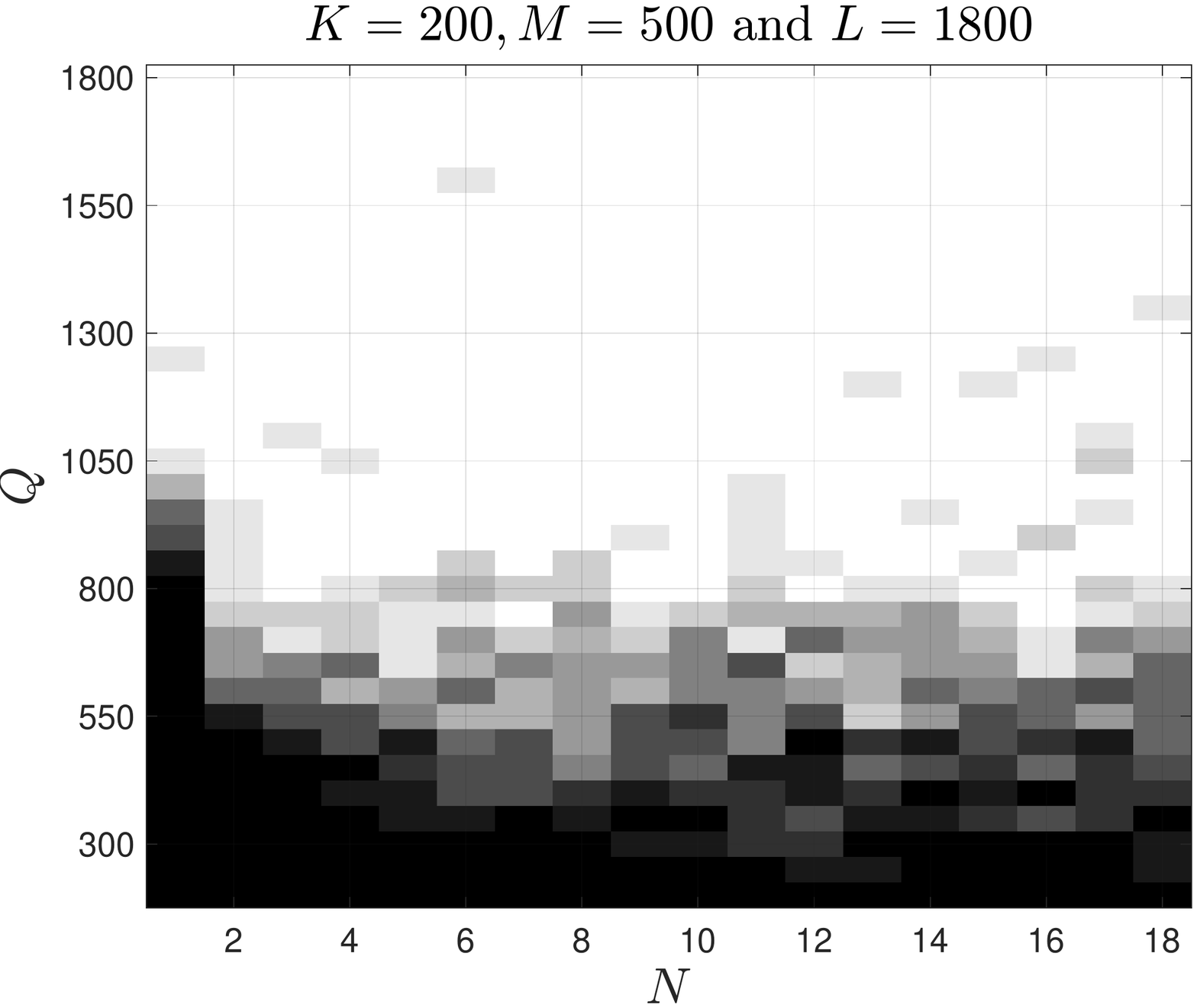}
		&\includegraphics[scale = 0.27, trim = 0.2cm 5cm -2cm 5cm,clip]{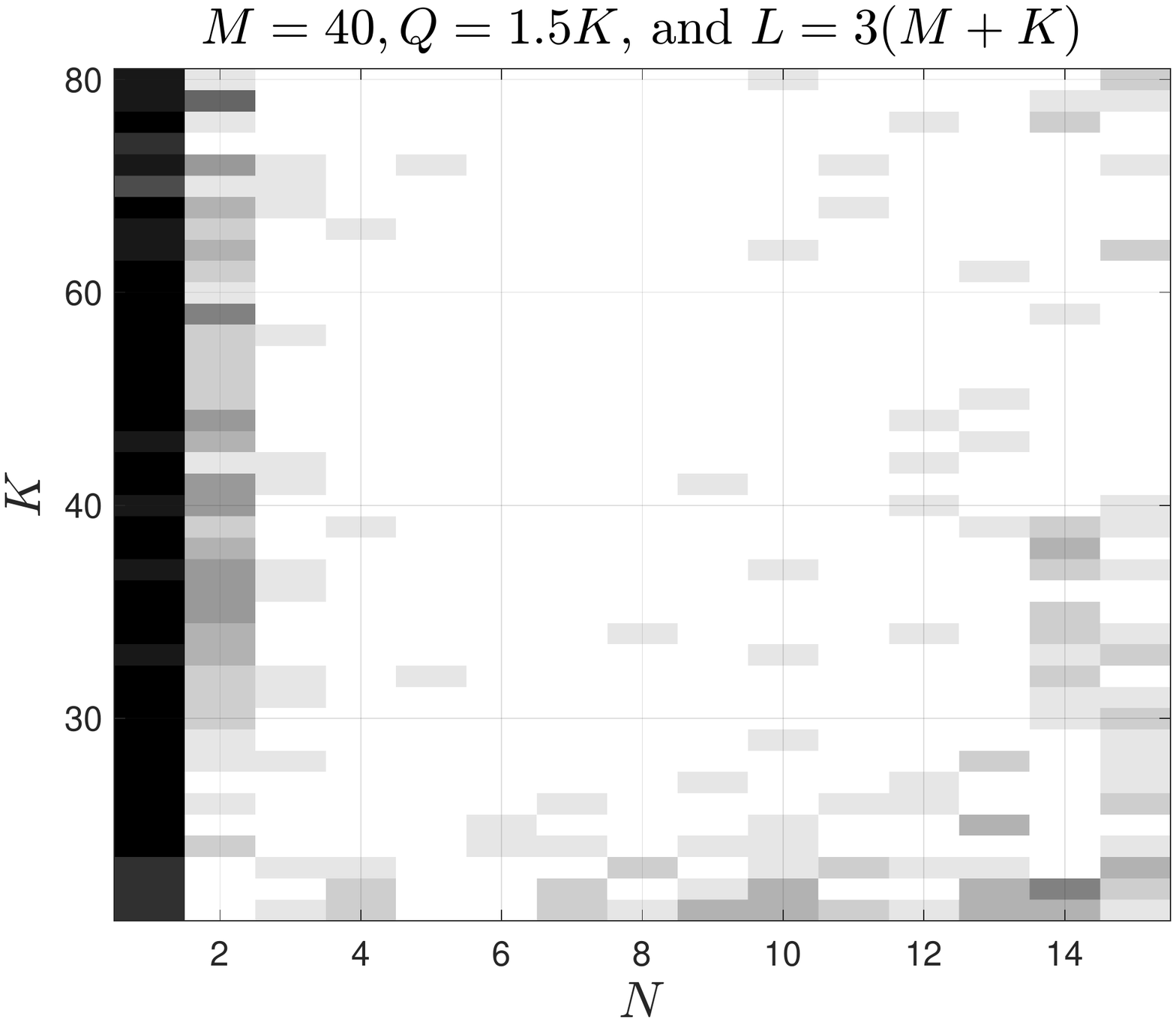}
	\end{tabular}
	\caption{\small\sl First four (left to right) are phase transitions of $K$ vs. $M$ for fixed $Q$, $N$, and $L$.  Together these four phase diagrams show that longer (larger $Q$) modulated inputs allow recovery with larger values of $K$, and $M$. Fifth phase transition is $Q$ vs. $N$ for a fixed $L$, $K$, and $M$.  Successful recovery almost always occurs when $Q \approx 2(K+M)$ across all values of $N$, and hence showing a better scaling than the linear scaling predicted by the theory in \eqref{eq:sample-complexity-main-thm}. Sixth phase transition is $K$ vs. $N$ for a fixed $M$, $Q$, and $L$. $Q$, and $L$ are chosen to be much more pessimistically than \eqref{eq:sample-complexity-main-thm}, however, the dominant white region shows that multichannel $N >1$ leads to favorable results even in such pessimistic regimes. }
	\label{fig:phase-transitions}
\end{figure*} 

 Recall the discussion in Section \ref{sec:discussion}, where we pointed out that the bound in \eqref{eq:sample-complexity-main-thm} is conservative by a factor of $N$. The fifth phase transition in Figure \ref{fig:phase-transitions} investigates the affect of $N$ on minimum value of $Q$ required for successful recovery, and shows that numerically this value of $Q$ does not increase with increasing $N$, and is roughly on the order of $K+M$, and not $KN+(M/N)$ as predicted in \eqref{eq:sample-complexity-main-thm} in Theorem \ref{thm:convergence}. 
 
 Finally, the last phase transition in Figure \ref{fig:phase-transitions} investigates $K$ vs. $N$ under fixed $M$, and $L$, and setting $Q = 1.5K$. It shows that increasing $N$ improves the frequency of successful recovery even under a pessimistic choice of $Q = 1.5K$  in comparison to  the bound in \eqref{eq:sample-complexity-main-thm}, which suggests that $Q\gtrsim KN+(M/N)$. 
 
 In summary, the phase diagrams suggest that $LN \geq QN \gtrsim (KN+M)$ is sufficient for exact recovery with high probability. 

\subsection{Image deblurring}\label{sec:image-deblurring}
In this section, we showcase the result of a synthetic experiment on image deblurring using random masks. 
We select three microscopic $150 \times 150$ images of human blood cells each of which is blurred using the same $10 \times 10$ ($M = 100$) Gaussian blur of variance 7. The original and blurred images are shown in the first and second row of Figure \ref{fig:RMI-Exps}, respectively. Each of the image is assumed to live in a known subspace\footnote{The known subspace of the original image is perhaps an unrealistic assumption in this case, however, a reasonably accurate estimate of the image subspace can be obtained from blurred (small blur) image by taking the multiscale structure of wavelets into account to recover the support of wavelet coefficients of the original image from blurred/smoothed out edges.}  of dimension $K = 3400$ spanned by the most significant wavelet coefficients. We mimic the random mask imaging setup  discussed in Section \ref{sec:RMI}, and pixelwise multiply each image with a $150 \times 150$ random $\pm 1$ mask. Given the observations on the detector array, we jointly deblur three ($N=3$) images using the proposed gradient descent algorithm.   The deblurred images are shown in the third row of Figure \ref{fig:RMI-Exps}.  The total relative mean squared error (MSE) of the three recovered images are $0.0184$, and the blur kernel is estimated within a relative MSE of $1.03 \times 10^{-4}$. 
\begin{figure}
	\centering
	\begin{tabular}{cc}
		\includegraphics[scale = 0.63, trim = 0.2cm 5cm 0 0.6cm,clip]{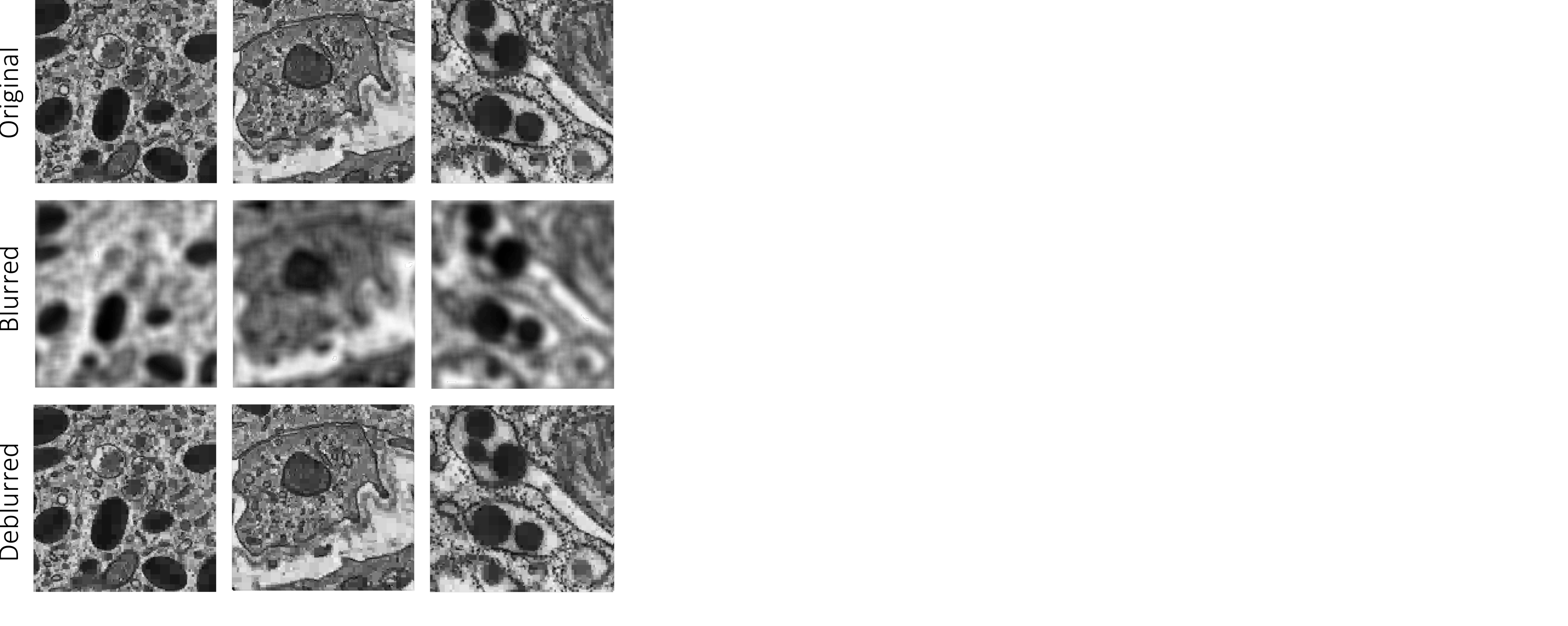}
	\end{tabular}
	\caption{\small\sl Random mask blind image deblurring via gradient descent. Three 150 $\times$ 150 blood cell images shown in the first row. Blurred images using same $10 \times 10$ Gaussian blur of variance 7 are shown in the second row.  Applying random masks on images before the unknown blurring (lens), and using gradient descent algorithm gives deblurred images in the third row. }
	\label{fig:RMI-Exps}
\end{figure} 

\subsection{Performance under noise and oversampling}

Noise performance of the algorithm is depicted in Figure \ref{fig:Stability}. Additive Gaussian noise $\hat{\ve}$ is added in the measurements as in \eqref{eq:measurements}. As before, we synthetically  generate $\vh_0$, and $\vx_0$ as Gaussian vectors, and $\mC$ is the subset of the columns of a DCT matrix. We plot (left) relative error (log scale) in \eqref{eq:relative-error} of the recovered vectors $\hat{\vh}$, and $\hat{\vx}$ averaged over hundred independent experiments vs. $
\text{SNR} := 10 \log_{10} \left(\|\vh_0\vx_0^*\|_F^2/\|\ve\|_2^2\right)$
, and (right) average relative error (log scale) vs. oversampling ratio := $L /(KN+M)$ under no noise. Oversampling ratio is a factor by which the number ($L$ in this case as $N=1$) of measurements exceed the number $K+M$ of unknowns. The left plot shows that the relative error degrades \textit{gracefully} by reducing SNR, and the right shows that relative error almost reduces to zero when the oversampling ration exceeds 2.1. 
\begin{figure*}
	\centering
	\begin{tabular}{ccc}
		& \includegraphics[scale = 0.27, trim = 0cm 5cm 0 5cm,clip]{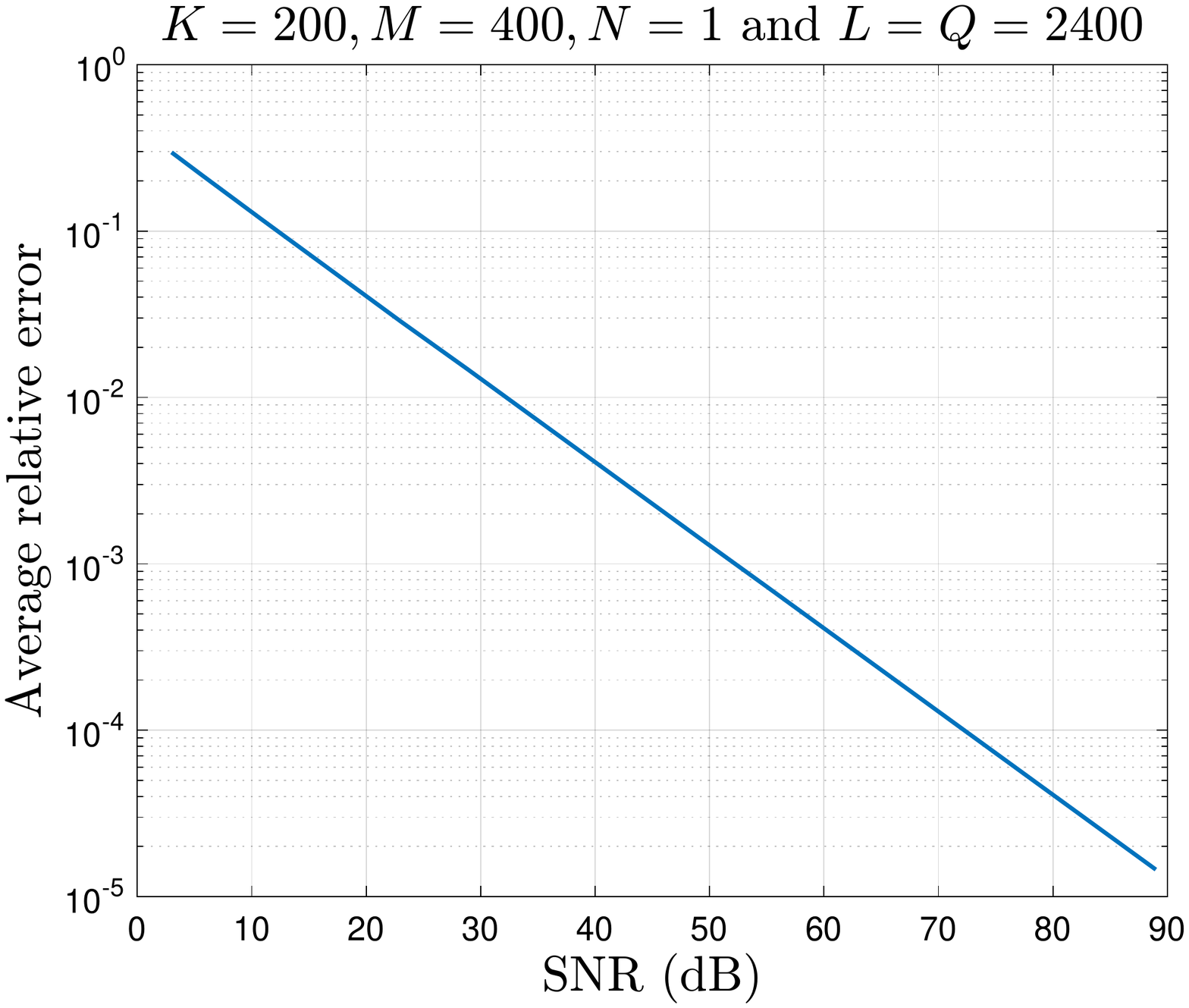}
		&\includegraphics[scale = 0.27, trim = 0cm 5cm 0 5cm,clip]{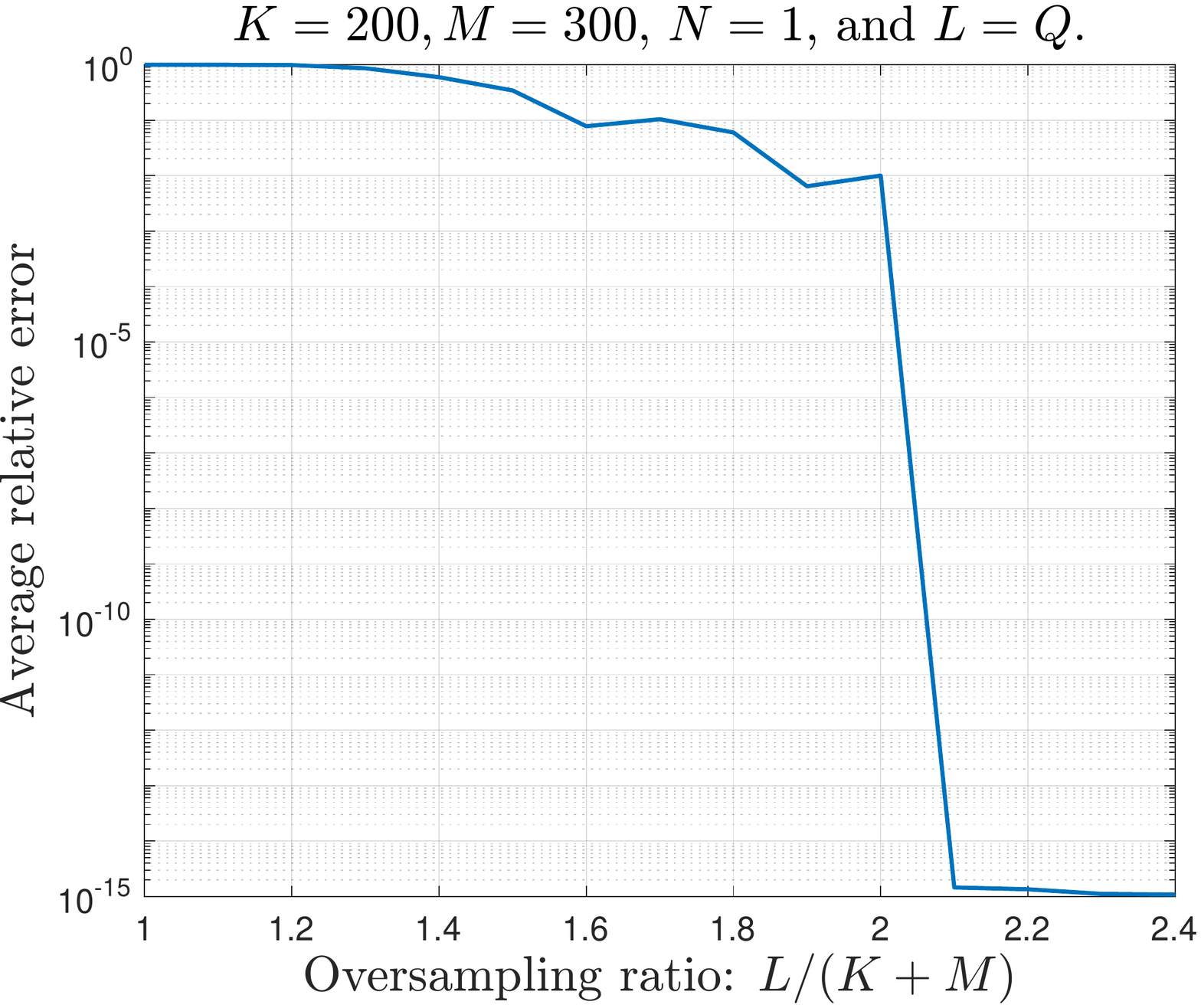}
	\end{tabular}
	\caption{\small\sl Performance in the presence of additive meaurement noise (left). Number of samples  vs. the relative error (right). }
	\label{fig:Stability}
\end{figure*} 

\section{Proofs} 

\subsection{Preliminaries}
Recall the function $F(\vh,\vx)$ defined in \eqref{eq:F-def}, and the gradients $\gfh$, and $\gfx$ in \eqref{eq:gFh-gFx-def}. By linearity, $\nabla \tilde{F}_{\vh} = \nabla F_{\vh} + \nabla G_{\vh}$, and similarly, $\nabla \tilde{F}_{\vx} = \nabla F_{\vx} + \nabla G_{\vx}$. Using the definitions of $F$, and $G$ in \eqref{eq:F-def}, and \eqref{eq:G-def}, the gradients w.r.t. $\vh$, and $\vx$ are 
\begin{align}\label{eq:gF-def}
\nabla F_{\vh}  &= \setA^*(\setA(\vh\vx^*-\vh_0\vx_0^*)-\ve) \vx, \notag\\
 \nabla F_{\vx} &= [\setA^*(\setA(\vh\vx^*-\vh_0\vx_0^*)-\ve)]^*\vh,
\end{align}
\begin{align}\label{eq:ghG-def}
\nabla G_{\vx} &= \frac{\rho}{2d}\Bigg[ G_0^\prime \left( \frac{\|\vx\|_2^2}{2d}\right) \vx \notag \\
& + \frac{QN}{4\nu^2} \sum_{q=1}^Q\sum_{n=1}^N G_0^\prime \left( \frac{QN |\vc_{q,n}^*\vx|^2}{8d \nu^2}\right) \vc_{q,n}\vc_{q,n}^* \vx\Bigg],
\end{align}
and 
\begin{align}\label{eq:gxG-def}
\nabla G_{\vh} = \frac{\rho}{2d}\Bigg[ G_0^\prime \left( \frac{\|\vh\|_2^2}{2d}\right) \vh + \frac{L}{4\mu^2} \sum_{\ell=1}^L G_0^\prime \left( \frac{L |\vf_{\ell}^*\vh|^2}{8d \mu^2}\right) \vf_{\ell}\vf_{\ell}^* \vh\Bigg].
\end{align}
We have the following useful lower and upper bounds on $F(\vh,\vx)$ using the triangle inequality
\begin{align*}
& - 2 \|\setA^*(\ve)\|_{2 \rightarrow 2}\|\vh\vx^*-\vh_0\vx_0^*\|_* \leq F(\vh,\vx)-\|\ve\|_2^2 \notag\\
& - \|\setA(\vh\vx^*-\vh_0\vx_0^*)\|_2^2 \leq   2 \|\setA^*(\ve)\|_{2 \rightarrow 2}\|\vh\vx^*-\vh_0\vx_0^*\|_*.
\end{align*}
For brevity, we set $\|\vh\vx^*-\vh_0\vx_0^*\|_F : = \delta d_0$. Since $\vh\vx^*-\vh_0\vx_0^*$ is a rank-2 matrix, we have  $\|\vh\vx^*-\vh_0\vx_0^*\|_* \leq \sqrt{2} \|\vh\vx^*-\vh_0\vx_0^*\|_F = \sqrt{2}\delta d_0$, where we used $\|\vh_0\|_2 = \|\vx_0\|_2 = \sqrt{d_0}$ from Lemma \ref{lem:local-regulaity-Del-norm-bounds}. Invoking Lemma \ref{lem:noise-stability}, and Lemma \ref{lem:local-RIP}  with $\xi = \frac{1}{4}$, we have
\begin{align}\label{eq:F(h,x)-lower-upper-bound}
\|\ve\|_2^2 +\frac{3}{4} \delta^2 d_0^2 - \frac{\varepsilon\delta d_0}{5} \leq F(\vh,\vx) \leq \|\ve\|_2^2 +\frac{5}{4} \delta^2 d_0^2 + \frac{\varepsilon\delta d_0}{5}.
\end{align}
In the analysis later, it will be convenient to uniquely decompose $\vh$, and $\vx$ as $\vh = \alpha_1\vh_0 + \tilde{\vh}$, and $\vx = \alpha_2 \vx_0 + \tilde{\vx}$, where $\tilde{\vh} \perp \vh$, $\tilde{\vx} \perp \vx$, and $
\alpha_1 = \frac{\vh_0^*\vh}{d_0},\ \text{and}\  \alpha_2 = \frac{\vx_0^*\vx}{d_0}.$
We also define specific vectors $\Delta\vh$, and $\Delta\vx$ that will repeatedly arise in the technical discussion later. 
\begin{align}\label{eq:Deltah-Deltax}
&\Delta \vh = \vh - \alpha \vh_0, \ \text{and} \ \Delta \vx = \vx - \bar{\alpha}^{-1} \vx_0, \\
&\text{where} \ \alpha(\vh,\vx) = \begin{cases}
(1-\delta_0)\alpha_1, & \text{if}~ \|\vh\|_2 \geq \|\vx\|_2\\
\frac{1}{(1-\delta_0)\bar{\alpha}_2}, & \text{if} ~ \|\vh\|_2 < \|\vx\|_2\notag
\end{cases}
\end{align}
with $\delta_0 := \tfrac{\delta}{10}$ --- the choice of $\alpha$ is mainly required for the proof of Lemma \ref{lem:local-regularity-G}. Note that 
\begin{align}\label{eq:difference-expansion}
\vh\vx^*-\vh_0\vx_0^* = (\alpha_1\bar{\alpha}_2-1) \vh_0\vx_0^*+\bar{\alpha}_2 \tilde{\vh}\vx_0^* + \alpha_1 \vh_0 \tilde{\vx}^* + \tilde{\vh}\tilde{\vx}^*.
\end{align}
The lemma below gives bounds on some relevant norms that will be useful in the proofs later. 
\begin{lem}\label{lem:local-regulaity-Del-norm-bounds}
	Recall that $\|\vh_0\|_2 = \|\vx_0\|_2 = \sqrt{d_0}$. If $\delta: = \frac{\|\vh\vx^*-\vh_0\vx_0^*\|_F}{d_0} < 1$ then for all $(\vh,\vx) \in \setN_{d_0}$, we have the following useful bounds $|\alpha_1| < 2$, $|\alpha_2| < 2$, and $|\alpha_1\bar{\alpha}_2-1| \leq \delta$.  For all $(\vh,\vx) \in \setN_{d_0} \cap \setN_{\varepsilon}$ with $\varepsilon \leq 1/15$, there holds $\|\Delta \vh\|_2^2 \leq 6.1\delta^2d_0$, $\|\Delta \vx\|_2^2 \leq 6.1\delta^2d_0$, and $\|\Delta \vh\|_2^2 \|\Delta \vx\|_2^2 \leq 8.4 \delta^4 d_0^2$. Moreover, if we assume $(\vh,\vx) \in \setN_{\mu} \cap \setN_{\nu}$, we have  $\sqrt{L}\|\mF_M\Delta\vh\|_\infty \leq 6 \mu \sqrt{d_0}$, and $\sqrt{QN}\|\mC ^{\otimes N}(\Delta \vx) \|_\infty \leq 6\nu \sqrt{d_0}$. 
\end{lem}
\noindent Proof of this lemma is provided in Appendix \ref{sec:local-regulaity-Del-norm-bounds}.
\subsection{Proof Strategy} 
Albeit important differences, the main template of the proof of Theorem \ref{thm:convergence} is similar to \cite{li2016rapid}.  To avoid overlap with \cite{li2016rapid}, we refer the reader on multiple points in the exposition below to consult \cite{li2016rapid}  for some intermediate results that are already proved in \cite{li2016rapid}. To facilitate this, the notation is kept very similar to \cite{li2016rapid}. 

The main lemmas required to prove Theorem \ref{thm:convergence} fall under one of the four key conditions \cite{li2016rapid} stated in Section \ref{sec:Key-conditions}. The lemmas under the important local RIP, and noise robustness conditions require completely different proofs compared to \cite{li2016rapid}, due to the new structured random linear map $\setA$ in \eqref{eq:linear-map} in this paper. The limited randomness in $\setA$ calls for a more intricate chaining argument to handle probabilistic deviation results required to prove the local-RIP.

We begin by stating a main lemma showing that the iterates of gradient descent algorithm decrease the objective function $\tf(\vh,\vx)$. Let
$\vz_t := (\vu_t,\vv_t) \in \C^{M+KN}$ be the $t$th iterate of Algorithm \ref{algo:gradient-descent} that are close enough to the truth $(\vh_0, \vx_0)$, i.e., $\vz_t \in \setN_{\varepsilon}$ and have a small enough loss $\tf(\vz_t) := \tf(\vu_t,\vv_t)$, that is, $\vz_t \in \setN_{\tf}$, where 
\begin{align}\label{eq:setNF}
\setN_{\tilde{F}} : = \bigg\{(\vh,\vx) \ | \ \tilde{F}(\vh,\vx) \leq \frac{1}{3} \varepsilon^2d_0^2 + \|\ve\|_2^2\bigg\}
\end{align}
is a sub-level set of the non-convex loss function $\tilde{F}(\vh,\vx)$. We show that with current iterate  $\vz_t \in \setN_{\varepsilon}\cap\setN_{\tf}$, the next iterate $\vz_{t+1}$ of the Algorithm \ref{algo:gradient-descent} belongs to the same neighborhood, and does not increase the loss function. 

\begin{lem}[Lemma 5.8 in \cite{li2016rapid}]\label{lem:main-lemma}
	Let the step size $\eta \leq 1/C_L$, $\vz_t := (\vu_t,\vv_t) \in \C^{M+KN}$, and $C_L$ be the constant defined in \eqref{eq:CL-def}. Then, as long as $\vz_t \in \setN_{\varepsilon} \cap \setN_{\tf}$, we have $\vz_{t+1} \in \setN_{\varepsilon} \cap \setN_{\tf}$, and 
	\begin{align*}
	\tf(\vz_{t+1}) \leq \tf(\vz_t) - \eta \| \gtf (\vz_t)\|_2^2. 
	\end{align*}
\end{lem} 
\begin{proof}
	The proof is exactly as the proof of Lemma 5.8 in \cite{li2016rapid}, and relies on the smoothness condition in Lemma \ref{lem:smoothness-CL} below. 
\end{proof}

\subsection{Key conditions}\label{sec:Key-conditions}
We now state the lemmas under four key conditions required to prove Lemma \ref{lem:main-lemma}, and the theorems. 
\subsubsection{Local smoothness}
\begin{lem}\label{lem:smoothness-CL} For any $\vz := (\vh,\vx)$ and $\vw := (\vu,\vv)$ such that $\vz$, $\vz+\vw \in \setN_{\varepsilon}\cap \setN_{\tilde{F}}$, there holds 
	\begin{align*}
	&\|\nabla \tilde{F}(\vz+\vw) - \nabla \tilde{F}(\vz)\|_2 \leq C_L\|\vw\|_2 ~ \text{with} ~\\
	& C_L \leq \sqrt{2}d_0 \left[ 10 \|\setA\|_{2 \rightarrow 2}^2 + \frac{\rho}{d^2} \left(5 + \frac{3L}{2\mu^2}+ \frac{3QN}{2\nu^2}\right)\right],
	\end{align*}
	where $\rho \geq d^2 + 2\|\ve\|_2^2$, and $\|\setA\|_{2 \rightarrow 2} \leq c_\alpha \sqrt{K \log (LN)}$ holds with probability at least $1-\setO((LN)^{-\alpha})$. In particular, $
	QN = \setO(\mu^2\nu_{\max}^2 KN^2 + \nu^2 M)\log^4(LN), \ \text{and} \ \|\ve\|^2 = \setO(\sigma^2d_0^2).$
	Therefore, $C_L$ can be simplified to 
	\begin{align}\label{eq:CL-def}
	C_L = \setO\left(d_0(1+\sigma^2)\big(\mu^2\nu_{\max}^2 KN^2+ L \big)\log^4(LN)\right)
	\end{align}
	by choosing $\rho\approx d^2 + 2 \|\ve\|_2^2$. 
\end{lem}
\noindent Proof of this lemma is provided in Appendix \ref{sec:smoothness-CL}.

\subsubsection{Local regularity}

Recall that from Lemma \ref{lem:main-lemma}, we set the step size $\eta \leq 1/C_L$. With the step size $\eta$ characterized through the constant $C_L$  in \eqref{eq:CL-def}, the next task is to find a lower bound on the norm of the gradient $\|\gtf(\vz_t)\|^2_2$ in Lemma \ref{lem:main-lemma}. Following lemma provides this bound. 
\begin{lem}[Lemma 5.18 in \cite{li2016rapid}]\label{lem:local-regularity}
	Let $\tf(\vh,\vx)$ be as defined in \eqref{eq:tF-def} and $\gtf(\vh,\vx) : = (\gtfh,\gtfx) \in \C^{M+KN}$. Then there exists a regularity constant $\omega = d_0/5000 >0$ such that 
	\begin{align*}
	\|\gtf(\vh,\vx)\|_2^2 \geq \omega \left[ \tf(\vh,\vx) -c\right]_+
	\end{align*}
	for any $(\vh,\vx) \in \setN_{d_0} \cap \setN_{\mu} \cap \setN_{\nu} \cap \setN_{\varepsilon}$, where $c= \|\ve\|_2^2+1700 \|\setA^*(\ve)\|_{2\rightarrow 2}^2$, and $\rho \geq d^2+\|\ve\|_2^2$. 
\end{lem}
\begin{proof}
	The proof hinges on following two conclusions 
	\begin{align*}
	&\Re{\<\gfh,\Delta\vh\>+\<\gfx,\Delta \vx\>} \geq \frac{\delta^2d_0^2}{8}-2\delta d_0 \|\setA^*(\ve)\|_{2\rightarrow 2},\notag\\
    &\Re{\<\ggh,\Delta\vh\>+\<\ggx,\Delta\vx\>} \geq \frac{\delta}{5}\sqrt{\rho G_0(\vh,\vx)};
	\end{align*}
	established in Lemma \ref{lem:local-regularity-F}, and \ref{lem:local-regularity-G} below.  Given the above two conditions the proof of this lemma reduces exactly to the proof of Lemma 5.18 in \cite{li2016rapid}. 
\end{proof}
\begin{lem}\label{lem:local-regularity-F}
	For any $(\vh,\vx) \in \setN_{d_0} \cap \setN_{\mu} \cap \setN_{\nu}\cap\setN_{\varepsilon}$ with $\varepsilon \leq \frac{1}{15}$, uniformly:
	\[
	\Re{\<\gfh,\Delta\vh\>+\<\gfx,\Delta\vx\>} \geq \frac{\delta^2d_0^2}{8}-2\delta d_0 \|\setA^*(\ve)\|_{2 \rightarrow 2},
	\]
	with probability at least
	\begin{align}\label{eq:probability-regularity}
	1-2\exp\left(-c\delta^2\frac{QN}{\mu^2\nu^2}\right)
	\end{align}
	provided 
	\begin{align}\label{eq:sample-complexity-regularity}
	QN \geq \frac{c}{\delta^2}\max(\mu^2\nu_{\max}^2 KN^2+\nu^2 M)\log^4(LN).
	\end{align}
\end{lem}
\begin{proof}
	Observe that $\Re{\<\nabla F_{\vh}, \Delta \vh\>+\<\nabla F_{\vx}, \Delta \vx\>} = \Re{\<\nabla F_{\vh}, \Delta \vh\>+\overline{\<\nabla F_{\vx}, \Delta \vx\>}}$. Using the gradients derived earlier in \eqref{eq:gF-def}, we have 
	\begin{align*} 
	&\<\nabla F_{\vh}, \Delta \vh\> = \< \setA^*(\setA(\vh\vx^*-\vh_0\vx_0^*)-\ve), \Delta\vh\vx^*\>, \ \text{and} \  \\
	& \overline{\<\nabla F_{\vx}, \Delta \vx\>} = \< \setA^*(\setA(\vh\vx^*-\vh_0\vx_0^*)-\ve), \vh\Delta \vx^*\>,
	\end{align*}
	and hence
	\begin{align}\label{eq:local-reg-F-expansion}
	&\<\nabla F_{\vh}, \Delta \vh\> + \overline{\< \nabla F_{\vx}, \Delta \vx\>} = - \<\setA^*(\ve), \Delta\vh\vx^*+\vh\Delta\vx^*\> \notag\\
	&\qquad + \<\setA(\vh\vx^*-\vh_0\vx_0^*),\setA(\Delta\vh\vx^*+\vh\Delta\vx^*)\>.
	\end{align}
	Using triangle inequality, $\|\Delta \vh\vx^*+ \vh\Delta\vx^*\|_F \geq \|\vh\vx^*-\vh_0\vx_0^*\|_F - \|\Delta\vh\Delta\vx^*\|_F$. Lemma \ref{lem:local-regulaity-Del-norm-bounds} shows $\|\Delta \vh\Delta \vx^* \|_F \leq 2.9\delta^2 d_0$ when $\delta \leq \varepsilon \leq  1/15$. This implies that $\|\Delta \vh\vx^*+ \vh\Delta\vx^*\|_F  \geq \delta d_0 - 2.9\delta^2 d_0 \geq 0.8 \delta d_0.$
	In a similar manner the upper bound can be established leading us to
	\begin{align}\label{eq:Dhx+hDx-fronorm}
	0.8 \delta d_0\leq \|\Delta \vh\vx^*+ \vh\Delta\vx^*\|_F \leq 1.2 \delta d_0
	\end{align}
	that holds when $\delta \leq \varepsilon \leq \frac{1}{15}$. 
	Using $\xi = \frac{1}{4}$, and $\delta \leq \varepsilon$ in Lemma \ref{lem:local-RIP}, and \ref{lem:local-Delh-Delx-RIP} below give the following conclusions
	\begin{align*}
	&\|\setA(\vh\vx^*-\vh_0\vx_0^*)\|^2_2 \geq \|\vh\vx^*-\vh_0\vx_0^*\|^2_F - \tfrac{1}{4}\delta^2d_0^2 = \tfrac{3}{4} \delta^2 d^2_0~ \text{and} ~ \\
	&\|\setA(\Delta \vh\vx^*+ \vh\Delta\vx^*)\|^2_2 \geq  \|\Delta \vh\vx^*+\vh\Delta\vx^*\|^2_{F} - \tfrac{1}{4}\delta^2d_0^2 \geq \tfrac{1}{4}\delta^2d_0^2,
	\end{align*}
	each holding with probability at least \eqref{eq:probability-regularity} under the sample complexity bound \eqref{eq:sample-complexity-regularity}. 	In addition, we also have 
	\begin{align*}
	\<\setA^*(\ve), \Delta\vh\vx^*+\vh\Delta\vx^*\> &\leq \sqrt{2} \|\setA^*(\ve)\|_{2 \rightarrow 2}\|\Delta\vh\vx^*+\vh\Delta\vx^*\|_F\\
	& \leq  2\delta d_0 \|\setA^*(\ve)\|_{2\rightarrow 2}. 
	\end{align*}
	Employing these bounds in \eqref{eq:local-reg-F-expansion}, we obtain the desired bound.
\end{proof}
\begin{lem}\label{lem:local-regularity-G}
	For any $(\vh,\vx) \in \setN_{\mu} \cap \setN_{\nu} \cap \setN_{d_0} \cap \setN_{\varepsilon}$ with $\varepsilon \leq 1/15$, and $0.9d_0 \leq d \leq 1.1d_0$, the following inequality holds uniformly $
	\Re{\<\nabla G_{\vh}, \Delta \vh \>+\<\nabla G_{\vx}, \Delta \vx \>}  \geq \frac{\delta}{5} \sqrt{\rho G_0(\vh,\vx)},$ where $\rho \geq d^2+ 2\|\ve\|_2^2$. 
\end{lem}
\noindent For the proof, see  Appendix \ref{sec:local-regularity-G} below. 
\subsubsection{Local RIP}	

	The following two lemmas state the  local restricted isometry property used above in the proof of Lemma \ref{lem:local-regularity-F}. 
\begin{lem}\label{lem:local-RIP}
	For all $(\vh,\vx) \in \setN_{d_0} \cap \setN_{\mu}\cap \setN_{\nu}$ such that $\|\vh\vx^*-\vh_0\vx_0^*\|_F = \delta d_0$, where $ \delta < 1$, the following local restricted isometry property:
	\begin{align}\label{eq:Local-RIP}
	\left| \|\setA(\vh\vx^*-\vh_0\vx_0^*)\|_2^2 - \|\vh\vx^*-\vh_0\vx_0^*\|_{F}^2 \right| \leq  \xi \delta^2 d_0^2
	\end{align}
	 holds for a $\xi \in (0,1)$ with probability at least $1-2\exp(-c\xi^2\delta^2 QN / \mu^2\nu^2)$
	  whenever 
	\begin{align}\label{eq:sample-complexity}
	QN \geq \frac{c}{\xi^2\delta^2}\big( \mu^2\nu_{\max}^2 KN^2 + \nu^2 M\big) \log^4(LN).
	\end{align}
\end{lem}
\begin{lem}\label{lem:local-Delh-Delx-RIP}
	For all $(\vh,\vx) \in \setN_{d_0} \cap \setN_{\mu}\cap \setN_{\nu}\cap \setN_{\varepsilon}$ such that  $0.8\delta d_0 \leq \|\Delta\vh\vx^*+\vh\Delta\vx^*\|_F \leq 1.2\delta d_0$, where $\delta \leq \varepsilon \leq 1/15$,  the following local restricted isometry property:
	\begin{align*}
&  \left|\|\setA(\Delta \vh\vx^*+ \vh \Delta \vx^*)\|_2^2 - \|\Delta \vh\vx^*+ \vh \Delta \vx^*\|_{F}^2 \right| \leq \xi\delta^2d_0^2 
	\end{align*}
	holds for a $\xi \in (0,1)$ with probability at least $1-2\exp(-c\xi^2\delta^2 QN / \mu^2\nu^2)$ whenever \eqref{eq:sample-complexity} holds.
\end{lem}
\noindent Proof of both these lemmas is the main technical contribution above \cite{li2016rapid}, and is provided in Section \ref{sec:local-RIP}, and Appendix \ref{sec:local-Delh-Delx-RIP}.  The usual probability concentration, and union bound argument \cite{baraniuk2008simple} to prove RIP is not sufficient due to the limited/structured randomness in $\setA$.   We therefore use the result in \cite{krahmer2014suprema} based on generic chaining, which abstracts out the entire signal space from a coarse to a fine scale and employs a more efficient use of union bound at each scale after probability concentration.

\subsubsection{Noise robustness}

Finally, we give the noise robustness result below that gives a bound on the noise term $\|\setA^*(\ve)\|_{2 \rightarrow 2}$ appearing in Lemma \ref{lem:local-regularity}.
\begin{lem}\label{lem:noise-stability}
	Fix $\alpha \geq 1$. For the linear map $\setA$ defined in \eqref{eq:linear-map}, it holds $\|\setA\|_{2 \rightarrow 2} \leq c_\alpha\sqrt{ K\log (LN)}$ with probability at least $1-\setO((LN)^{-\alpha})$. Moreover, let $\ve \in \C^{LN}$ be additive noise introduced in \eqref{eq:measurements}, distributed as $\ve \sim \text{Normal}(\boldsymbol{0},\frac{\sigma^2d_0^2}{2LN} \mI_{LN}) + \iota \text{Normal}(\boldsymbol{0},\frac{\sigma^2d_0^2}{2LN} \mI_{LN})$, there holds 
	\begin{align}\label{eq:A(e)-bound}
	\|\setA^*(\ve) \|_{2 \rightarrow 2} \leq \tfrac{2\varepsilon}{50} d_0
	\end{align}
	with probability at least $1-\setO((LN)^{-\alpha})$ whenever $LN \geq  \frac{\sigma^2}{\varepsilon^2} c_\alpha^\prime (M,KN\log(LN))\log(LN)$, where $c_\alpha$, and $c^\prime_\alpha$ are absolute constants depending on the free parameter $\alpha \geq 1$. 
\end{lem}
\noindent Please refer to Appendix \ref{sec:noise-stability} for the proof of this lemma.

\subsection{Proof of Theorem \ref{thm:convergence}}\label{sec:convergence}
Given all the intermediate results above, we are in a position to prove Theorem \ref{thm:convergence} below. 
\begin{proof}
We denote by $\vz_t = (\vu_t,\vv_t)$, the iterates of gradient descent algorithm, and $\delta(\vz_t) = \|\vu_t\vv_t^*-\vh_0\vx_0^*\|_F/d_0$. At the initial guess $\vz_0 :=(\vu_0,\vv_0) \in \tfrac{1}{\sqrt{3}} \setN_{d_0} \cap \tfrac{1}{\sqrt{3}} \setN_{\mu} \cap \tfrac{1}{\sqrt{3}} \setN_{\nu} \cap \setN_{\frac{2}{5}\varepsilon}$, it is easy to verify using the definitions of  $\setN_{d_0}$, $\setN_{\mu}$, and $\setN_{\nu}$ that $G(\vu_0,\vv_0) = 0$. For example, $(\vu_0,\vv_0) \in \frac{1}{\sqrt{3}}\setN_{\nu}$ gives 
\begin{align*}
\frac{QN|\vc_{q,n}^*\vx|^2}{8d\nu^2} \leq \frac{QN}{8d\nu^2} \cdot \frac{16d_0\nu^2}{3QN} = \frac{2d_0}{3 d} < 1,
\end{align*}
which immediately implies that the fourth term of $G(\vu_0,\vv_0)$ in \eqref{eq:G-def} is zero. Similar calculations show that all other terms in $G(\vu_0,\vv_0)$ are zero. The remaining proof is the exact repetition of the proof of Theorem \ref{thm:convergence} in \cite{li2016rapid}, and uses Lemma \ref{lem:main-lemma}, and \ref{lem:local-regularity} to produce 
\begin{align*}
&\|\vu_{t+1}\vv_{t+1}^*-\vh_0\vx_0^*\|_F \leq \tfrac{2}{3} (1-\eta\omega)^{(t+1)/2} \varepsilon d_0 + 50 \|\setA^*(\ve)\|_{2 \rightarrow 2},
\end{align*}
where $\eta$ is the step size in the gradient decent algorithm and satisfies $\eta \leq 1/C_L$ for a constant $C_L$ defined in Lemma \ref{lem:smoothness-CL}, and the constant $\omega$ is characterized in Lemma \ref{lem:local-regularity}. 
\end{proof}

Due to space constraints, the proofs of the remaining lemmas are moved to the appendices, which include the proof of the key lemmas on local RIP that constitute our main technical contribution. 

\section{Conclusion}

We studied the blind deconvolution problem under a practically relevant model of modulated input signals. We discussed several applications to motivate the problem, and presented some recovery guarantees. We believe that a better proof technique may show that the regularization term $G(\vh,\vx)$ is not required. Moreover, we also conjecture that the approximate recovery guarantees may also be improved to exact recovery. 
\bibliographystyle{IEEEtran}
\bibliography{IEEEabrv,ModBD}
\newpage.
\appendix
We now complete the proofs of lemmas not covered in the main body of the paper due to space constraints. We begin with the local RIP proofs that constitute our main technical contribution, and rely on the chaining arguments. 

\subsection{Proof of Lemma \ref{lem:local-RIP}}\label{sec:local-RIP}

Recall that $\setA$ in \eqref{eq:linear-map} maps the unknowns $\vh_0\vx_0^*$ to the noiseless convolution measurements in the Fourier domain. Using the isometry of the DFT matrix $\mF$ (an $L \times L$ normalized DFT matrix), we have 
\begin{align*}
&\setA(\vh_0\vx_0^*) = \\
&\sum_{n=1}^N \|\vh_0\circledast \mR_n\mC\vx_{0,n}\|_2^2= \sum_{n=1}^N\|\text{circ}(\vh_0)\text{diag}(\mC\vx_{0,n})\vr_n\|_2^2,
\end{align*}
where as before $\vx_0 = \text{vec}([\vx_{0,n}])$, and $\mR_n = \text{diag}(\vr_n)$.  Similar calculation shows that
\begin{align}\label{eq:2nd-order-chaos-process}
&\|\setA(\vh\vx^*-\vh_0\vx_0^*)\|_2^2 \notag\\
&= \sum_{n=1}^N \left\|\text{circ}(\vh)\text{diag}(\mC\vx_n)\vr_n - \text{circ}(\vh_0) \text{diag}(\mC\vx_{0,n})\vr_n\right\|_2^2\notag\\
& = \|(\mH_{\vh}\mX_{\vx}-\mH_{\vh_0}\mX_{\vx_0})\vr\|_F^2,
\end{align}
where $\vr = \text{vec}([\vr_n])$, and the block diagonal matrices $\mH_{\vh} \in \C^{LN \times QN}$, and $\mX_{\vx} \in \C^{QN \times QN}$ are defined as
\begin{align}\label{eq:Hh-Xx-def}
\mH_{\vh}:= [\text{circ}(\vh)]^{\otimes N}, \ \text{and} \ \mX_{\vx} :=  [\text{diag}\left(\mC\vx_{n}\right)]^{\otimes N}. 
\end{align}
The empirical process in \eqref{eq:2nd-order-chaos-process} is known as 2nd order chaos process, where $\vr$  is a standard Rademacher $QN$-vector, and $(\mH_{\vh}\mX_{\vx}-\mH_{\vh_0}\mX_{\vx_0})$ is a deterministic matrix. The expected value of this random quantity is simply 
\begin{align}\label{eq:expected-value}
\E \|(\mH_{\vh}\mX_{\vx}-\mH_{\vh_0}\mX_{\vx_0})\vr\|_2^2 = \|(\mH_{\vh}\mX_{\vx}-\mH_{\vh_0}\mX_{\vx_0})\|_F^2.
\end{align}
Recall that $\circledast$ denotes $L$-point circular convolution, and therefore, $\text{circ}(\vh) = \mF^*\text{diag}(\hat{\vh})\mF_Q \in \C^{L \times Q}$, where $\hat{\vh} = \sqrt{L}\mF_M\vh$. Note a simple identity $ \|\mH_{\vh}\mX_{\vx}-\mH_{\vh_0}\mX_{\vx_0}\|_F^2 = \|\vh\vx^*-\vh_0\vx_0^*\|_F^2$; its proof follows from the couple of simple steps below
\begin{align}\label{eq:distance-vectors-to-matrices}
&\|\mH_{\vh}\mX_{\vx}-\mH_{\vh_0}\mX_{\vx_0}\|_F^2 = \notag\\
&\sum_{n} \|\mF^*\text{diag}(\hat{\vh})\mF_Q \text{diag}(\mC\vx_{n})- \mF^*\text{diag}(\hat{\vh}_0)\mF_Q\text{diag}(\mC\vx_{0,n})\|_F^2\notag\\
& = \sum_{n} \|\text{diag}(\hat{\vh})\mF_Q \text{diag}(\mC\vx_{n})- \text{diag}(\hat{\vh}_0)\mF_Q\text{diag}(\mC\vx_{0,n})\|_F^2\notag\\
& = \sum_{n} \| \mF_Q\odot( \hat{\vh} \vx_n^*- \hat{\vh}_0 \vx_{0,n}^*)\|_F^2 = \|\vh\vx^*-\vh_0\vx_0^*\|_F^2,
\end{align}
where the last two equalities follow from the fact the $\mC^*\mC = \mI_K$,  $(\mF_M)^*\mF_M = \mI_M$, and that the entries of the DFT matrix $\sqrt{L}\mF_Q$ have unit magnitude.  Define the sets $\setH$, and $\setX$ indexed by $\vh$, and $\vx$, respectively, as below
\begin{align}\label{eq:setX-setH}
\setH:= \{\mH_{\vh} \vert  \vh \in \setN_{d_0} \cap \setN_{\mu}\}, \ \setX: = \{ \mX_{\vx} \vert    \vx \in \setN_{d_0} \cap \setN_\nu\}. 
\end{align}
Using \eqref{eq:2nd-order-chaos-process}, \eqref{eq:expected-value}, and the identity \eqref{eq:distance-vectors-to-matrices}, the local-RIP over all $(\vh,\vx)$ such that $\|\vh\vx^*-\vh_0\vx_0^*\|_F = \|(\mH_{\vh}\mX_{\vx}-\mH_{\vh_0}\mX_{\vx_0})\|_F =  \delta d_0$  as stated in Lemma \ref{lem:local-RIP} can be restated as 
\begin{align*}
& \sup_{\mH_{\vh} \in \setH}\sup_{\mX_{\vx} \in \setX}\Big| \|(\mH_{\vh}\mX_{\vx}-\mH_{\vh_0}\mX_{\vx_0})\vr\|_2^2 - \notag \\
&\qquad\qquad  \E \|(\mH_{\vh}\mX_{\vx}-\mH_{\vh_0}\mX_{\vx_0})\vr\|_2^2 \Big| \leq \xi \delta^2 d_0^2
\end{align*}
holds with high probability for a $\xi \in (0,1)$. 

This bound above depends on the notion of geometrical complexity of both the sets $\setX$, and $\setH$. The definition of this complexity is subtle and is measured in terms of the Talagrand's $\gamma_2$-functional \cite{talagrand2005generic} for these sets relative to two different distance metrics. Given a set $\setS$, and a distance defined by a norm $\|\cdot\|$, the $\gamma_2$-functional quantifies that how well $\setS$ can be approximated at different scales. 

The $\gamma_2$-functional can be directly related to the rate at which the size of the best $\epsilon$-cover of the set $\setS$ grows as $\epsilon$ decreases. Although this is a purely geometric characteristics of set $\setS$, the $\gamma_2$-functional gives a tight bound on the supremum of a Gaussian process. For example, if $\mG$ is an $M_1 \times M_2$ random matrix whose entries are independent and distributed $\text{Normal}(0,1)$, then $\sup_{\mS \in \setS} \<\mS,\mG\> \sim \gamma_2(\setS,\|\cdot\|_F).$ 

Along with $\gamma_2$, the other geometrical quantities that appear in the final bound are the diameters $d_{2\rightarrow 2}(\setS) := \sup_{\mS \in \setS} \|\mS\|_{2 \rightarrow 2}$, and $d_{F}(\setS):= \sup_{\mS \in \setS} \|\mS\|_F$ of the set $\setS$ with respect to the matrix operator, and Frobenius (sum of squares) norms, respectively.

Since the random quantity $\|(\mH_{\vh}\mX_{\vx}-\mH_{\vh_0}\mX_{\vx_0})\vr\|_2^2$ is a second-order-chaos process, we now present a result \cite{krahmer2014suprema} that controls the deviation of a general second-order-chaos process from its mean in terms of the geometrical quantities introduced above.

\begin{thm}[Theorem 3.1 in \cite{krahmer2014suprema}]\label{thm:Mendelson}
	Let $\setS$ be a set of matrices, and $\vr$ be a random vector whose entries $r_j$ are  independent mean zero, variance 1, and $\alpha$-subgaussian random variables. Let $d_F(\setS)$, and $d_{2 \rightarrow 2}(\setS)$ denote the diameters of $\setS$ under $\|\cdot\|_F$, and $\|\cdot\|_{2 \rightarrow 2}$ norms. Set 
	\begin{align*}
	E &= \gamma_2(\setS,\|\cdot\|_{2 \rightarrow 2})(\gamma_2(\setS,\|\cdot\|_{2 \rightarrow 2})+d_{F}(\setS)) + d_{F}(\setS) d_{2\rightarrow 2}(\setS)\\
	V &= d_{2 \rightarrow 2}(\setS)(\gamma_2(\setS, \|\cdot\|_{2\rightarrow 2})+d_F(\setS)), \ \text{and} \ U = d_{2 \rightarrow 2}^2 (\setS).
	\end{align*}
	Then for $t \geq 0$, 
	\begin{align*}
	&\PP\big(\sup_{\mS \in \setS}\left| \|\mS\vr\|_2^2 - \E \|\mS\vr\|_2^2 \right| \geq c_1 E +t  \big) \leq \\
	& \qquad \qquad  2 \exp\Big( -c_2\min\Big\{ \tfrac{t^2}{V^2},\tfrac{t}{U}\Big\}\Big).
	\end{align*}
	The constants $c_1$, and $c_2$ depend only on $\alpha$. 
\end{thm}
The proof of the local restricted isometry property is an application of the above result. For a fixed $\mH_{\vh_0} \in \setH$, and $\mX_{\vx_0} \in \setX$, we start by defining the set of matrices as $\setS :=\{\mH_{\vh}\mX_{\vx}-\mH_{\vh_0}\mX_{\vx_0}| \mH_{\vh} \in \setH, \ \mX_{\vx} \in \setX \}$.
Recalling that $\text{circ}(\vh) = \mF^*\text{diag}(\hat{\vh}) \mF_Q$ is an $L \times Q$ circulant matrix, and again $\mF$ is an $L \times L$ normalized DFT matrix. Using \eqref{eq:setX-setH}, it is now easy to see that $
\sup_{\mH_{\vh} \in \setH}\|\mH_{\vh}\|_{2 \rightarrow 2} = \|\mF^*\text{diag}(\hat{\vh}) \mF_Q\|_{2 \rightarrow 2} = \|\text{diag}(\hat{\vh}) \mF_Q\|_{2 \rightarrow 2}  = \sqrt{L}\|\mF_M\vh\|_\infty  \leq 4\sqrt{d_0}\mu$, and likewise  $\|\mH_{\vh_0}\|_F = \mu\sqrt{d_0}$.  Similarly, for $\mX_{\vx} = [\text{diag}( \mC\vx_{n})]^{\otimes N}$, we have
\begin{align*}
\sup_{\mX_{\vx} \in \setX} &\|\mX_{\vx}\|_{2 \rightarrow 2} = \|\mC^{\otimes N}\vx\|_\infty \leq \tfrac{4\nu \sqrt{d_0}}{\sqrt{QN}}, \\
&\|\mX_{\vx_0}\|_{2 \rightarrow 2} = \|\mC^{\otimes N}\vx_{0}\|_\infty  = \tfrac{\nu \sqrt{d_0}}{\sqrt{QN}}.
\end{align*}
An upper bound on the diameter $d_{2 \rightarrow 2}(\setS)$ can then be obtained as 
\begin{align}\label{eq:d2-tight-bound}
& d^2_{2 \rightarrow 2}(\setS) = \sup_{\mH_{\vh}\in \setH} \sup_{\mX_{\vx}\in \setX} \|\mH_{\vh}\mX_{\vx} - \mH_{\vh_0}\mX_{\vx_0}\|^2_{2 \rightarrow 2} \leq \notag\\
& \left(\sup_{\mH_{\vh}\in \setH} \|\mH_{\vh}\|_{2 \rightarrow 2}\cdot \sup_{\mX_{\vx}\in \setH} \|\mX_{\vx}\|_{2 \rightarrow 2} + \|\mH_{\vh_0}\|_{2 \rightarrow 2}\|\mX_{\vx_0}\|_{2 \rightarrow 2}\right)\notag \\
&\leq \frac{1}{\sqrt{QN}}\left(4\mu\sqrt{d_0} \cdot 4\nu\sqrt{d_0}+\mu\sqrt{d_0}\cdot \nu \sqrt{d_0}\right) = \frac{17\mu\nu d_0}{\sqrt{QN}}
\end{align}

Since we only consider all $(\vh,\vx)$ such that $\|\mH_{\vh}\mX_{\vx}-\mH_{\vh_0}\mX_{\vx_0}\|_F = \delta d_0$, the Frobenius diameter is then simply 
\begin{align}\label{eq:dF-bound}
& d_F(\setS) = \sup_{\mH_{\vh} \in \setH} \sup_{\mX_{\vx} \in \setX} \|\mH_{\vh}\mX_{\vx}-\mH_{\vh_0}\mX_{\vx_0}\|_F   = \delta d_0.
\end{align}

The $\gamma_2$-functional can be directly related to the complexity of the space under consideration. To make this precise, we need to introduce a covering set. A set $\setC$ is an $\epsilon$-cover of the set $\setS$ in the norm $\|\cdot\|$ if every point in the (infinite) set is within an $\epsilon$ of the finite set $\setC$:
\begin{align*}
\sup_{\mS\in \setS} \sup_{\mC \in \setC} \|\mC-\mS\| \leq \epsilon.
\end{align*}
The covering number $N(\setS, \|\cdot\|, \epsilon)$ is the size of the smallest $\epsilon$-cover of $\setS$. 

We can bound the $\gamma_2$-functional in terms of covering numbers using Dudley's integral \cite{dudley1967sizes,talagrand2005generic}
\begin{align}\label{eq:Dudley-Integral}
\gamma_2(\setS,\|\cdot\|_{2 \rightarrow 2}) \leq c\int_0^{d_{2\rightarrow 2}(\setS)} N(\setS,\|\cdot\|_{2 \rightarrow 2},\epsilon)d\epsilon,
\end{align}
where $c$ is a known constant, and $d_{2 \rightarrow 2}(\setS)$ is the diameter of $\setS$ in the operator norm $\|\cdot\|_{2 \rightarrow 2}$. 
The distance between $\mH_{\tilde{\vh}}\mX_{\tilde{\vx}}-\mH_{\vh_0}\mX_{\vx_0}  \in \setC \subseteq \setS$, and $\mH_{\vh}\mX_{\vx}-\mH_{\vh_0}\mX_{\vx_0} \in \setS$ is  
\begin{align}\label{eq:distance}
& \|(\mH_{\tilde{\vh}}\mX_{\tilde{\vx}}-\mH_{\vh_0}\mX_{\vx_0}) - (\mH_{\vh}\mX_{\vx}-\mH_{\vh_0}\mX_{\vx_0})\|_{2 \rightarrow 2} \notag \\
& = \|\mH_{\tilde{\vh}}\mX_{\tilde{\vx}}-\mH_{\vh}\mX_{\vx}\|_{2 \rightarrow 2}\notag \\
&\leq \|\mH_{\tilde{\vh}}\|_{2 \rightarrow 2}\|\mX_{\tilde{\vx}}-\mX_{\vx}\|_{2 \rightarrow 2} + \|\mX_{\vx}\|_{2 \rightarrow 2}\|\mH_{\tilde{\vh}}-\mH_{\vh}\|_{2 \rightarrow 2} = \notag\\
&\sqrt{L} \left(\|\mF_M \tilde{\vh}\|_\infty \|\mC^{\otimes N}( \tilde{\vx}-\vx)\|_\infty +  \|\mC^{\otimes N}\vx\|_\infty\|\mF_M (\tilde{\vh}-\vh)\|_\infty\right)\notag \\
&\leq \frac{4\mu\sqrt{d_0}}{\sqrt{Q}} \|\tilde{\vx}-\vx\|_c + \frac{4\nu \sqrt{d_0} }{\sqrt{QN}} \|\tilde{\vh}-\vh\|_f.
\end{align}
The last inequality follows from the fact that $\mX_{\tilde{\vx}} \in \setX$, $\mH_{\vh} \in \setH$ implying that $\tilde{\vx} \in \set{N}_{\mu} \cap \setN_{d_0}$, and $\vh \in \setN_{\nu} \cap \setN_{d_0}$. From the distance measure in \eqref{eq:distance}, it is clear that setting the following norms to 
\begin{align}\label{eq:exotic-norms}
&\|\tilde{\vx}-\vx\|_{c}:=\sqrt{Q}\|\mC^{\otimes N}(\tilde{\vx}-\vx)\|_\infty \leq \frac{\epsilon}{2} \cdot \frac{\sqrt{Q}}{4\mu\sqrt{d_0}}, \notag\\
&\|\tilde{\vh}-\vh\|_{f} := \sqrt{L} \|\mF_M(\tilde{\vh}-\vh)\|_\infty \leq \frac{\epsilon}{2}\cdot \frac{\sqrt{QN}}{4\nu \sqrt{d_0}}
\end{align}
gives $\|\mH_{\tilde{\vh}}\mX_{\tilde{\vx}}-\mH_{\vh}\mX_{\vx}\|_{2 \rightarrow 2} \leq \epsilon.$ Precisely, if for every point $\mH_{\tilde{\vh}}\mX_{\tilde{\vx}}-\mH_{\vh_0}\mX_{\vx_0} \in \setS$, where $\mH_{\vh} \in \setH$, and $\mX_{\vx} \in \setX$, there exists an $\mX_{\tilde{\vx}} \in \setC_{\setX}$ such that $\|\mX_{\vx}- \mX_{\tilde{\vx}}\|_{2 \rightarrow 2} =  \|[\mC(\vx_n-\tilde{\vx}_n)]\|_{\infty} \leq \epsilon\sqrt{Q}/8\mu\sqrt{d_0}$ ($\setC_{\setX}$ is an $\epsilon\sqrt{Q}/8\mu\sqrt{d_0}$-cover of $\setX$ in $\|\cdot\|_c$ norm), and an $\mH_{\tilde{\vh}} \in \setC_{\setH}$ such that $\|\mH_{\vh} - \mH_{\tilde{\vh}}\|_{2 \rightarrow 2} = \sqrt{L} \|\mF(\vh-\tilde{\vh})\|_\infty \leq \epsilon\sqrt{QN}/8\nu\sqrt{d_0}$ ($\setC_{\setH}$ is an $\epsilon\sqrt{QN}/8\nu\sqrt{d_0}$-cover of $\setH$ in $\|\cdot\|_f$ norm) then from \eqref{eq:distance}, it is clear that the point $\mH_{\tilde{\vh}}\mX_{\tilde{\vx}}-\mH_{\vh_0}\mX_{\vx_0}$ obeys $\| (\mH_{\tilde{\vh}}\mX_{\tilde{\vx}}-\mH_{\vh_0}\mX_{\vx_0}) - (\mH_{\vh}\mX_{\vx}-\mH_{\vh_0}\mX_{\vx_0})\|_{2 \rightarrow 2}\leq \epsilon$. This implies that $\setC := \{\mH_{\tilde{\vh}}\mX_{\tilde{\vx}}-\mH_{\vh_0}\mX_{\vx_0} :  \mH_{\tilde{\vh}} \in \setC_{\setH}, \mX_{\tilde{\vx}} \in \setC_{\setX}\}$ is an $\epsilon$-cover of $\setS$ in $\|\cdot\|_{2 \rightarrow 2}$ norm, and 
\begin{align*}
& N(\setS,\|\cdot\|_{2 \rightarrow 2}, \epsilon) \leq \\
&  N\bigg(\setH, \|\cdot\|_c,\frac{\epsilon \sqrt{Q}}{8\mu\sqrt{d_0}}\bigg)\cdot N\bigg(\setX, \|\cdot\|_f,\frac{\epsilon\sqrt{QN}}{8\nu\sqrt{d_0}}\bigg) \leq \\
&  N\bigg(2\sqrt{d_0}B_2^{KN}, \|\cdot\|_c,\frac{\epsilon \sqrt{Q}}{8\mu\sqrt{d_0}}\bigg)\cdot N\bigg(2\sqrt{d_0}B_2^M, \|\cdot\|_f,\frac{\epsilon\sqrt{QN}}{8\nu\sqrt{d_0}}\bigg) \\
& = N\bigg(B_2^{KN}, \|\cdot\|_c,\frac{\epsilon \sqrt{Q}}{16\mu d_0}\bigg)\cdot N\bigg(B_2^M, \|\cdot\|_f,\frac{\epsilon\sqrt{QN}}{16\nu d_0}\bigg).
\end{align*}
We evaluate the Dudley integral as follows
\begin{align}\label{eq:gamma2-bound}
&\int_0^{d_{2\rightarrow 2}(\setS)} \sqrt{\log N(\setS,\|\cdot\|_{2 \rightarrow 2}, \epsilon)}d\epsilon \notag\\
&\leq
\int_0^{\tfrac{17\mu\nu d_0}{\sqrt{QN}}} \Bigg(\sqrt{\log N\bigg(B_2^{KN}, \|\cdot\|_c,\frac{\epsilon \sqrt{Q}}{16\mu d_0}\bigg)}\notag \\
&\qquad\qquad \qquad\qquad\qquad  + \sqrt{\log N\bigg(B_2^M, \|\cdot\|_f,\frac{\epsilon\sqrt{QN}}{16\nu d_0}\bigg)}\Bigg) d\epsilon\notag\\
&= \frac{16\mu d_0}{\sqrt{Q}} \int_{0}^{\frac{17\nu}{16\sqrt{N}}} \sqrt{\log N(B_2^{KN},\|\cdot\|_c,\epsilon)}d\epsilon +\notag \\
& \qquad \qquad\qquad \frac{16\nu d_0}{\sqrt{QN}} \int_0^{\frac{17}{16}\mu} \sqrt{\log N(B_2^M,\|\cdot\|_f,\epsilon)}d\epsilon\notag\\
&\leq \frac{16\mu d_0}{\sqrt{Q}} \sqrt{KN}\int_{0}^{\frac{2\nu}{\sqrt{N}}} \sqrt{\log N(B_1^{KN},\|\cdot\|_c,\epsilon)}d\epsilon + \notag\\
& \qquad \qquad \qquad \frac{16\nu d_0}{\sqrt{QN}} \sqrt{M}\int_0^{2\mu} \sqrt{\log N(B_1^M,\|\cdot\|_f,\epsilon)}d\epsilon\notag\\
& \qquad \lesssim \frac{\mu \nu_{\max}  d_0}{\sqrt{Q}} \sqrt{KN \log^4 (QN)} + \frac{\nu d_0}{\sqrt{QN}} \sqrt{M \log^4 L},
\end{align}
where the second last inequality follows from $B_2^{KN} \subseteq \sqrt{KN} B_1^{KN}$, and $B_2^L \subseteq \sqrt{L} B_1^L$, and finally the last inequality is the result of by now standard entropy calculations that can be found in, for example, \cite{krahmer2014suprema}, and Section 8.4 in \cite{rauhut2010compressive}. Combining this result with the Dudley's integral in \eqref{eq:Dudley-Integral} gives a bound on the $\gamma_2$ functional. We now have all the ingredients required in Theorem \ref{thm:Mendelson}. Recall that $\delta = \tfrac{\|\vh\vx^*-\vh_0\vx_0\|_F}{d_0} = \tfrac{\|\mH_{\vh}\mX_{\vx}-\mH_{\vh_0}\mX_{\vx_0}\|_F}{d_0}$. Observe that 
\begin{align*}
\nu^2 &= QN \frac{\|\mC^{\otimes N} \vx\|_\infty^2}{\|\vx\|_2^2} \leq \frac{Q \|\mC^{\otimes N} \|_\infty^2 \cdot N \|\vx\|_1^2}{\|\vx\|_2^2} \\
&\leq \frac{Q \|\mC^{\otimes N} \|_\infty^2 \cdot KN^2 \|\vx\|_2^2}{\|\vx\|_2^2} = \nu_{\max}^2 KN^2.
\end{align*} 
Using this fact, we have 
\begin{align}\label{eq:d2-lose-bound}
d^2_{2 \rightarrow 2}(\setS) \leq  \frac{17\mu^2\nu^2d_0^2}{QN} \leq \frac{\mu^2\nu_{\max}^2 KN^2}{QN}.
\end{align}
Upper bounds in \eqref{eq:dF-bound},\eqref{eq:gamma2-bound}, and \eqref{eq:d2-lose-bound} produce
\begin{align*}
& E \lesssim  d_0^2\Bigg[\left(\mu^2 \nu_{\max}^2\frac{ KN}{Q} + \nu^2\frac{M}{QN} \right) \log^4(QN+L) +\\
& \sqrt{\delta^2\left(\mu^2 \nu_{\max}^2\frac{ KN}{Q} + \nu^2\frac{M}{QN} \right) \log^4(QN+L)} +\delta \sqrt{\frac{\mu^2\nu^2_{\max}KN^2}{QN}}\Bigg], 
\end{align*}
Similarly, the \eqref{eq:dF-bound},\eqref{eq:gamma2-bound}, and \eqref{eq:d2-tight-bound} give
\begin{align*}
U \lesssim \frac{\mu^2\nu^2 d_0^2}{QN},  V &\lesssim \frac{\mu\nu d_0}{\sqrt{QN}} \Bigg( \frac{\mu\nu_{\max} d_0}{\sqrt{Q}} \sqrt{KN\log^4 (QN)} \\
&\qquad\qquad + \frac{\nu d_0}{\sqrt{QN}}\sqrt{M \log^4 L} + \delta d_0\Bigg), 
\end{align*}
Using the fact that $L \geq Q$, and choosing $QN$ as in \eqref{eq:sample-complexity}, and $t = \frac{1}{2} \xi d_0^2\delta^2$, the tail bound in Theorem \ref{thm:Mendelson} now gives
\begin{align*}
& \mathbb{P} \Bigg(\sup_{\mH_{\vh}\in \setH}\sup_{\mX_{\vx}\in \setX}\Bigg| \|(\mH_{\vh}\mX_{\vx}-\mH_{\vh_0}\mX_{\vx_0})\vr\|_2^2 -\\ &\|\mH_{\vh}\mX_{\vx}-\mH_{\vh_0}\mX_{\vx_0}\|_{F}^2 \Bigg| \geq \xi\delta^2\d_0^2  \Bigg)\leq 2 \exp\left(-c \xi^2\delta^2\frac{QN}{\mu^2\nu^2}\right),
\end{align*}
which completes the proof.

\subsection{Proof of Lemma \ref{lem:local-Delh-Delx-RIP}}\label{sec:local-Delh-Delx-RIP}
Just as $\Delta \vh$, and $\Delta \vx$ in \eqref{eq:Deltah-Deltax}, we define $\Delta \mH_{\vh} = \mH_{\vh} - \alpha \mH_{\vh_0}$, and $\Delta \mX_{\vx} = \mX_{\vx}-\bar{\alpha}^{-1}\mX_{\vx_0}$. Similar to \eqref{eq:distance-vectors-to-matrices}, one can also show that $\|\Delta \mH_{\vh} \mX_{\vx} + \mH_{\vh} \Delta \mX_{\vx}\|_F^2 = \|\Delta \vh \vx^* + \vh \Delta \vx^*\|_F^2 \leq 1.2\delta d_0$, where the last inequality is already shown in \eqref{eq:Dhx+hDx-fronorm}, and holds for $\delta \leq \varepsilon \leq 1/15$, where $\delta = \|\vh\vx^*-\vh_0\vx_0^*\|_F/d_0$. 

Using similar steps as laid out in the proof of Lemma \ref{lem:local-RIP}, the local-RIP in Lemma \ref{lem:local-Delh-Delx-RIP} reduces to showing that for a $0 < \xi < 1$, the following holds
\begin{align*}
& \sup_{\mH_{\vh} \in \setH} \sup_{\mX_{\vx} \in \setX}\Big| \|(\Delta \mH_{\vh} \mX_{\vx} + \mH_{\vh} \Delta \mX_{\vx})\vr\|_2^2 - \\
&\qquad\E \|(\Delta \mH_{\vh} \mX_{\vx} + \mH_{\vh} \Delta \mX_{\vx})\vr\|_2^2\Big| \leq \xi \delta^2d_0^2
\end{align*}
with high probability. Define 
\begin{align}\label{eq:setS-2}
\setS =\{\Delta \mH_{\vh} \mX_{\vx} + \mH_{\vh} \Delta \mX_{\vx} \ | \ \mH_{\vh}\in \setH, \mX_{\vx}\in \setX, (\vh, \vx) \in \setN_{\varepsilon}\},
\end{align}
where $\setH$, and $\setX$ are already defined in \eqref{eq:setX-setH}. Let $(\Delta \mH_{\vh} \mX_{\vx} + \mH_{\vh} \Delta \mX_{\vx}),$ and  $(\Delta \mH_{\tilde{\vh}} \mX_{\tilde{\vx}} + \mH_{\tilde{\vh}} \Delta \mX_{\tilde{\vx}})$ be the elements of  $\setS$, and observe that  $(\Delta \mH_{\vh} \mX_{\vx} + \mH_{\vh} \Delta \mX_{\vx})-(\Delta \mH_{\tilde{\vh}} \mX_{\tilde{\vx}} + \mH_{\tilde{\vh}} \Delta \mX_{\tilde{\vx}}) = (\Delta \mH_{\vh} - \Delta \mH_{\tilde{\vh}})\mX_{\tilde{\vx}} + \Delta \mH_{\vh} (\mX_{\vx}-\mX_{\tilde{\vx}}) + (\mH_{\vh}-\mH_{\tilde{\vh}})\Delta \mX_{\vx} + \mH_{\tilde{\vh}}(\Delta\mX_{\vx}-\Delta\mX_{\tilde{\vx}})$, which gives 
\begin{align*}
&\|(\Delta \mH_{\vh} \mX_{\vx} + \mH_{\vh} \Delta \mX_{\vx})-(\Delta \mH_{\tilde{\vh}} \mX_{\tilde{\vh}} + \mH_{\tilde{\vh}} \Delta \mX_{\tilde{\vx}})\|_{2 \rightarrow 2} \\
&\leq \|\mH_{\vh}-\mH_{\tilde{\vh}}\|_{2 \rightarrow 2} \|\mX_{\tilde{\vx}}\|_{2 \rightarrow 2}  + \|\Delta \mH_{\vh}\|_{2 \rightarrow 2}  \|\mX_{\vx}-\mX_{\tilde{\vx}}\|_{2 \rightarrow 2}\\
& \qquad +\|\mH_{\vh}-\mH_{\tilde{\vh}}\|_{2 \rightarrow 2} \|\Delta \mX_{\vx}\|_{2 \rightarrow 2}  + \|\mH_{\tilde{\vh}}\|_{2 \rightarrow 2}  \|\mX_{\vx}-\mX_{\tilde{\vx}}\|_{2 \rightarrow 2} = \\
&\sqrt{L}\big(\|\mF_M(\vh-\tilde{\vh})\|_\infty \|\mC^{\otimes N}\tilde{\vx}\|_\infty + \|\mF_M(\Delta \vh)\|_\infty  \|\mC^{\otimes N}(\vx-\tilde{\vx})\|_\infty\\
&+\|\mF_M(\vh-\tilde{\vh})\|_\infty \|\mC^{\otimes N} (\Delta\vx)\|_\infty +\|\mF_M\tilde{\vh}\|_\infty  \|\mC^{\otimes N}(\vx-\tilde{\vx})\|_\infty\big).
\end{align*}
 As it is clear form the definition of set $\setS$ that the index vectors $(\vh,\vx)$, and $(\tilde{\vh},\tilde{\vx})$ of the elements of $\setS$ lie in  $\setN_{d_0} \cap \setN_\mu \cap \setN_{\nu} \cap \setN_{\varepsilon}$, and by assumption $\varepsilon \leq 1/15$, therefore, we have $\sqrt{L}\|\mF_M\Delta\vh\|_\infty \leq 6 \mu \sqrt{d_0}$, and $\sqrt{QN}\|\mC^{\otimes N} \Delta \vx\|_\infty \leq 6 \nu \sqrt{d_0}$ using Lemma \ref{lem:local-regulaity-Del-norm-bounds}. This results in 
\begin{align*}
&\|(\Delta \mH_{\vh} \mX_{\vx} + \mH_{\vh} \Delta \mX_{\vx})-(\Delta \mH_{\tilde{\vh}} \mX_{\tilde{\vh}} + \mH_{\tilde{\vh}} \Delta \mX_{\tilde{\vx}})\|_{2 \rightarrow 2}  \\
&\leq  \frac{4\nu\sqrt{d_0}}{\sqrt{QN}} \sqrt{L}\|\mF_M(\vh-\tilde{\vh})\|_\infty+  6\mu\sqrt{d_0}\|\mC^{\otimes N}(\vx-\tilde{\vx})\|_\infty\\
& + \frac{6\nu\sqrt{d_0}}{\sqrt{QN}}  \sqrt{L}\|\mF_M(\vh-\tilde{\vh})\|_\infty +4\mu\sqrt{d_0} \|\mC^{\otimes N}(\vx-\tilde{\vx})\|_\infty \\
&=   \frac{10\nu\sqrt{d_0}}{\sqrt{QN}} \|\vh-\tilde{\vh}\|_f+ \frac{10\mu\sqrt{d_0}}{\sqrt{Q}} \|\vx-\tilde{\vx}\|_c,
\end{align*}
where the last equality follows by using the $\|\cdot\|_c$, $\|\cdot\|_f$ norms, defined earlier in \eqref{eq:exotic-norms}. Similar to the discussion before, this means that the $\epsilon$-cover of $\setS$ in \eqref{eq:setS-2} is obtained by the $\epsilon\sqrt{QN}/20\sqrt{d_0}\nu$-cover of $\setH$ in $\|\cdot\|_f$ norm, and $\epsilon\sqrt{Q}/20\mu\sqrt{d_0}$-cover of $\setX$ in $\|\cdot\|_c$ norm. With this fact in place the rest of the proof follows exactly the same outline as the proof of Lemma \ref{lem:local-RIP}. 

\subsection{Proof of Lemma \ref{lem:local-regulaity-Del-norm-bounds}}\label{sec:local-regulaity-Del-norm-bounds}
Recall that $\alpha_1 = \vh^*\vh_0/d_0$, which directly gives $|\alpha_1| \leq \|\vh\|\|\vh_0\|/d_0 \leq 2$ by using the Cauchy-Schwartz inequality, and the fact that $\vh \in \setN_{d_0}$.  In a similar manner, we can also show that $|\alpha_2| \leq 2$. Expand $\|\vh\vx^*-\vh_0\vx_0^*\|_F^2 = \delta^2d_0^2$ using \eqref{eq:difference-expansion} to obtain
	\[
	\delta^2d_0^2 = (\alpha_1\bar{\alpha}_2-1)^2d_0^2 +|\bar{\alpha}_2|^2 \|\tilde{\vh}\|_2^2 d_0 + |\alpha_1|^2\|\tilde{\vx}\|_2^2 d_0 + \|\tilde{\vh}\|_2^2\|\tilde{\vx}\|_2^2,
	\]
	which implies $|\alpha_1\bar{\alpha}_2-1| \leq \delta$. 
	 
	 The identities $\|\Delta \vh\|_2^2 \leq 6.1\delta^2d_0$, $\|\Delta \vx\|_2^2 \leq 6.1\delta^2d_0$, $\|\Delta \vh\|_2^2 \|\Delta \vx\|_2^2 \leq 8.4 \delta^4 d_0^2$, and $\sqrt{L}\|\mF_M\Delta\vh\|_\infty \leq 6 \mu \sqrt{d_0}$ are proved in Lemma 5.15 in \cite{li2016rapid}. We now prove that $\sqrt{QN} \max_n\|\mC (\Delta \vx_n) \|_\infty \leq 6\nu \sqrt{d_0}$.\\
	
	\textbf{Case 1}: $\|\vh\|_2 \geq \|\vx\|_2$, and $\alpha = (1-\delta_0)\alpha_1$. Observe that in this case 
	\begin{align*}
	|\alpha_2| &\leq \frac{\|\vx\|_2\|\vx_0\|_2}{d_0} \leq \frac{1}{\sqrt{d_0}}\sqrt{\|\vh\|_2\|\vx\|_2} \\
	&\leq \frac{1}{\sqrt{d_0}} \sqrt{\|\vh\vx^* - \vh_0\vx_0^*\|_F + \|\vh_0\vx_0^*\|_F} = \sqrt{1+\delta},
	\end{align*}
	where we used the fact that $\|\vh\vx^*-\vh_0\vx_0^*\|_F = \delta d_0$, and $\|\vh_0\|_2 = \|\vx_0\|_2 = \sqrt{d_0}$. Therefore, $\tfrac{1}{|(1-\delta_0)\alpha_1|} = \tfrac{|\alpha_2|}{|(1-\delta_0)\bar{\alpha}_2\alpha_1|} \leq \tfrac{\sqrt{1+\delta}}{|1-\delta_0||1-\delta|} \leq 2$, where the last inequality follows using our choice $\delta \leq \varepsilon \leq 1/15$, and $\delta_0 = \delta/10$. This gives us
	\begin{align*}
	&\max_n\sqrt{QN} \|\mC(\Delta\vx_n)\|_\infty \\
	&\leq \max_n\sqrt{QN}\|\mC \vx_n\|_\infty + \tfrac{1}{(1-\delta_0)|\alpha_1|} \max_n\sqrt{QN} \|\mC\vx_{n,0}\|_\infty\\
	&\leq 4\nu\sqrt{d_0}+2\nu \sqrt{d_0} \leq 6\nu \sqrt{d_0}.
	\end{align*}
	
	\textbf{Case 2}: $\|\vh\|_2 < \|\vx\|_2$, and $\alpha = \frac{1}{(1-\delta_0)\bar{\alpha}_2}$. Since $|\alpha_2| \leq 2$, we have
	\begin{align*}
	&\max_n\sqrt{QN}\|\mC(\Delta\vx_n)\|_\infty \\
	&\leq \max_n\sqrt{QN}\|\mC \vx_n\|_\infty + (1-\delta_0)|\alpha_2| \max_n\sqrt{QN} \|\mC\vx_{n,0}\|_\infty\\
	&\leq 4\nu\sqrt{d_0}+2(1-\delta_0)\nu \sqrt{d_0} \leq 6\nu \sqrt{d_0}.
	\end{align*}
	This completes the proof. 
	
\subsection{Proof of Lemma \ref{lem:local-regularity-G}}\label{sec:local-regularity-G}
The proof is adapted from Lemma 5.17 in \cite{li2016rapid}.

\textbf{Case 1}: $\|\vh\|_2 \geq \|\vx\|_2$, and $\alpha = (1-\delta_0) \alpha_1.$ Using $\delta \leq \varepsilon \leq 1/15$, we have the following easily verifiable (Lemma 5.17 in \cite{li2016rapid})  identities $\<\vh,\Delta \vh\> \geq \delta_0 \|\vh\|_2^2, ~ \|\vx\|_2^2 < 2d$, and also 
\begin{align}\label{eq:inner-prod-bound-case1}
&\Re{\<\vf_{\ell}\vf_{\ell}^*\vh,\Delta \vh\>} \geq \frac{2d\mu^2}{L}   \ \text{when}  \ L \frac{|\vf_{\ell}^*\vh|^2}{8d\mu^2} > 1, \notag\\
& \Re{\<\vc_{q,n}\vc_{q,n}^*\vx,\Delta \vx\>} \geq \frac{d\nu^2}{QN} \ \text{when} \ QN \frac{|\vc_{q,n}^*\vx|^2}{8d\nu^2} > 1.
\end{align}
For example, the last identity an simply be proven as follows
\begin{align*}
&\Re{\<\vc_{q,n}\vc_{q,n}^*\vx, \vx-\bar{\alpha}^{-1}\vx_0\>}\\
&\geq |\vc_{q,n}^*\vx|^2 - \frac{1}{(1-\delta_0)|\alpha_1|}|\vc_{q,n}^*\vx||\vc_{q,n}^*\vx_0|\\
&= |\vc_{q,n}^*\vx|^2 - \frac{|\alpha_2|}{(1-\delta_0)|\alpha_1\bar{\alpha}_2|}|\vc_{q,n}^*\vx||\vc_{q,n}^*\vx_0|.
\end{align*}
Using Lemma \ref{lem:local-regulaity-Del-norm-bounds}, we have that $|\alpha_2| \leq 2$, $|\alpha_1\bar{\alpha_2} -1 | \leq \delta$, and the fact that $(\vh,\vx) \in \setN_\mu \cap \setN_\nu \cap \setN_{d_0}$, we further obtain 
\begin{align*}
&\text{Re}(\<\vc_{q,n}\vc_{q,n}^*\vx, \vx-\bar{\alpha}^{-1}\vx_0\>)\\
&\geq |\vc_{q,n}^*\vx|^2  - \frac{2}{(1-\delta)(1-\delta_0)}|\vc_{q,n}^*\vx||\vc_{q,n}^*\vx_0|\\
& \geq \frac{8d\nu^2}{QN} - \frac{2}{(1-\delta)(1-\delta_0)}\sqrt{\frac{8d\nu^2}{QN}}\cdot\sqrt{\frac{10d\nu^2 }{9QN}}\geq \frac{d\nu^2}{QN},
\end{align*}
where the last inequality is obtained by using $|\vc_{q,n}^*\vx_0| = \tfrac{\nu \sqrt{d_0}}{\sqrt{QN}}$, and $0.9d_0 \leq d \leq 1.1 d_0$. 

\textbf{Case 2}: $\|\vh\|_2 < \|\vx\|_2$, $\alpha = \tfrac{1}{(1-\delta_0)\bar{\alpha}_2}$. Given $\delta \leq \varepsilon \leq 1/15$, one can show (Lemma 5.17 in \cite{li2016rapid}) that $\<\vx,\Delta \vx\> \geq \delta_0 \|\vx\|_2^2, ~ \|\vh\|_2^2 < 2d$, and also 
\begin{align}\label{eq:inner-prod-bound-case2}
&\text{Re}(\<\vf_{\ell}\vf_{\ell}^*\vh,\Delta \vh\>) \geq \frac{d\mu^2}{L}  \ \text{when} \ L \frac{|\vf_{\ell}^*\vh|^2}{8d\mu^2} > 1,\notag\\
 &\Re{\<\vc_{q,n}\vc_{q,n}^*\vx,\Delta \vx\>}\geq \frac{2d\nu^2}{QN} \ \text{when} \ QN\frac{ |\vc_{q,n}^*\vx|^2}{8d\nu^2} > 1.
\end{align}
Expanding gradients, it is easy to see that 
\begin{align}\label{eq:gGh-gGx-inner-prod}
&\Re{\<\nabla G_{\vh}, \Delta \vh \> + \<\nabla G_{\vx}, \Delta \vx \>} = \frac{\rho}{d} \Bigg( G_0^\prime\Big(\frac{\|\vh\|_2^2}{2d}\Big) \Re{\< \vh,\Delta \vh\>}+ \notag\\
& G_0^\prime\Big(\frac{\|\vx\|_2^2}{2d}\Big) \Re{\< \vx,\Delta \vx\>} + G_0^\prime \Big( \frac{L|\vf_{\ell}^*\vh|^2}{8d\mu^2} \Big) \frac{L}{4\mu^2} \Re{\<\vf_{\ell}\vf_{\ell}^*\vh, \Delta \vh\>}\notag\\
& + G_0^\prime \Big( \frac{QN|\vc_{q,n}^*\vx|^2}{8d\nu^2} \Big)\frac{QN}{4\nu^2} \Re{\<\vc_{q,n}\vc_{q,n}^*\vx, \Delta \vx\>} \Bigg). 
\end{align}
We can now conclude the following inequality for both of the above cases
\begin{align*}
G_0^\prime \left(\frac{\|\vh\|_2^2}{2d}\right) \< \vh,\Delta \vh\> \geq \frac{\delta d}{5}G_0^\prime \left(\frac{\|\vh\|_2^2}{2d}\right).
\end{align*}
To see this, note that it holds trivially when $\|\vh\|_2^2 < 2d$, and in the contrary scenario when $\|\vh\|_2^2 \geq 2d$, Case-2 is not possible, and Case-1 always has $\<\vh,\Delta\vh\> \geq \delta_0\|\vh\|_2^2$, and hence $\<\vh,\Delta\vh\> \geq \delta d /5$ shows that the inequality holds. Similarly, we can also argue that 
\begin{align*}
 ~ G_0^\prime \left(\frac{\|\vx\|_2^2}{2d}\right) \< \vx,\Delta \vx\> \geq \frac{\delta d}{5}G_0^\prime \left(\frac{\|\vx\|_2^2}{2d}\right).
\end{align*}
Moreover, the following inequalities 
\begin{align}
& G_0^\prime \left( \frac{L|\vf_{\ell}^*\vh|^2}{8d\mu^2} \right) \frac{L}{4\mu^2} \text{Re}(\<\vf_{\ell}\vf_{\ell}^*\vh, \Delta \vh\> )\geq \frac{d}{4} G_0^\prime \left( \frac{L|\vf_{\ell}^*\vh|^2}{8d\mu^2} \right), \label{eq:G0-h-coherence-bound}\\
& G_0^\prime \left( \frac{QN|\vc_{q,n}^*\vx|^2}{8d\nu^2} \right) \frac{QN}{4\nu^2} \text{Re}(\<\vc_{q,n}\vc_{q,n}^*\vx, \Delta \vx\> ) \notag \\
& \qquad\qquad \qquad\geq \frac{d}{4} G_0^\prime \left( \frac{QN|\vc_{q,n}^*\vx|^2}{8d\nu^2} \right)\label{eq:G0-x-coherence-bound}
\end{align}
hold in general. Again to see this, note that both hold trivially when $QN|\vc_{q,n}^*\vx|^2 > 8d\nu^2$, and $L |\vf_\ell^*\vh|^2 > 8d\mu^2$, and in the contrary case, we have from the bounds \eqref{eq:inner-prod-bound-case1}, and \eqref{eq:inner-prod-bound-case2} that the \eqref{eq:G0-h-coherence-bound}, and \eqref{eq:G0-x-coherence-bound} above hold in Case 1 and 2. Plugging these results in \eqref{eq:gGh-gGx-inner-prod} proves the lemma. 
\subsection{Proof of Lemma \ref{lem:smoothness-CL}}\label{sec:smoothness-CL}
Given that $\vz = (\vh,\vx), \ \vz + \vw = (\vh+\vu, \vx+\vv) \in \setN_{\tf} \cap \setN_{\varepsilon}$, and lemma below, it follows that $\vz+\vw \in \setN_{d_0}\cap \setN_{\mu}\cap \setN_{\nu}$.  
	\begin{lem}
		There holds $\setN_{\tf} \subset \setN_{d_0} \cap \setN_{\mu}\cap \setN_{\nu}$ under local-RIP, and noise robustness lemmas in Section \ref{sec:Key-conditions}.
	\end{lem}
\noindent Proof of this lemma follows from exact same line of reasoning as the proof of Lemma 5.5 in \cite{li2016rapid}. 
	
	Using the gradient $\nabla F_{\vh}$ expansion in \eqref{eq:gF-def}, we estimate the upper bound of $\|\nabla F_{\vh}(\vz+\vw)-\nabla F_{\vh} (\vz)\|_2$. A straight forward calculation gives 
		\begin{align*}
		&\gfh(\vz+\vw)-\gfh(\vz)  = \setA^*\setA(\vu\vx^*+\vh\vv^*+\vu\vv^*)\vx+\\
		&\qquad \quad\setA^*\setA((\vh+\vu)(\vx+\vv)^*-\vh_0\vx_0^*)\vv-\setA^*(\ve)\vv.
		\end{align*}
		Note that $\vz,\vz+\vw \in \setN_{d_0}$ directly implies $
		\|\vu\vx^*+\vh\vv^*+\vu\vv^*\|_{F} \leq \|\vu\|_2\|\vx\|_2+\|\vh+\vu\|_2\|\vv\|_2 \leq 2 \sqrt{d_0}(\|\vu\|_2+\|\vv\|_2),$ where $\|\vh+\vu\|_2 \leq 2 \sqrt{d_0}$. Moreover, $\vz+\vw \in \setN_{\varepsilon}$ implies $\|(\vh+\vu)(\vx+\vv)^*-\vh_0\vx_0^*\|_F \leq \varepsilon d_0.$ Using $\|\vx\|_2 \leq 2 \sqrt{d_0}$ together with $\|\setA^*(\ve)\|_{2 \rightarrow 2}\leq \varepsilon d_0$, which follows from Lemma \ref{lem:noise-stability}, gives
		\begin{align}\label{eq:grad-Fh-zw-z}
		&\|\gfh(\vz+\vw)-\gfh(\vz)\|_2 \leq 4d_0\|\setA\|_{2 \rightarrow 2}^2 (\|\vu\|_2+\|\vv\|_2) + \notag \\
		& \varepsilon d_0 \|\setA\|_{2 \rightarrow 2}^2 \|\vv\|_2 + \varepsilon d_0 \|\vv\|_2\leq 5d_0\|\setA\|_{2 \rightarrow 2}^2 (\|\vu\|_2+\|\vv\|_2).
		\end{align}
		In a similar manner, we can show that 
		\begin{align}\label{eq:grad-Fx-zw-z}
		\|\nabla F_{\vx}(\vz+\vw)-\nabla F_{\vx} (\vz)\|_2\leq 5d_0 \|\setA\|_{2 \rightarrow 2}^2 (\|\vu\|_2+\|\vv\|_2).
		\end{align}
		 Plugging in the gradient expressions from \eqref{eq:ghG-def}, we have
		\begin{align}\label{eq:grad-G-zw-z}
		&\|\nabla G_{\vh} (\vz+\vw) - \nabla G_{\vh}(\vz)\|_2\leq \notag \\
		& \frac{\rho}{2d}\left| G_0^\prime \left(\frac{\|\vh+\vu\|_2^2}{2d}\right)-G_0^\prime \left(\frac{\|\vh\|_2^2}{2d}\right)\right|\|\vh+\vu\|_2+\notag\\
		&\frac{\rho}{2d}\left| G_0^\prime \left(\frac{\|\vh\|_2^2}{2d}\right)\right|\|\vu\|_2 + \frac{L\rho}{8d\mu ^2 }\Bigg\| \sum_{\ell} \notag\\
		&\underbrace{\Big[G_0^\prime \left( \frac{L |\vf_{\ell}^*(\vh+\vu)|^2}{8d \mu^2}\right)\vf_{\ell}^* (\vh+\vu)-G_0^\prime \left( \frac{L |\vf_{\ell}^*\vh|^2}{8d \mu^2}\right) \vf_{\ell}^* \vh\Big]}_{\alpha_\ell}\vf_{\ell}\Bigg\|_2.
		\end{align}
		Begin by noting that $G_0^\prime(\vz) \leq 2|\vz|$, and for any $\vz_1, \vz_2 \in \R$, it holds that $|G_0^\prime(\vz_1)-G_0^\prime(\vz_2)| \leq 2|\vz_1-\vz_2|$; moreover,  $(\vh+\vu) \in \setN_{d_0}$, and simplifying using the triangle inequality, we obtain
		\begin{align}\label{eq:grad-G-hu-h}
		&\left| G_0^\prime \left( \frac{\|\vh+\vu\|_2^2}{2d}\right) - G_0^\prime \left(\frac{\|\vh\|_2^2}{2d}\right)\right| \leq \frac{\|\vh+\vu\|_2+\|\vx\|_2}{d}\|\vu\|_2 \leq\notag\\
		& \qquad \frac{4 \sqrt{d_0}}{d}\|\vu\|_2, ~ \text{and}~ G_0^\prime \left(\frac{\|\vh\|_2^2}{2d}\right) \leq 2 \frac{\|\vh\|_2^2}{2d} \leq 4\frac{d_0}{d}.
		\end{align}
		Using same identities as above, and $\vh,\vh+\vu \in \setN_{\mu}$,  we have
		\begin{align*}
		& G_0^\prime \left(\frac{L|\vf_\ell^*(\vh+\vu)|^2}{8d\mu^2}\right)-G_0^\prime \left(\frac{L|\vf_\ell^*\vh|^2}{8d\mu^2}\right) \leq  \frac{L}{4d\mu^2} (|\vf_\ell^*(\vh+\vu)|+\\
		& \qquad |\vf_\ell^*\vh|)|\vf_\ell^*\vu| \leq \frac{L}{4d \mu^2} \cdot \frac{8\mu\sqrt{d_0}}{\sqrt{L}}  |\vf_\ell^*\vu|, 
		\end{align*}
		and 
		\[
		G_0^\prime \left(\frac{L |\vf_\ell^*(\vh+\vu)|^2}{8d\mu^2}\right) \leq 2 \frac{L |\vf_\ell^*(\vh+\vu)|^2}{8d\mu^2} \leq 4\frac{d_0}{d}.
		\]
		We can now use above two displays to obtain $\alpha_\ell = \pm \frac{L}{4d \mu^2}(\vf_\ell^* (\vh+\vu) + \vf_\ell^*\vu) \vf_\ell^*\vh \vf_\ell^*\vu + 4\frac{d_0}{d} \vf_\ell^*\vu$, which eventually gives 
		\begin{align}\label{eq:sum-alpha-ell-f-ell}
		& \big\|\sum_{\ell}\alpha_\ell \vf_\ell \big\|_2 \leq \notag\\
		& \max_\ell\left(\frac{L}{4d\mu^2}\left|(\vf_\ell^* (\vh+\vu) + \vf_\ell^*\vu) \vf_\ell^*\vh\right|+4\frac{d_0}{d}\right) \left\|\sum_{\ell} \vf_\ell^*\vu \vf_\ell\right\|_2\notag \\
		& = \frac{12d_0}{d} \|\vu\|_2,
		\end{align}
		 where we used the fact that $\sum_{\ell} \vf_\ell\vf_\ell^* = \mI$. Finally, using $\vh+\vu \in \setN_{d_0}$, and plugging \eqref{eq:sum-alpha-ell-f-ell}, and \eqref{eq:grad-G-hu-h} in \eqref{eq:grad-G-zw-z} gives us
		\begin{align}\label{eq:grad-Gh-zw-z}
		\| \nabla G_{\vh} (\vz+\vw) - \nabla G_{\vh}(\vz)\|_2 \leq 5 \rho\frac{d_0}{d^2}\|\vu\|_2 + \frac{3d_0L\rho}{2d^2\mu^2}  \|\vu\|_2.
		\end{align}
		In an exactly similar manner, we have 
		\begin{align*}
		&\|\nabla G_{\vx} (\vz+\vw) - \nabla G_{\vx}(\vz)\|_2\leq \frac{\rho}{2d} \Bigg(\\
		& \left| G_0^\prime \left(\frac{\|\vx+\vv\|_2^2}{2d}\right)-G_0^\prime \left(\frac{\|\vx\|_2^2}{2d}\right)\right|\|\vx+\vv\|_2+\left| G_0^\prime \left(\frac{\|\vx\|^2}{2d}\right)\right|\|\vv\|_2\Bigg)\\
		&\qquad + \frac{QN\rho}{8d\nu ^2 }\Bigg\| \sum_{q,n}\beta_{q,n}\vc_{q,n}\Bigg\|_2,
		\end{align*}
		where
		\begin{align*}
		\beta_{q,n} &=  \Bigg[G_0^\prime \left( \frac{QN |\vc_{q,n}^*(\vx+\vv)|^2}{8d \nu^2}\right)\vc_{q,n}^* (\vx+\vv)-\\
		& \qquad\qquad G_0^\prime \left( \frac{QN|\vc_{q,n}^*\vx|^2}{8d \nu^2}\right) \vc_{q,n}^* \vx\Bigg],
		\end{align*}
		and one can show using the facts that $\vx,\vx+\vv$ are the members of the set $\setN_{d_0} \cap \setN_{\nu}$, $\underset{q,n}{\sum} \vc_{q,n}\vc_{q,n}^* = \mI$, and the approach similar to obtain the bound \eqref{eq:grad-Gh-zw-z} that 
		\begin{align}\label{eq:grad-Gx-zw-z}
		\| \nabla G_{\vx} (\vz+\vw) - \nabla G_{\vx}(\vz)\|_2 \leq 5 \rho\frac{d_0}{d^2}\|\vv\|_2 + \frac{3d_0QN\rho}{2d^2\nu^2}  \|\vv\|_2.
		\end{align}
	 Using \eqref{eq:grad-Fh-zw-z}, \eqref{eq:grad-Fx-zw-z}, \eqref{eq:grad-Gh-zw-z}, and \eqref{eq:grad-Gx-zw-z} together with the fact that $\nabla \tf(\vz) = (\nabla F_{\vh} (\vz) + \nabla G_{\vh}(\vz), \nabla F_{\vx} (\vz) + \nabla G_{\vx}(\vz)),$ and using $\|\vu\|_2+\|\vv\|_2 \leq \sqrt{2}\|\vw\|_2$, we obtain 
	 \begin{align*}
	 &\|\gtf(\vz+\vw)-\gtf(\vz)\|_2 \leq \\
	 & \sqrt{2}d_0\Big[10\|\setA\|_{2 \rightarrow 2}^2+\frac{\rho}{d^2}\Big( 5+ \frac{3L}{2\mu^2} + \frac{3QN}{2\nu^2} \Big)\Big]\|\vw\|_2.
	 \end{align*}

\subsection{Proof of Lemma \ref{lem:noise-stability}}\label{sec:noise-stability}
We begin by controlling the operator norm of the linear map $\setA$ in \eqref{eq:linear-map}, where $\vf_\ell^*$, $\hat{\vc}_{\ell,n}^*$ are the rows of $\mF_M$, and $\sqrt{L}(\mF_Q\mR_n \mC)^{\otimes N}$, respectively. It is easy to see that $\|\setA\| = \max_{\ell,n} \|\vf_\ell\hat{\vc}_{\ell,n}^*\|_F $, which follows from the fact that $\<\vf_\ell\hat{\vc}_{\ell,n}^*,\vf_{\ell^\prime}\hat{\vc}_{\ell^\prime,n^\prime}^*\> = 0$ whenever $\ell \neq \ell^\prime$ or $n \neq n^\prime$. Since $\|\vf_\ell\hat{\vc}_{\ell,n}^*\|_F = \|\hat{\vc}_{\ell,n}\|_2$, we only require an upper bound on $\|\hat{\vc}_{\ell,n}\|_2$. 
	
As introduced earlier, $\mR_n = \text{diag}(\vr_n)$, where $\vr_n$ is a $Q$-vector of Rademacher random variables. Defining $\vc_{q,n}^*$ as the rows of $\mC^{\otimes N}$, we can write $\hat{\vc}_{\ell,n}^*$ as the random sum
	\begin{align*}
	\hat{\vc}_{\ell,n} = \sqrt{L}\sum_{q=1}^Q f_{\ell}[q]r_n[q]\vc_{q,n},
	\end{align*}
	where $r_n[q]$ is the $q$th entry of $\vr_n$. The upper bound on $\|\hat{\vc}_{\ell,n}\|_2$ can now be obtained by an application of Proposition \ref{prop:conc_ineq} below. The sequence $\{\mZ_k\}$ in the statement of Proposition \ref{prop:conc_ineq} in this case is $\sqrt{L}\{f_{\ell}[q]r_n[q]\vc_{q,n}\}_{q=1}^Q$. Using the identities, $\sum_{q=1}^Q \vc_{q,n}\vc_{q,n}^* = \mI$, and $\sum_{q=1}^Q \|\vc_{q,n}\|_2^2 = K$, which follows from the fact that $Q \times K$ matrix $\mC$ has orthonormal columns, i.e., $\mC^*\mC = \mI$, a simple calculation shows that the variance $\sigma^2_Z$ in \eqref{eq:variance-conc-ineq} is $\sigma_Z^2 \leq K+1$. Choosing $ t^2 = \alpha K \log(LN)$, and using the bound in \eqref{eq:conc-ineq-bound} results in 
	\begin{align}\label{eq:hatbln-2norm-bound}
	\max_{\ell,n}\|\hat{\vc}_{\ell,n}\|_2  \leq \sqrt{\alpha K\log(LN)}
	\end{align}
	with probability at least $1-\setO(LN^{-\alpha})$, where $\alpha \geq 1$ is a free parameter. This proves the first claim in the statement of the lemma. 
	
	As for the second claim, we begin by writing the vector $\setA^*(\ve)$ as a sum of random matrices
	\begin{align*}
	\setA^*(\ve) = \sum_{\ell=1}^L \sum_{n=1}^N  \hat{e}_n[\ell]\hat{\vc}_{\ell,n}\vf^*_\ell =  \frac{\sigma d_0}{\sqrt{LN}} \sum_{\ell=1}^L\sum_{n=1}^N g_n[\ell] \hat{\vc}_{\ell,n}\vf^*_\ell,
	\end{align*}
	where the second equality follows by rewriting the Gaussian random variables $\hat{e}_n[\ell]$ as a scaling of the standard Gaussian random variables $g_n[\ell] \sim \text{Normal}(0,\tfrac{1}{2}) + \iota \text{Normal}(0,\tfrac{1}{2})$. 
	We employ the matrix concentration inequality in Proposition \ref{prop:conc_ineq} to control the operator norm of the random matrix above. The summand matrices $\{\mZ_k\}$ in Proposition \ref{prop:conc_ineq} in this case are simply $\{\hat{\vc}_{\ell,n}\vf^*_\ell\}_{\ell,n}$. The computation of  variance in \eqref{eq:variance-conc-ineq} now reduces to 
	\begin{align*}
	\sigma^2_Z &= \frac{\sigma^2 d_0^2}{LN} \max\Bigg\{\left\|\sum_{\ell=1}^L \sum_{n=1}^N \|\vf_\ell\|_2^2 \hat{\vc}_{\ell,n} \hat{\vc}_{\ell,n}^*\right\|_{2 \rightarrow 2},\\
	&\qquad \qquad\qquad \qquad \qquad  \left\|\sum_{\ell=1}^L \sum_{n=1}^N \|\hat{\vc}_{\ell,n}\|_2^2 \vf_{\ell} \vf_{\ell}^*\right\|_{2 \rightarrow 2} \Bigg\}.
	\end{align*}
	Recall that $\vf_\ell^*$, $\hat{\vc}_{\ell,n}^*$ are the rows of $\mF_M$, and $\sqrt{L}(\mF_Q\mR_n \mC)^{\otimes N}$, respectively, therefore, 
	\begin{align*}
	&\|\vf_\ell\|_2^2 = \frac{M}{L}, \ \sum_{\ell=1}^L \vf_\ell\vf_\ell^* = \mI,\\
	&\sum_{\ell=1}^L \sum_{n=1}^N \hat{\vc}_{\ell,n} \hat{\vc}_{\ell,n}^* = L (\mC^*\mR_n^*(\mF_Q)^*\mF_Q \mR_n\mC)^{\otimes N} = L \mI_{KN\times KN}. 
	\end{align*}
	Using the above display together with \eqref{eq:hatbln-2norm-bound}, the variance in this case is upper bounded as 
	\begin{align*}
	\sigma^2_Z \leq \frac{\sigma^2d_0^2}{LN} \max \left( M, \alpha KN \log(LN) \right). 
	\end{align*}
	Choosing $t = \tfrac{2\varepsilon d_0}{50}$, and 
	\[
	LN \geq \frac{\sigma^2}{\varepsilon^2} c^\prime_{\alpha} \max(M,\alpha KN \log(LN))\log(LN)
	\]
	 and using the inequality \eqref{eq:conc-ineq-bound} in Proposition \ref{prop:conc_ineq} proves the claim. 

\begin{prop}
	[Corollary 4.2 in \cite{tropp12us}]\label{prop:conc_ineq} Consider a finite sequence $\{\mZ_k\}$ of fixed matrices with dimensions $d_1 \times d_2$, and let $\{g_k\}$ be a finite sequence of independent Gaussian or Rademacher random variables. Define the variance 
	\begin{align}\label{eq:variance-conc-ineq}
	\sigma_Z^2 := \max \left\{ \left\|\sum_k \mZ_k\mZ_k^*\right\|_{2 \rightarrow 2}, \left\|\sum_k \mZ_k^*\mZ_k\right\|_{2 \rightarrow 2} \right\}. 
	\end{align}
	Then for all $t \geq 0$
	\begin{align}\label{eq:conc-ineq-bound}
	\mathbb{P} \left(\left\| \sum_k g_k \mZ_k\right\|_{2 \rightarrow 2} \geq \ \ t \right) \leq (d_1+d_2)\mathrm{e}^{-t^2/2\sigma_Z^2}.
	\end{align}
\end{prop}

\subsection{Proof of Theorem \ref{thm:initialization}}\label{sec:Theorem2-Proof}
We now give the proof of Theorem \ref{thm:initialization} by explicitly constructing a good initial guess: $(\vu_0,\vv_0) \in \frac{1}{\sqrt{3}}\setN_{d_0} \cap \frac{1}{\sqrt{3}}\setN_{\mu}  \cap \frac{1}{\sqrt{3}}\setN_{\nu}\cap \setN_{\frac{2}{5}\varepsilon}$, from the measurements $\vy$, and the knowledge of the model $\setA$. 

\begin{proof} 
	Recall that $\delta d_0 = \|\vh\vx^*-\vh_0\vx_0^*\|_F $, and under $(\vh,\vx) = (\boldsymbol{0},\boldsymbol{0})$, this implies that $\delta d_0 = \|\vh_0\vx_0^*\|_F = d_0$ giving 
	$\delta = 1$.  In this scenario, one can conclude from Lemma \ref{lem:local-RIP} that $\|\vh_0\vx_0^*\|_F^2 - \xi d_0^2 \leq \|\setA(\vh_0\vx_0^*)\|_2^2 \leq  \|\vh_0\vx_0^*\|_F^2 + \xi d_0^2$ with probability at least $
	1-2\exp\left(-c\xi^2QN/\mu^2\nu^2\right)$ whenever
	\begin{align}\label{eq:sample-complexity-initialization}
	QN \geq \frac{c}{\xi^2}\left( \mu^2\nu_{\max}^2 KN^2 + \nu^2 M\right)\log^4(LN).
	\end{align}
    This implies that 
\begin{align*}
|\< (\setA^*\setA-\setI)(\vh_0\vx_0^*), \vh_0\vx_0^*\>| 
&\leq \xi \|\vh_0\vx_0^*\|_F^2,
\end{align*}
and hence $\|\setA^*\setA(\vh_0\vx_0^*) - \vh_0\vx_0^*\|_{2 \rightarrow 2} \leq \xi d_0.$ Using triangle inequality, and \eqref{eq:linear-map}, we obtain 
\begin{align}\label{eq:A*(y)-h0x0-operator-norm}
\|\setA^*(\hat{\vy})-\vh_0\vx_0^*\|_{2\rightarrow 2} &\leq \|\setA^*\setA(\vh_0\vx_0^*)-\vh_0\vx_0^*\|_{2 \rightarrow 2} + \|\setA^*(\ve)\|_{2 \rightarrow 2}\notag\\
& \leq \xi d_0 +  \frac{2\varepsilon d_0}{50} \leq  \frac{3\varepsilon d_0}{50} := \gamma d_0,
\end{align}
where the last display follows from Lemma \ref{lem:noise-stability}, and choosing $\xi = \tfrac{\varepsilon}{50}$. Recall from Algorithm \ref{algo:initialization} that $d$, $\hat{\vh}_0$, and $\hat{\vx}_0$ denote the highest singular value of $\setA^*(\hat{\vy})$, the corresponding left, and right singular vectors, respectively. Assuming without loss of generality that $d_0 = 1$ gives $|d-1| \leq \tfrac{3\varepsilon}{50}$. Use this together with $\varepsilon \leq \tfrac{1}{15}$ to conclude that $0.9 d_0 \leq d \leq 1.1d_0$. 

The initializer $\vv_0$ of $\vx_0$ computed by solving a minimization program in Algorithm \ref{algo:initialization} is basically a projection of $\sqrt{d} \vx_0$ onto the convex set $\setZ = \{\vz | \sqrt{QN} \|\mC^{\otimes N}\vz\|_\infty \leq 2 \sqrt{d} \nu \}$. Now $\vv_0 \in \setZ$ implies that $\sqrt{QN}\|\mC^{\otimes N}\vv_0\|_\infty \leq 2 \sqrt{d} \nu \leq \frac{4\nu}{\sqrt{3}}$, and hence $\vv_0 \in \frac{1}{\sqrt{3}} \setN_{\nu}$. In addition, we have 
\begin{align}\label{eq:projection-lemma-result}
&\|\sqrt{d}\hat{\vx}_0- \vp\|_2^2 = \notag\\
&\|\sqrt{d}\hat{\vx}_0 - \vv_0\|_2^2 + 2\Re{\<\hat{\vx}_0-\vv_0,\vv_0-\vp\>}+\|\vv_0-\vp\|_2^2\notag\\
&\geq \|\sqrt{d}\hat{\vx}_0 - \vv_0\|_2^2+\|\vv_0-\vp\|_2^2
\end{align}
for all $\vp \in \setZ$, where the last inequality is the result of Lemma \ref{lem:convex-set-projection} on the inner product. A specific choice of $\vp = \boldsymbol{0} \in \setZ$ in the above inequality gives $\|\vv_0\| \leq \sqrt{d} \leq \frac{2}{\sqrt{3}}$, and hence $\vv_0 \in \frac{1}{\sqrt{3}}\setN_{d_0}$. We have thus shown that $\vv_0 \in \frac{1}{\sqrt{3}}\setN_{d_0} \cap \frac{1}{\sqrt{3}}\setN_{\nu}$. In an exactly similar manner, we can show that $\vu_0 \in \frac{1}{\sqrt{3}}\setN_{d_0} \cap \frac{1}{\sqrt{3}}\setN_{\mu}$. 

It remains now to show that $(\vu_0,\vv_0) \in \setN_{\frac{2}{5}\varepsilon}$. Begin by noting that $\|\setA^*(\hat{\vy})-\vh_0\vx_0^*\|\leq \gamma$ implies that $\sigma_i(\setA^*(\hat{\vy})) \leq \gamma$ for $i \geq 2$, where $\sigma_i(\setA^*(\hat{\vy}))$ denotes the $i$th largest singular value of the matrix $\setA^*(\hat{\vy})$. This implies using triangle inequality, and \eqref{eq:A*(y)-h0x0-operator-norm} that 
\begin{align}\label{eq:distnace-initialization}
&\|d \hat{\vh}_0\hat{\vx}_0^* - \vh_0 \vx_0^*\|_{2 \rightarrow 2} \leq \|\setA^*(\hat{\vy}) - d \hat{\vh}_0\hat{\vx}_0^*\|_{2 \rightarrow 2} +\notag \\
& \qquad\qquad \|\setA^*(\hat{\vy}) - \vh_0\vx_0^*\|_{2 \rightarrow 2} \leq 2 \gamma,
\end{align}
where  $d$ is the highest singular value of $\setA^*(\vy)$, and $\hat{\vh}_0$ and $\hat{\vx}_0$ are the corresponding singular vectors already introduced in Algorithm \ref{algo:initialization}. We also have 
\begin{align*}
& \|\hat{\vx}_0^*(\mI-\vx_0\vx_0^*)\|_2 = \|(\hat{\vx}_0\hat{\vh}_0^* \hat{\vh}_0 \hat{\vx}_0^*)(\mI-\vx_0\vx_0^*) \|_F \\
& = \|\hat{\vx}_0 \hat{\vh}_0^* (\setA^*(\vy) - d\hat{\vh}_0\hat{\vx}_0^* + \hat{\vh}_0\hat{\vx}_0^* - \vh_0\vx_0^*) (\mI-\vx_0\vx_0^*)\|_F\\
& \leq \|\hat{\vx}_0\hat{\vh}_0^*(\setA^*(\vy)-\vh_0\vx_0^*)(\mI-\vx_0\vx_0^*)\|_F + |d-1| \leq 2\gamma,
\end{align*}
where the second equality follows from $\vh_0\vx_0^*(\mI-\vx_0\vx_0^*) = \mathbf{0}$, and $\hat{\vx}_0\hat{\vh}_0^*(\setA^*(\vy)-d\hat{\vh}_0\hat{\vx}_0^*) = \mathbf{0}$. Denoting $\beta_0 = \sqrt{d}\hat{\vx}_0^*\vx_0$, the above inequality can be equivalently written as 
\begin{align}\label{eq:initialization-ineq}
\|\sqrt{d} \hat{\vx}_0 - \beta_0 \vx_0\|_2 \leq 2 \sqrt{d} \gamma.
\end{align}
Observe that $\vp = \beta_0\vx_0 \in \setZ$, which follows from $\sqrt{QN}|\beta_0|\|\mC^{\otimes N} \vx_0\|_\infty = |\beta_0| \nu \leq \sqrt{d} \nu < 2\sqrt{d}\nu$. Therefore, using $\vp = \beta_0\vx_0 \in \setZ $ in \eqref{eq:projection-lemma-result} gives $\|\sqrt{d} \hat{\vx}_0 - \alpha_0\vx_0\|_2 \geq \|\vv_0 - \beta_0\vx_0\|_2$, which combined with \eqref{eq:initialization-ineq} yields 
\begin{align}\label{eq:initialization-ineq-final}
\|\vv_0-\beta_0\vx_0\|_2 \leq 2 \sqrt{d} \gamma. 
\end{align}
In an exactly similar manner, one can also show that 
\begin{align}
\|\vu_0 - \alpha_0 \vh_0\|_2 \leq 2 \sqrt{d} \gamma,
\end{align}
where $\alpha_0 = \sqrt{d} \vh_0^*\hat{\vh}_0$. Finally, 
\begin{align*}
&\|\vu_0\vv_0^* - \vh_0 \vx_0^*\|_F \leq \|\vu_0\vv_0^* - \alpha_0 \vu_0 \vx_0^*\| _F \\
& \qquad\qquad + \|\alpha_0 \vu_0 \vx_0^*  - \alpha_0\beta_0 \vh_0 \vx_0^*\|_F + \|\alpha_0\beta_0 \vh_0\vx_0^* - \vh_0\vx_0^*\|_F \\
&\leq \|\vu_0\|_2\|\vv_0-\alpha_0\vx_0\|_2 + |\alpha_0|\|\vx_0\|_2 \| \vu_0 - \beta_0 \vh_0\|_2 \\
&\qquad + \|d\vh_0\vh_0^*\hat{\vh}_0\hat{\vx}_0^* \vx_0\vx_0^*  - \vh_0 \vx_0^*\|_F\\
& = \|\vu_0\|_2 \|\vv_0-\alpha_0\vx_0\|_2 +  |\alpha_0| \| \vu_0 - \beta_0 \vh_0\|_2 \\ 
&\qquad\qquad  + \|d\hat{\vh}_0\hat{\vx}_0^* - \vh_0 \vx_0^*\|_F \\
& \leq \frac{2}{\sqrt{3}} \cdot 2\sqrt{d} \gamma + \sqrt{d} \cdot 2 \sqrt{d} \gamma + 2\gamma < \frac{20}{3} \gamma,
\end{align*}
which amounts to showing that $\|\vu_0\vv_0^*-\vh_0\vx_0^*\|_F \leq \frac{2}{5}\varepsilon$ using $\gamma$ defined in \eqref{eq:A*(y)-h0x0-operator-norm}, and $d_0 =1$ as before. This shows that $(\vu_0,\vv_0) \in \setN_{\frac{2}{5}\varepsilon}$. Plugging the choice $\xi = \tfrac{\varepsilon}{50}$ in $
1-2\exp\left(-c\xi^2QN/\mu^2\nu^2\right)$, computed above, and in \eqref{eq:sample-complexity-initialization} gives the claimed probability, and sample complexity bound, respectively. 
\end{proof}
\begin{lem}[Theorem 2.8 in \cite{escalante2011alternating}]\label{lem:convex-set-projection} Let $\setZ$ be a closed nonempty convex set. There holds
	\begin{align*}
	\Re{\<\vq-\setP_{\setZ}(\vq),\vz - \setP_{\setZ}(\vq)\>} \leq 0, \ \forall \vz \in \setZ, \vq \in \C^W,
	\end{align*}
	where $\setP_{\setZ}(\vq)$ is the projection of $\vq$ onto $\setZ$. 
\end{lem}


\end{document}